\documentclass{article}
\usepackage[a4paper,margin=1.5in]{geometry}
\usepackage[latin1]{inputenc}
\usepackage{latexsym}
\usepackage{amsmath}
\usepackage{amsthm}
\usepackage{amssymb}
\usepackage{graphicx}
\usepackage{ifthen}
\usepackage{calc}
\usepackage{stmaryrd}
\usepackage{enumitem}
\usepackage{hyperref}
\usepackage{alphalph}
\usepackage{tikz-cd}

\usepackage{tikz}
\usetikzlibrary{arrows,calc}
\usetikzlibrary{shapes.geometric}

\makeatletter
\renewcommand{\section}{\@startsection
  {section}%
  {1}%
  {0mm}%
  {-1\baselineskip}%
  {0.5\baselineskip}%
  {\normalfont\large\bfseries}%
}
\renewcommand{\subsection}{\@startsection
  {subsection}%
  {2}%
  {0mm}%
  {-1\baselineskip}%
  {0.5\baselineskip}%
  {\normalfont\large\itshape}%
}

\renewcommand{\subsubsection}{\@startsection
  {subsubsection}%
  {3}%
  {0mm}%
  {-1\baselineskip}%
  {0.5\baselineskip}%
  {\normalfont\itshape}%
}

\newsavebox{\tempbox}
\renewcommand{\@makecaption}[2]{
  \vspace{10pt}
  \sbox{\tempbox}{\textbf{#1.} #2}
  \ifthenelse{\lengthtest{\wd\tempbox > \linewidth}}{
    \textbf{#1.} #2\par
  }{
    \begin{center}
      \textbf{#1.} #2
    \end{center}
  }
}

\makeatother

\numberwithin{equation}{section}

\numberwithin{figure}{section}

\newtheoremstyle{mythm}%
  {}%
  {}%
  {\itshape}%
  {}%
  {\bfseries}%
  {.}%
  {.5em}%
  {\thmname{#1}~\thmnumber{#2}\ifthenelse{\equal{\thmnote{#3}}{}}{}{~(\thmnote{#3})}}%

\newtheoremstyle{mydefn}%
  {}%
  {}%
  {\upshape}%
  {}%
  {\bfseries}%
  {.}%
  {.5em}%
  {\thmname{#1}~\thmnumber{#2}\ifthenelse{\equal{\thmnote{#3}}{}}{}{~(\thmnote{#3})}}%

\newtheoremstyle{myremark}%
  {}%
  {}%
  {\upshape}%
  {}%
  {\itshape}%
  {.}%
  {.5em}%
  {\thmname{#1}~\thmnumber{#2}\ifthenelse{\equal{\thmnote{#3}}{}}{}{~(\thmnote{#3})}}%

\theoremstyle{mythm}
\newtheorem{theo}{Theorem}[section]
\newtheorem{lem}[theo]{Lemma}
\newtheorem{prop}[theo]{Proposition}
\newtheorem{cor}[theo]{Corollary}

\theoremstyle{mydefn}

\newtheorem{exa}[theo]{Example}

\newtheorem{ass}[theo]{Assumption}
\theoremstyle{myremark}
\newtheorem{rem}[theo]{Remark}
\theoremstyle{mythm}

\newcommand{\uend}{\hfill$\lrcorner$}

\newcounter{claimcounter}
\newenvironment{claim}[1][]{
  \renewcommand{\proof}{\smallskip\par\noindent\textit{Proof. }}
  \medskip\par\noindent%
  \ifthenelse{\equal{#1}{}}{%
    \setcounter{claimcounter}{0}\refstepcounter{claimcounter}\textit{Claim~\arabic{claimcounter}.}
  }{%
    \ifthenelse{\equal{#1}{resume}}{%
      \refstepcounter{claimcounter}\textit{Claim~\arabic{claimcounter}.}
    }{%
      \textit{Claim~#1.}
    }
  }
}{
  \par\medskip
}

\newcommand{\case}[1]{\par\medskip\noindent\textit{Case #1: }}

\newenvironment{cs}{
  \begin{description}
    \renewcommand{\case}[1]{\item[\itshape\mdseries Case ##1:]}
  }{
  \end{description}
}

\newlist{caselist}{description}{10}
\setlist[caselist]{font=\itshape\mdseries}

\setenumerate[1]{label=(\arabic*)}
\newlist{eroman}{enumerate}{2}
\setlist[eroman,1]{label=(\roman*)}
\setlist[eroman,2]{label=(\alph*)}
\newlist{ealph}{enumerate}{1}
\setlist[ealph]{label=(\Alph*)}

\newcounter{nlistcounter}
\newenvironment{nlist}[1]{
  \renewcommand{\thenlistcounter}{\upshape(#1.\arabic{nlistcounter})}
  \begin{list}{\bfseries\thenlistcounter}{%
      \usecounter{nlistcounter}
      \setlength{\labelwidth}{1.5em}%
      \setlength{\leftmargin}{\labelwidth}%
      \addtolength{\leftmargin}{\labelsep}%
      \setlength{\listparindent}{0em}%
      \setlength{\topsep}{5pt}%
      \setlength{\itemsep}{5pt}%
      \setlength{\parsep}{0pt}%
    }
  }{
  \end{list}
}

\definecolor{blau}{RGB}{0,84,159}
\definecolor{hellblau}{RGB}{142,168,229}
\definecolor{petrol}{RGB}{0,97,101}
\definecolor{tuerkis}{RGB}{0,152,161}
\definecolor{gruen}{RGB}{87,171,39}
\definecolor{maigruen}{RGB}{189,205,0}
\definecolor{gelb}{RGB}{255,237,0}
\definecolor{orange}{RGB}{255,128,0}
\definecolor{magenta}{RGB}{227,0,102}
\definecolor{rot}{RGB}{204,7,30}
\definecolor{bordeaux}{RGB}{161,16,53}
\definecolor{violett}{RGB}{97,33,88}
\definecolor{lila}{RGB}{122,111,172}
\definecolor{grey}{gray}{0.7}
\definecolor{mittelblau}{RGB}{0,128,255}
\definecolor{rosa}{RGB}{255,153,204}

\newcommand{\red}{\color{rot}}

\newcommand{\bigmid}{\;\big|\;}
\newcommand{\Bigmid}{\;\Big|\;}

\renewcommand{\mathbf}[1]{\textit{\bfseries #1}}

\renewcommand{\hat}{\widehat}

\newcommand{\angles}[1]{\left\langle#1\right\rangle}

\newcommand{\slashes}[1]{{\backslash#1/}}

\renewcommand{\phi}{\varphi}
\renewcommand{\epsilon}{\varepsilon}

\newcommand{\CM}{{\mathcal M}}
\newcommand{\CN}{{\mathcal N}}

\newcommand{\CP}{{\mathcal P}}

\newcommand{\CT}{{\mathcal T}}

\newcommand{\CY}{{\mathcal Y}}
\newcommand{\CZ}{{\mathcal Z}}

\newcounter{rbcounter}
\usepackage{algorithmic}

\algsetup{linenodelimiter=.}

\usepackage{float}
\newfloat{algorithm}{ht}{alg}
\floatname{algorithm}{Algorithm}

\newcommand{\ord}{\operatorname{ord}}

\newcommand{\torso}[2]{#1\llbracket#2\rrbracket}

\newcommand{\CMT}{\CT_{\min}}
\newcommand{\CNDT}{\CT_{\textup{nd}}}

\newcommand{\CMS}{{\mathcal{ST}}_{\min}}
\newcommand{\CNDS}{{\mathcal{ST}}_{\textup{nd}}}

\newcommand{\CHT}{\hat{\mathcal T}}

\newcommand{\EC}{E^\times}
\newcommand{\ECND}{\EC_{\textup{nd}}}

\newcommand{\THT}{\mathit{TH}_{+3}}
\newcommand{\THF}{\mathit{TH}_{+4}}
\newcommand{\TRT}{\mathit{TR}_{+3}}

\newcommand{\fc}{\operatorname{fc}}
\newcommand{\sep}{\operatorname{sep}}
\newcommand{\Sep}{\operatorname{Sep}}

\begin{document}
\title{Quasi-4-Connected Components}
\author{Martin Grohe\\\normalsize RWTH Aachen
  University\\\normalsize\texttt{grohe@informatik.rwth-aachen.de}}
\date{}
\maketitle

\begin{abstract}
  We introduce a new decomposition of a graphs into
  \emph{quasi-4-connected} components, where we call a graph
  \emph{quasi-4-connected} if it is 3-connected and it only has separations of order $3$ that
    remove a single vertex. Moreover, we give a cubic time algorithm
  computing the decomposition of a given graph.

  Our decomposition into quasi-4-connected components refines the
  well-known decompositions of graphs into biconnected and
  triconnected components. We relate our decomposition to Robertson
  and Seymour's theory of tangles by establishing a correspondence
  between the quasi-4-connected components of a graph and its tangles
  of order $4$.
\end{abstract}

\section{Introduction}

Decompositions of graphs into their connected, biconnected and
triconnected components are fundamental in structural graph theory,
and they also belong to the basic toolbox of algorithmic graph
theory. The existence of such decompositions goes back to work of
MacLane~\cite{mac37} from the 1930s (also see Tutte~\cite{tut84}). In
the 1970s, Hopcroft and Tarjan~\cite{hoptar73a,tar72} showed that the
decompositions can be computed in linear time.

In modern terms, the decompositions into biconnected and triconnected
components are best described as tree decompositions. To state the
decomposition theorems and also our main results, a few technical
definitions are unavoidable. Recall that a
\emph{tree decomposition} of a graph $G$ is a pair $(T,\beta)$, where
$T$ is a tree and $\beta$ a mapping that associates a set
$\beta(t)\subseteq  V(G)$, called the \emph{bag} at $t$, with every
node $t$ of the tree $T$. The \emph{adhesion} of the decomposition is
the maximum of the sizes $|\beta(t)\cap\beta(u)|$ for tree edges $tu$,
which intuitively is the order of the separations of the
decomposition. Now the decomposition into biconnected
components can be phrased as follows: every graph $G$ has a tree
decomposition $(T,\beta)$ of adhesion at most $1$ such that for all tree nodes $t$ the induced
subgraph $G[\beta(t)]$ is either 2-connected or a complete graph of
order at most $2$. %
The decomposition into triconnected components is more complicated,
mainly because the triconnected components of a graph are no longer
induced subgraphs, but just topological subgraphs. We say that the
\emph{torso} of a set $X\subseteq V(G)$ of vertices of a graph $G$ is
the graph $\torso GX$ obtained from the induced subgraph $G[X]$ by
adding edges $vw$ for all distinct $v,w\in X$ such that there is a
connected component $C$ of $G\setminus X$ with $v,w\in N(C)$, the
neighbourhood of $C$ in $G$. For example, the torso of the set
$X=\{x_1,\ldots,x_4\}$ in the graph $G$ shown in Figure~\ref{fig:3conn}(a)
is the complete graph on $X$. Now the decomposition into triconnected
components can be phrased as follows: every graph $G$ has a tree
decomposition $(T,\beta)$ of adhesion at most $2$ such that for all tree nodes $t$ the torso
$\torso G{\beta(t)}$ is a topological subgraph of $G$ that is either 3-connected or a complete graph of
order at most $3$.

\begin{figure}[h]
  \centering
  \begin{tikzpicture}
  [
  vertex/.style={draw,circle,fill=black,inner sep=0mm,minimum
    size=2mm},
  tn/.style={draw,inner sep=3pt},
  every edge/.style={draw,thick}
  ]

  \begin{scope}
    \path
    (0,0) node[vertex] (x1) {} node[below] {$x_1$}
    (90:1cm) node[vertex] (x2) {}  node[above] {$x_2$}
    (210:1cm) node[vertex] (x3) {}  node[left] {$x_3$}
    (330:1cm) node[vertex] (x4) {}  node[right] {$x_4$}
    (150:1.2cm) node[vertex] (y1) {}  node[above] {$y_1$}
    (270:1.2cm) node[vertex] (y2) {}  node[below] {$y_2$}
    (30:1.2cm) node[vertex] (y3) {}  node[above] {$y_3$}
  ;

    \path 
    (x1) edge (x2) edge (x3) edge (x4)
    (y1) edge (x2) edge (x3) 
    (y2) edge (x3) edge (x4) 
    (y3) edge (x4) edge (x2) 
    ;

    \path (0,-2) node {(a)};
    \end{scope}

    \begin{scope}[xshift=6cm]
      \path
      (0,-0.1) node[tn] (s) {$\{x_1,x_2,x_3,x_4\}$}
      (-1.6,1) node[tn] (t1) {$\{x_2,x_3,y_1\}$}
      (0,-1.2) node[tn] (t2) {$\{x_3,x_4,y_2\}$}
      (1.6,1) node[tn] (t3) {$\{x_4,x_1,y_3\}$}
      ;

      \path (s) edge (t1) edge (t2) edge (t3);

    \path (0,-2) node {(b)};
    \end{scope}
    
\end{tikzpicture}
  \caption{A graph and its decomposition into triconnected components}
  \label{fig:3conn}
\end{figure}
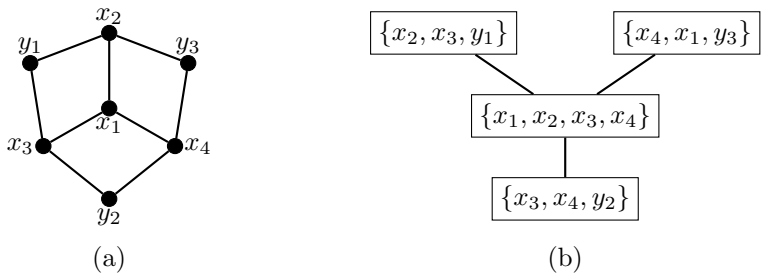

How about decompositions into 4-connected components, or
$k$-connected components for $k\ge4$? At least in the clean form of the
above decomposition theorems, they simply do not exist. Consider, for
example, a hexagonal grid (see Figure~\ref{fig:hexgrids}). Even though
the grid is not 4-connected, and it does not even have a nontrivial 4-connected
subgraph, there is no good way of decomposing it in a tree like
fashion by separations of order $3$. However, the only separations of the
grid of order $3$ are those splitting off a single vertex. If we ignore
such separations, we may view the whole grid as one highly connected
region. Let us call a graph $G$ \emph{quasi-4-connected} if it is
3-connected and for all separations $(Y,S,Z)$ of order $3$
(that is, $|S|=3$ and $Y,S,Z$ form a partition of $V(G)$ and there are no edges
between $Y$ and $Z$), either $|Y|\le 1$ or $|Z|\le 1$. Surprisingly,
with this mild relaxation of 4-connectivity we get a nice
decomposition theorem along the lines of the decompositions into
biconnected and triconnected components.

\begin{figure}
  \centering
  \begin{tikzpicture}
  [
  current point is local=true,
  line width=0.3mm,
  every node/.style={draw,circle,fill=black,inner sep=0mm,minimum
     size=1.2mm},
  every edge/.style={draw},
  scale=0.7,
  ]

  \newcommand{\grd}{1.0392cm}

  \draw ++(90:6mm) {
    +(90:6mm) node {} -- 
    +(150:6mm) node {} -- 
    +(210:6mm) node {} -- 
    +(270:6mm) node {} -- 
    +(330:6mm) node {} -- 
    +(30:6mm) node {} -- 
    cycle
  }
  ++(240:\grd) {
    +(90:6mm) node {} -- 
    +(150:6mm) node {} -- 
    +(210:6mm) node {} -- 
    +(270:6mm) node {} -- 
    +(330:6mm) node {} -- 
    +(30:6mm) node {} -- 
    cycle
  }
  ++(0:\grd) {
    +(90:6mm) node {} -- 
    +(150:6mm) node {} -- 
    +(210:6mm) node {} -- 
    +(270:6mm) node {} -- 
    +(330:6mm) node {} -- 
    +(30:6mm) node {} -- 
    cycle
  }
  ++(60:\grd) { [rounded corners]
    +(90:6mm) node {} -- 
    +(150:6mm) node {} -- 
    +(210:6mm) node {} -- 
    +(270:6mm) node {} -- 
    +(330:6mm) node {} -- 
    +(30:6mm) -- 
    cycle
  }
  ++(120:\grd) { [rounded corners]
    +(90:6mm) --
    +(150:6mm) node {} -- 
    +(210:6mm) node {} -- 
    +(270:6mm) node {} -- 
    +(330:6mm) node {} -- 
    +(30:6mm) --
    cycle
  }
  ++(180:\grd) { [rounded corners]
    +(90:6mm) -- 
    +(150:6mm) -- 
    +(210:6mm) node {} -- 
    +(270:6mm) node {} -- 
    +(330:6mm) node {} -- 
    +(30:6mm) node {} -- 
    cycle
  }
  ++(240:\grd) { [rounded corners]
    +(90:6mm) node {} -- 
    +(150:6mm) -- 
    +(210:6mm) node {} -- 
    +(270:6mm) node {} -- 
    +(330:6mm) node {} -- 
    +(30:6mm) node {} -- 
    cycle
  }
  ++(240:\grd) { [rounded corners]
    +(90:6mm) node {} -- 
    +(150:6mm) -- 
    +(210:6mm) -- 
    +(270:6mm) node {} -- 
    +(330:6mm) node {} -- 
    +(30:6mm) node {} -- 
    cycle
  }
  ++(300:\grd) { [rounded corners]
    +(90:6mm) node {} -- 
    +(150:6mm) node {} -- 
    +(210:6mm) -- 
    +(270:6mm) -- 
    +(330:6mm) node {} -- 
    +(30:6mm) node {} -- 
    cycle
  }
  ++(0:\grd) { [rounded corners]
    +(90:6mm) node {} -- 
    +(150:6mm) node {} -- 
    +(210:6mm) node {} -- 
    +(270:6mm) -- 
    +(330:6mm) node {} -- 
    +(30:6mm) node {} -- 
    cycle
  }
  ++(0:\grd) { [rounded corners]
    +(90:6mm) node {} -- 
    +(150:6mm) node {} -- 
    +(210:6mm) node {} -- 
    +(270:6mm) -- 
    +(330:6mm) -- 
    +(30:6mm) node {} -- 
    cycle
  }
  ++(60:\grd) { [rounded corners]
    +(90:6mm) node {} -- 
    +(150:6mm) node {} -- 
    +(210:6mm) node {} -- 
    +(270:6mm) node {} -- 
    +(330:6mm) -- 
    +(30:6mm) -- 
    cycle
  }
  ;
 
\end{tikzpicture}
\hspace{3cm}
  \begin{tikzpicture}
  [
  current point is local=true,
  line width=0.3mm,
  every node/.style={draw,circle,fill=black,inner sep=0mm,minimum
     size=1.2mm},
  every edge/.style={draw},
  scale=0.7,
  ]

  \newcommand{\grd}{1.0392cm}

  \draw ++(90:6mm) {
    +(90:6mm) node {} -- 
    +(150:6mm) node {} -- 
    +(210:6mm) node {} -- 
    +(270:6mm) node {} -- 
    +(330:6mm) node {} -- 
    +(30:6mm) node {} -- 
    cycle
  }
  ++(240:\grd) {
    +(90:6mm) node {} -- 
    +(150:6mm) node {} -- 
    +(210:6mm) node {} -- 
    +(270:6mm) node {} -- 
    +(330:6mm) node {} -- 
    +(30:6mm) node {} -- 
    cycle
  }
  ++(0:\grd) {
    +(90:6mm) node {} -- 
    +(150:6mm) node {} -- 
    +(210:6mm) node {} -- 
    +(270:6mm) node {} -- 
    +(330:6mm) node {} -- 
    +(30:6mm) node {} -- 
    cycle
  }
  ++(60:\grd) {
    +(90:6mm) node {} -- 
    +(150:6mm) node {} -- 
    +(210:6mm) node {} -- 
    +(270:6mm) node {} -- 
    +(330:6mm) node {} -- 
    +(30:6mm) node {} -- 
    cycle
  }
  ++(120:\grd) {
    +(90:6mm) node {} -- 
    +(150:6mm) node {} -- 
    +(210:6mm) node {} -- 
    +(270:6mm) node {} -- 
    +(330:6mm) node {} -- 
    +(30:6mm) node {} -- 
    cycle
  }
  ++(180:\grd) {
    +(90:6mm) node {} -- 
    +(150:6mm) node {} -- 
    +(210:6mm) node {} -- 
    +(270:6mm) node {} -- 
    +(330:6mm) node {} -- 
    +(30:6mm) node {} -- 
    cycle
  }
  ++(240:\grd) {
    +(90:6mm) node {} -- 
    +(150:6mm) node {} -- 
    +(210:6mm) node {} -- 
    +(270:6mm) node {} -- 
    +(330:6mm) node {} -- 
    +(30:6mm) node {} -- 
    cycle
  }
  ++(240:\grd) {
    +(90:6mm) node {} -- 
    +(150:6mm) node {} -- 
    +(210:6mm) node {} -- 
    +(270:6mm) node {} -- 
    +(330:6mm) node {} -- 
    +(30:6mm) node {} -- 
    cycle
  }
  ++(300:\grd) {
    +(90:6mm) node {} -- 
    +(150:6mm) node {} -- 
    +(210:6mm) node {} -- 
    +(270:6mm) node {} -- 
    +(330:6mm) node {} -- 
    +(30:6mm) node {} -- 
    cycle
  }
  ++(0:\grd) {
    +(90:6mm) node {} -- 
    +(150:6mm) node {} -- 
    +(210:6mm) node {} -- 
    +(270:6mm) node {} -- 
    +(330:6mm) node {} -- 
    +(30:6mm) node {} -- 
    cycle
  }
  ++(0:\grd) {
    +(90:6mm) node {} -- 
    +(150:6mm) node {} -- 
    +(210:6mm) node {} -- 
    +(270:6mm) node {} -- 
    +(330:6mm) node {} -- 
    +(30:6mm) node {} -- 
    cycle
  }
  ++(60:\grd) {
    +(90:6mm) node {} -- 
    +(150:6mm) node {} -- 
    +(210:6mm) node {} -- 
    +(270:6mm) node {} -- 
    +(330:6mm) node {} -- 
    +(30:6mm) node {} -- 
    cycle
  }
  ++(60:\grd) {  [rounded corners] 
    +(90:6mm) node {} -- 
    +(150:6mm) node {} -- 
    +(210:6mm) node {} -- 
    +(270:6mm) node {} -- 
    +(330:6mm) node {} -- 
    +(30:6mm) -- 
    cycle
  }
  ++(120:\grd) {  [rounded corners] 
    +(90:6mm) node {} -- 
    +(150:6mm) node {} -- 
    +(210:6mm) node {} -- 
    +(270:6mm) node {} -- 
    +(330:6mm) node {} -- 
    +(30:6mm) -- 
    cycle
  }
  ++(120:\grd) {  [rounded corners] 
    +(90:6mm) -- 
    +(150:6mm) node {} -- 
    +(210:6mm) node {} -- 
    +(270:6mm) node {} -- 
    +(330:6mm) node {} -- 
    +(30:6mm) -- 
    cycle
  }
  ++(180:\grd) {  [rounded corners] 
    +(90:6mm) -- 
    +(150:6mm) node {} -- 
    +(210:6mm) node {} -- 
    +(270:6mm) node {} -- 
    +(330:6mm) node {} -- 
    +(30:6mm) node {} -- 
    cycle
  }
  ++(180:\grd) {  [rounded corners] 
    +(90:6mm) -- 
    +(150:6mm) -- 
    +(210:6mm) node {} -- 
    +(270:6mm) node {} -- 
    +(330:6mm) node {} -- 
    +(30:6mm) node {} -- 
    cycle
  }
  ++(240:\grd) {  [rounded corners] 
    +(90:6mm) node {} -- 
    +(150:6mm) -- 
    +(210:6mm) node {} -- 
    +(270:6mm) node {} -- 
    +(330:6mm) node {} -- 
    +(30:6mm) node {} -- 
    cycle
  }
  ++(240:\grd) {  [rounded corners] 
    +(90:6mm) node {} -- 
    +(150:6mm) -- 
    +(210:6mm) node {} -- 
    +(270:6mm) node {} -- 
    +(330:6mm) node {} -- 
    +(30:6mm) node {} -- 
    cycle
  }
  ++(240:\grd) {  [rounded corners] 
    +(90:6mm) node {} -- 
    +(150:6mm) -- 
    +(210:6mm) -- 
    +(270:6mm) node {} -- 
    +(330:6mm) node {} -- 
    +(30:6mm) node {} -- 
    cycle
  }
  ++(300:\grd) {  [rounded corners] 
    +(90:6mm) node {} -- 
    +(150:6mm) node {} -- 
    +(210:6mm) -- 
    +(270:6mm) node {} -- 
    +(330:6mm) node {} -- 
    +(30:6mm) node {} -- 
    cycle
  }
  ++(300:\grd) {  [rounded corners] 
    +(90:6mm) node {} -- 
    +(150:6mm) node {} -- 
    +(210:6mm) -- 
    +(270:6mm) -- 
    +(330:6mm) node {} -- 
    +(30:6mm) node {} -- 
    cycle
  }
  ++(0:\grd) {  [rounded corners]
    +(90:6mm) node {} -- 
    +(150:6mm) node {} -- 
    +(210:6mm) node {} -- 
    +(270:6mm) -- 
    +(330:6mm) node {} -- 
    +(30:6mm) node {} -- 
    cycle
  }
  ++(0:\grd) {  [rounded corners] 
    +(90:6mm) node {} -- 
    +(150:6mm) node {} -- 
    +(210:6mm) node {} -- 
    +(270:6mm) -- 
    +(330:6mm) node {} -- 
    +(30:6mm) node {} -- 
    cycle
  }
  ++(0:\grd) {  [rounded corners] 
    +(90:6mm) node {} -- 
    +(150:6mm) node {} -- 
    +(210:6mm) node {} -- 
    +(270:6mm) -- 
    +(330:6mm) -- 
    +(30:6mm) node {} -- 
    cycle
  }
  ++(60:\grd) {  [rounded corners] 
    +(90:6mm) node {} -- 
    +(150:6mm) node {} -- 
    +(210:6mm) node {} -- 
    +(270:6mm) node {} -- 
    +(330:6mm) -- 
    +(30:6mm) node {} -- 
    cycle
  }
  ++(60:\grd) {  [rounded corners] 
    +(90:6mm) node {} -- 
    +(150:6mm) node {} -- 
    +(210:6mm) node {} -- 
    +(270:6mm) node {} -- 
    +(330:6mm) -- 
    +(30:6mm) -- 
    cycle
  }
  ;
 
\end{tikzpicture}
  \caption{Hexagonal grids of radius 2 and 3}
  \label{fig:hexgrids}
\end{figure}
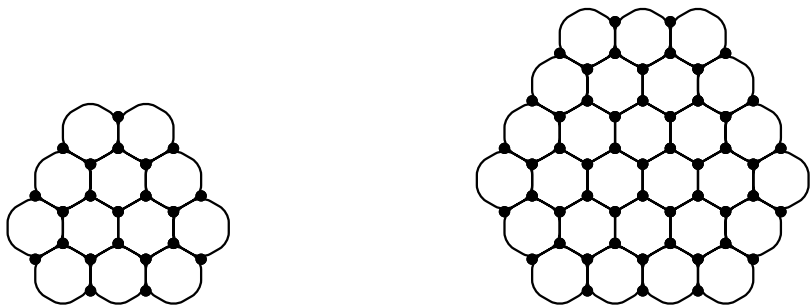

\begin{theo}[Decomposition Theorem]\label{theo:dec}
  Every graph $G$ has a tree decomposition $(T,\beta)$ of adhesion at
  most $3$ such that for all tree nodes $t$ the torso
  $\torso G{\beta(t)}$ is a minor of $G$ that is either
  quasi-4-connected or a complete graph of order at most $4$.

  Furthermore, this decomposition can be computed in cubic time.
\end{theo}

There have been earlier attempts to generalise the decomposition of
graphs into triconnected components. The most prominent of these are
Robertson and Seymour's tangles~\cite{gm10}, which play an important
role in the structure theory for graphs with excluded
minors~\cite{gm16}. Intuitively, a tangle of order $k$ describes a
``$k$-connected region'' in a graph by ``pointing to it'', that is,
by assigning a direction to each separation of order less than $k$
in such a way that ``most'' of the region described by the tangle is
on the side the separation is directed towards. It is known that the tangles of
orders 1, 2, 3 are in one-to-one correspondence to the connected,
biconnected and triconnected of a graph \cite{gm10,gro16a}. We
establish a similar correspondence between the tangles of order $4$
and the quasi-4-connected components. This is our second main theorem,
which I think is interesting in its own right, but is also essential
for the proof of Theorem~\ref{theo:dec}. We defer the precise technical
statement of this Correspondence Theorem to the main part of the
paper (Theorem~\ref{theo:corr}). %

This paper grew out of my work on descriptive complexity theory for
graph classes with excluded minors~\cite{gro12,gro12+a}, and this may
also serve as an illustration of potential applications of our
Decomposition Theorem. Separations of order 3 play a special, but
somewhat annoying role in the main structure theorems for graph
classes with excluded minors such as the ``Flat Grid Theorem'' of
\cite{gm13} and the structure theorem of \cite{gm16}, and the theorems
simplify for quasi-4-connected graphs. In \cite{gro12+a}  I exploited some of the main ideas underlying
our Decomposition Theorem to obtain such simplifications in the
context of logical definability, and I believe the Decomposition
Theorem proved here may turn out to be similarly useful in an
algorithmic context.\footnote{Let
  me clarify the relation of this work to Chapter~10 of the
  forthcoming monograph~\cite{gro12+a}. The basic ideas are the same,
  and actually my original motivation for the present paper was to
  make these ideas accessible to readers not interested in logic. However, only when I started to work on this paper I noticed the
  connection to tangles, and it is this connection that provides the
  right framework and also makes the decomposition much
  simpler. On the other hand, the main goal of \cite{gro12+a} is to
  obtain a decomposition that is definable in fixed-point logic with
  counting, and the decomposition we obtain here is not. So, except
  for some of the basic lemmas in Section~\ref{sec:4t1}, the results
  are incomparable.}

\subsection{Related work}
 It was shown in \cite{gm10,cardiehar+13a} that for every $k$, every
  graph admits a canonical decomposition into its tangles of order
  $k$. Related to this is the decomposition into
  so-called $(k-1)$-blocks due to \cite{cardiehun+14}. These
  decompositions (for $k=4$) are related to ours. An important
  difference between these results and ours, or rather an additional
  feature of our decomposition, is that the pieces of our
  decomposition are quasi-4-connected graphs in their own right and
  can be dealt with separately (for example in an algorithmic
  context), whereas tangles of order $4$ or $3$-blocks are only
  defined within the surrounding graph. 

  On the algorithmic side, it was shown in \cite{groschwe15a} that the
  decomposition  into its tangles of order
  $k$ can be computed in time $n^{O(k)}$. I believe that our
  techniques can be used to improve this to cubic time for $k=4$.

  There is a different line of work on ``$k$-connected components''
  that, as far as I can see, is completely unrelated to ours. There,
  $k$-connected components are simply defined as maximal $k$-connected
  subgraphs (see, for example, \cite{mak88,nagwat93,nagiba08}). This
  leads to completely different decompositions. For example, a graph
  of maximum degree $3$ will only have trivial 4-connected components
  in this framework. However, what I see as the
  crucial difference between our form of decomposition and this line
  of work is that we get tree decompositions into independent parts with a small
  interface (technically, small adhesion). This is important for typical dynamic-programming or
  divide-and-conquer algorithms on the decomposition.

\section{Preliminaries}
\label{sec:prel}
We assume basic knowledge of graph theory and refer the reader to
\cite{die05} for background. Our notation is standard, let us just
review the most important and frequently used notations. All graphs
considered in this paper are finite and simple. The vertex set
and edge set of a graph $G$ are denoted by $V(G)$ and $E(G)$,
respectively. The \emph{order} of $G$ is $|G|:=|V(G)|$. For a set
$W\subseteq V(G)$, we denote the induced subgraph of $G$ with vertex
set $W$ by $G[W]$ and the induced subgraph with vertex set
$V(G)\setminus W$ by $G\setminus W$. 
For a vertex $v$, we denote the
set of neighbours of $v$ in $G$ by $N^G(v)$. In this and similar
notations, we omit the index~${}^G$ if $G$ is clear from the
context. For a set $W\subseteq V(G)$, we define 
$
N^G(W):=\Big(\bigcup_{v\in W}N^G(v)\Big)\setminus W,
$
and for a subgraph $H\subseteq G$ we let $N^G(H):=N^G(V(H))$. %

A minor of $G$ is a graph obtained from $G$ by deleting vertices
and edge and contracting edges. An \emph{model} of $H$ in $G$ consists of a family $(M_w)_{w\in
  V(H)}$ of mutually disjoint connected subsets of $V(G)$ and a family
$(e_f)_{f\in E(H)}$ of edges of $G$ such that for every edge $f=ww'$ of $H$ the edge
$e_f$ has one endvertex in $M_w$ and one endvertex in
$M_{w'}$. Then $H$ is a \emph{minor} of $G$ if and only if there is
a model of $H$ in $G$. We call the sets $M_w$, for $w\in V(H)$, the
\emph{branch sets} of the model $\CM$. When reasoning about a model, it is
often enough to know the branch sets. 

A \emph{faithful model} of $H$ in $G$ is a model $\big((M_w)_{w\in
  V(H)},(e_f)_{f\in E(H)}\big)$ such that $w\in M_w$ for all $w\in
V(H)$. We say that $H$ is a \emph{faithful minor} of $G$ if
$V(H)\subseteq V(G)$ and there is a faithful model of $H$ in $G$.

Separations of a graph $G$ are usually defined as pairs of subgraphs
(see the appendix). However, in this paper it will be
more convenient to view them as 
partitions of the vertex set. We say
that a \emph{separation} of $G$ is a triple $(Y,S,Z)$ of
(possibly empty) mutually disjoint subsets of $V(G)$ such that $Y\cup S\cup Z=V(G)$ and
there is no edge $vw\in E(G)$ such that $v\in Y$ and $z\in Z$. The
order of the separation $(Y,S,Z)$ is $|S|$, and the
separation is \emph{proper} if both $Y$ and $Z$ are nonempty. 
The set of all separations of $G$ is
denoted by $\Sep(G)$, and the subset of all separations of
order less than $k$ (at most $k$, exactly $k$) by $\Sep_{<k}(G)$
(resp.~$\Sep_{\le k}(G)$, $\Sep_{= k}(G)$). 

A set $S\subseteq V(G)$ is a \emph{separator} of $G$ of \emph{order}
$k:=|S|$, or a \emph{$k$-separator}, if there are two vertices $v,w\in
V(G)\setminus S$ such that there is a path from $v$ to $w$ in $G$, but
no path from $v$ to $w$ in $G\setminus S$. Note that if $G$ is connected then $S$ is
a separator if and only if there is a proper separation
$(Y,S,Z)$ of $G$.

A graph $G$ is \emph{$k$-connected} if $|G|>k$ and $G$ has no proper
$(k-1)$-separation. 
 
A subset $X\subseteq V(G)$ of the vertex set of a graph $G$ is
\emph{$k$-inseparable} if $|X|>k$ and there is no separation $(Y,S,Z)$
of $G$ of order at most $k$ such that $X\cap Y\neq\emptyset$ and
$X\cap Z\neq\emptyset$.

\section{Tangles}

Let $G$ be a graph. Deviating from Robertson and Seymour's \cite{gm10} original
definition, we define tangles as families of separations of the vertex
set (as we defined them in Section~\ref{sec:prel}) rather
than separations viewed pairs of graphs or partitions of the edge
set. (In the appendix, we show that the two notions are
equivalent.)
A \emph{$G$-tangle} of order $k$ is a family $\CT\subseteq\Sep_{<k}(G)$ of separations of
$G$ of order less than $k$ satisfying the following conditions.
\begin{nlist}{T}
\item\label{li:t1}
  For all separations $(Y,S,Z)\in\Sep_{<k}(G)$ either
  $(Y,S,Z)\in\CT$ or $(Z,Y,S)\in\CT$.
\item\label{li:t2}
  If $(Y_1,S_1,Z_1), (Y_2,S_2,Z_2), (Y_3,S_3,Z_3)\in\CT$ then either
  $Z_1\cap Z_2\cap Z_3\neq\emptyset$ or there is an edge $e\in E(G)$
  that has an endvertex in each $Z_i$.
\item\label{li:t3}
  $Z\neq\emptyset$ for all $(Y,S,Z)\in\CT$.
\end{nlist}
In the following, we collect a few basic facts about tangles. For more
background and examples, I refer the reader to
\cite{gm10,gro16a}. 

\subsection{Basic Facts}
For $(Y,S,Z),(Y',S',Z')\in \Sep(G)$, we let 
\begin{align*}
  (Y,S,Z)\cap(Y',S',Z')&:=(Y\cup Y',(S\cap Z')\cup (S\cap
                           S')\cup(Z\cap S'),Z\cap Z'),\\
  (Y,S,Z)\cup(Y',S',Z')&:=(Y\cap Y',(S\cap Y')\cup (S\cap
                           S')\cup(Y\cap S'),Z\cup Z')
\end{align*}
(see Figure~\ref{fig:int} for an illustration). Note that both $ (Y,S,Z)\cap(Y',S',Z')$ and $(Y,S,Z)\cup(Y',S',Z')$
are separations of $G$.

\begin{figure}
  \centering
  \begin{tikzpicture}
  \begin{scope}
        \fill[black!10] (-1.75,-1.75) rectangle (-0.25,1.75);
        \fill[black!10] (-1.75,-0.25) rectangle (1.75,1.75);
        \fill[black!20] (-0.25,-1.75) rectangle (0.25,1.75);
        \fill[black!20] (-1.75,-0.25) rectangle (1.75,0.25);
        
      \draw[thick] (-1.75,-1.75) rectangle (1.75,1.75)
                   (-0.25,-1.75) rectangle (0.25,1.75) 
                   (-1.75,-0.25) rectangle (1.75,0.25);
      
      \path (-2,1) node {$Z$}             
            (-2,0) node {$S$}        
            (-2,-1) node {$Y$}  
            (-1,2) node {$Z'$}
            (0,2) node {$S'$}
            (1,2) node {$Y'$}
      ;

      \path (0,-2.3) node {(a) $(Y,S,Z)$ and $(Y',S',Z')$};
      \end{scope}

   \begin{scope}[xshift=4.5cm]
        \fill[black!10] (-1.75,0.25) rectangle (-0.25,1.75);
        \fill[black!20] (-0.25,0.25) rectangle (0.25,1.75);
        \fill[black!20] (-1.75,-0.25) rectangle (0.25,0.25);
        
      \draw[thick] (-1.75,-1.75) rectangle (1.75,1.75)
                   (-0.25,-1.75) rectangle (0.25,1.75) 
                   (-1.75,-0.25) rectangle (1.75,0.25);
      
      \path (-2,1) node {$Z$}             
            (-2,0) node {$S$}        
            (-2,-1) node {$Y$}  
            (-1,2) node {$Z'$}
            (0,2) node {$S'$}
            (1,2) node {$Y'$}
      ;

      \path (0,-2.3) node {(b) $(Y,S,Z)\cap(Y',S',Z')$};
      \end{scope}

  \begin{scope}[xshift=9cm]
        \fill[black!10] (-1.75,-1.75) rectangle (-0.25,1.75);
        \fill[black!10] (-1.75,-0.25) rectangle (1.75,1.75);
        \fill[black!20] (-0.25,-1.75) rectangle (0.25,0.25);
        \fill[black!20] (0.25,-0.25) rectangle (1.75,0.25);
         
      \draw[thick] (-1.75,-1.75) rectangle (1.75,1.75)
                   (-0.25,-1.75) rectangle (0.25,1.75) 
                   (-1.75,-0.25) rectangle (1.75,0.25);
      
      \path (-2,1) node {$Z$}             
            (-2,0) node {$S$}        
            (-2,-1) node {$Y$}  
            (-1,2) node {$Z'$}
            (0,2) node {$S'$}
            (1,2) node {$Y'$}
      ;

      \path (0,-2.3) node {(c) $(Y,S,Z)\cup(Y',S',Z')$};
      \end{scope}

\end{tikzpicture}
  \caption{Intersection and union of separations}
  \label{fig:int}
\end{figure}
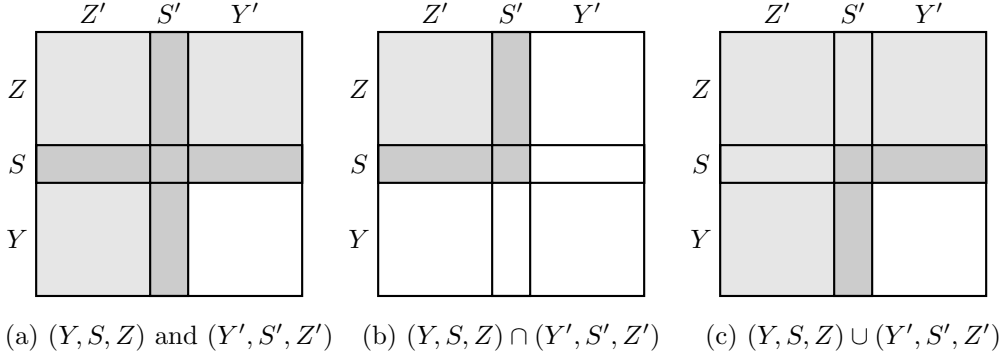

\begin{lem}[\cite{gm10}]\label{lem:tangle-closure}
  Let $G$ be a graph and $\CT$ a $G$-tangle of order $k$.
  \begin{enumerate}
  \item If $(X,Y,Z)\in\Sep(G)$ with $|Y\cup S|<k$ then $(Y,S,Z)\in\CT$. 
  \item  If $(Y,S,Z)\in\CT$ and $(Y',S',Z')\in\Sep_{<k}(G)$ such that
    $Z\subseteq Z'$ then $(Y',S',Z')\in\CT$.
  \item  If $(Y,S,Z),(Y',S',Z')\in\CT$ such that $(Y,S,Z)\cap
    (Y',S',Z')\in\Sep_{<k}(G)$
    then $(Y,S,Z)\cap (Y',S',Z') \in\CT$.
  \end{enumerate}
\end{lem}

\begin{cor}\label{cor:tangle-closure}
  Let $G$ be a graph and $\CT$ a $G$-tangle of order $k$. Let
  $(Y,S,Z),(Y',S',Z')\in\CT$. Then $|(S\cup Z)\cap(S'\cup Z')|\ge k$.
\end{cor}

The following lemma slightly strengthens
Lemma~\ref{lem:tangle-closure}(1). 

\begin{lem}\label{lem:t2}
  Let $G$ be a graph and $\CT$ a $G$-tangle of order $k$. Then for all
  $(Y,S,Z)\in\Sep_{<k}(G)$, if $|Y\cup S|\le
  \frac{3}{2}\cdot(k-1)$ then $(Y,S,Z)\in\CT$.
\end{lem}

\begin{proof}
  Let $(Y,S,Z)\in\Sep_{<k}(G)$ such that $|Y\cup S|\le
  \frac{3}{2}\cdot(k-1)$. By Lemma~\ref{lem:tangle-closure}(1)
  we may assume that $|Y\cup S|\ge k$. 
  Let $S_1\subseteq Y\cup S$ such that $S\subseteq S_1$ and
  $|S_1|=k-1$, and let
  $Y_1:=Y\setminus S_1$. Then it suffices to prove
  $(Y_1,S_1,Z)\in\CT$, because this implies $(Y,S,Z)\in\CT$ by
  Lemma~\ref{lem:tangle-closure}(2).

  As $|Y\cup S|\le \frac{3}{2}\cdot(k-1)$, we can choose subsets
  $S_2,S_3\subseteq Y\cup S$ of cardinality $|S_i|=k-1$ such that for
  all $x,x'\in Y\cup S$ (not necessarily distinct) there is an $i$
  such that $x,x'\in S_i$.  Note that $Y_1\subseteq S_2\cup S_3$.  By
  Lemma~\ref{lem:tangle-closure}(1), we have
  $(\emptyset,S_i,V(G)\setminus S_i)\in \CT$.

  Suppose for contradiction that $(Z,S_1,Y_1)\in\CT$. We have
  $Y_1\cap (V(G)\setminus S_2)\cap (V(G)\setminus S_3)=Y_1\setminus
  (S_2\cup S_3)=\emptyset$.
  Furthermore, let $e=xx'\in E(G)$. If $e$ has an endvertex in $Y_1$
  then $x,x'\in Y_1\cup S_1=Y\cup S$ and not $x,x'\in S_1$. Thus either
  $x,x'\in S_2$ or $x,x'\in S_3$, and either $e$ has no endvertex in
  $V(G)\setminus S_2$ or no endvertex in $V(G)\setminus S_3$. This
  contradicts \ref{li:t2}.
\end{proof}

The next lemma shows that highly connected sets within a graph induce
tangles. For a set $X\subseteq V(G)$ and $k\ge 1$, we let 
\[
\CT^k(X):=\big\{(Y,S,Z)\in\Sep_{<k}(G)\bigmid X\subseteq S\cup
Z\big\}.
\]
Of course in general, $\CT^k(X)$ is not a tangle, and neither are all
$G$-tangles of order $k$ of the form $\CT^k(X)$. However, we will see
in Section~\ref{sec:t3} that they are if $k\le 3$.

\begin{lem}\label{lem:t1}
  Let $G$ be a graph and $k\ge 1$. Let $X\subseteq V(G)$ be a $(k-1)$-inseparable set of
  cardinality $|X|>\frac{3}{2}\cdot (k-1)$. Then $\CT^k(X)$ is a
  $G$-tangle of order $k$.
\end{lem}

\begin{proof}
  To see that $\CT:=\CT^k(X)$ satisfies \ref{li:t1}, note that the
  $(k-1)$-inseparability of $X$ implies $X\subseteq S\cup Y$ or
  $X\subseteq S\cup Z$ for every $(Y,S,Z)\in\Sep_{<k}(G)$.

  To see that $\CT$ satisfies \ref{li:t2}, let
  $(Y_i,S_i,Z_i)\in\CT$ for $i=1,2,3$. We have $X\cap Y_i=\emptyset$
  and thus $|X\setminus Z_i|\le|S_i|\le k-1$. 
  As $|X|>\frac{3}{2}\cdot (k-1)$, there is a vertex $x\in X$
  such that $x$ is contained in at most one of the sets $S_i$ and
  hence in at least two of the sets $Z_j$. Say,
  $x\in Z_2\cap Z_3$. If $x\in Z_1$, then
  $Z_1\cap Z_2\cap Z_3\neq\emptyset$. So let us assume that
  $x\in S_1$.  

  Let $x_1,\ldots,x_{k-1}\in X\setminus\{x\}$ be distinct. Such $x_i$
  exists because $|X|\ge\frac{3}{2}(k-1)+1\ge k$. As $X$ is
  $(k-1)$-inseparable, for all $i$ there is a path $P_i$ from $x$ to $x_i$
  such that $V(P_i)\cap V(P_j)=\{x\}$ for $i\neq j$. Let
  $y_i$ be the last vertex of $P_i$ (in the direction from $x$ to
  $x_i$) that is in $S_1\cup Y_1$ (possibly, $y_i=x_i$). We claim that
  $y_i\in S_1$. This is the case if $y_i=x_i\in  X\subseteq S_1\cup Z_1$. If $y_i\neq x_i$, let $z_i$ be the successor of $y_i$ on
  $P_i$. Then $z_i\in Z_1$, and as $y_iz_i\in
  E(G)$, it follows that $y_i\in S_1$.

  Thus $x,y_1,\ldots,y_{k-1}\in S_1$, and as $|S_1|\le k-1$ and the
  $y_i$ are mutually distinct, it follows that $y_i=x$ for some
  $i$. As $x\neq x_i$, the vertex $z_i$ exists. The edge $xz_i$ has
  endvertices $z_i$ in $Z_1$ and $x$ in $Z_2$ and $Z_3$.

  Finally, $\CT$ satisfies
  \ref{li:t3}, because for every $(Y,S,Z)\in\CT$ we have $X\cap
  Z\neq\emptyset$, because $X\subseteq S\cup Z$ and $X>k-1\ge|S|$.
\end{proof}

It follows from Lemma~\ref{lem:t2} that the lower bound on $|X|$ in
the Lemma~\ref{lem:t1} is tight.

\subsection{Minimal Elements}
Let $G$ be a graph. We define a partial order $\preceq$ on $\Sep(G)$ by letting
\begin{equation}
  \label{eq:2}
  (Y,S,Z)\preceq(Y',S',Z')\quad:\Longleftrightarrow\quad
  S\cup Z\subset S'\cup Z'\text{ or }\big(S\cup
  Z=S'\cup Z'\text{  and }S\subseteq S'\big).
\end{equation}
Note that if $|S|=|S'|$, then
$(Y,S,Z)\preceq(Y',S',Z')\iff (Z,S,Y)\succeq(Z',S',Y')$; this is not
necessarily the case if $|S|\neq |S'|$. For a $G$-tangle $\CT$, we let
$\CMT$ be the set of minimal elements of $\CT$ with respect to the
partial order $\preceq$.

\begin{lem}[Reed~\cite{ree97}]\label{lem:reed}
  Let $\CT$ be a $G$-tangle of order $k$. Then for every set
  $S\subseteq V(G)$ of cardinality $|S|<k$ there is a connected
  component $C_{\CT}(S)$ of $G\setminus S$ such that for all $Y,Z$
  such that
  $(Y,S,Z)\in\Sep_{<k}(G)$,
  \[
  (Y,S,Z)\in\CT\iff V(C_{\CT}(S))\subseteq Z.
  \]
\end{lem}

For a $G$-tangle $\CT$ of order $k$ and a set $S\subseteq V(G)$ of
cardinality $|S|<k$, we let 
\begin{equation}\label{eq:3}
\CZ_{\CT}(S):=V(C_{\CT}(S)),
\end{equation}
where
$C_{\CT}(S)$ is the connected component of $G\setminus S$ from
Lemma~\ref{lem:reed}.  Furthermore, we let 
\begin{equation}\label{eq:4}
\CY_{\CT}(S):=V(G)\setminus\big(S\cup\CZ_{\CT}(S)\big).
\end{equation}
Note that $(\CY_\CT(S),S,\CZ_{\CT}(S))\in\CT$.
 
\begin{cor}\label{cor:reed2}
  Let $\CT$ be a tangle, and let $(Y,S,Z)\in\CMT$. Then
  $S=N(Z)$ and $Z=\CZ_\CT(S)$.
\end{cor}

\begin{cor}\label{cor:reed3}
  Let $\CT$ be a $G$-tangle of order $k$, and let
  $(Y_1,S_1,Z_1), (Y_2,S_2,Z_2)\in\CMT$ be distinct. Then
  $|(S_1\cap Z_2)\cup(S_1\cap S_2)\cup(Z_1\cap S_2)|\ge k$ and $S_i\cap
  Z_{3-i}\neq\emptyset$ for $i=1,2$.
\end{cor}

\begin{proof}
  $S:=(S_1\cap Z_2)\cup(S_1\cap S_2)\cup(Z_1\cap S_2)$ is the
  separator of the separation
  $(Y,S,Z):=(Y_1,S_1,Z_1)\cap(Y_2,S_2,Z_2)$. As $(Y_1,S_1,Z_1)$ and
  $(Y_2,S_2,Z_2)$ are distinct, either $(Y,S,Z)\prec(Y_1,S_1,Z_1)$ or
  $(Y,S,Z)\prec(Y_2,S_2,Z_2)$. By the minimality of $(Y_i,S_i,Z_i)$
  this implies $(Y,S,Z)\not\in\CT$, and thus by
  Lemma~\ref{lem:tangle-closure}(2), $|S|\ge k$.

  If $S_i\cap Z_{3-i}=\emptyset$ then $S\subseteq S_{3-i}$, and as
  $|S_{3-i}|<k$, this contradicts $|S|\ge k$.
\end{proof}

\begin{cor}\label{cor:reed4}
  Let $\CT$ be a $G$-tangle and
  $(Y,S,Z)\in\CMT$. Then for every $(Y',S',Z')\in\Sep(G)$ with $S'\subseteq S$ it holds that $Z\cup (S\setminus S')\subseteq Z'$.
\end{cor}

\begin{proof}
  By Corollary~\ref{cor:reed2}, $Z\cup(S\setminus S')$ is connected and hence
  contained in some connected component of $G\setminus S'$. Thus
  $Z\cup (S\setminus S')\subseteq Z'$ or $Z\cup (S\setminus
  S')\subseteq Y'$. In the latter case, we have $|(Z\cup S)\cap
  (Z'\cap S')|=|S'|<k$, which contradicts Corollary~\ref{cor:tangle-closure}.
\end{proof}

\subsection{Tangles of Order at Most $3$}
\label{sec:t3}
Let $G$ be a graph and $k\ge 1$. Following \cite{cardiehun+14}, we
call an inclusionwise maximal $k$-inseparable set
$X\subseteq V(G)$ a \emph{$k$-block} of $G$. We call a $k$-block
$X$ \emph{proper} if $|X|\ge k+2$. Observe that if $X$ is a proper
$k$-block then the torso $\torso GX$ is $(k+1)$-connected. 

It can be shown that for $k=2,3$ the torsos $\torso GX$ of the
$(k-1)$-blocks $X$, which for
$k\le 2$ coincide with the induced subgraphs $G[X]$ and for $k=3$ are
topological subgraphs of $G$, are precisely the biconnected and
triconnected components appearing the decomposition described in the introduction.

By Lemma~\ref{lem:t1}, if $X$ is a $(k-1)$-block for $k=1,2$ or a
proper $k$-block for $k=3$, then $\CT^k(X)$ is a $G$-tangle of order
$k$. The following theorem shows that all $G$-tangles of order at most
$3$ are of this form. 

\begin{theo}[\cite{gm10,gro16a}]\label{theo:tangles_vs_3cc} 
  Let $G$ be a graph, and let $\CT$ be a $G$-tangle of order $k\le
  3$. Then the set 
  \begin{equation}
    \label{eq:5}
    X_{\CT}:=\bigcap_{(Y,S,Z)\in\CT}(S\cup Z)
  \end{equation}
  is a $k$-block (proper if $k=3$) and $\CT=\CT^k(X_{\CT})$.
\end{theo}

The theorem utterly fails for $k=4$: a hexagonal grid $H$ (see
Figure~\ref{fig:hexgrids}) has a unique $H$-tangle $\CT$ of order $4$,
but the set $X_{\CT}$ (defined as in \eqref{eq:5}) is empty, and in
fact $H$ has no $3$-inseparable set.

As a motivation for our definition of ``quasi-4-connected regions'' in
Section~\ref{sec:q4r}, let us give an alternative characterisation of
the proper $2$-blocks. We have already remarked that they are
precisely the vertex sets of the triconnected components. In view of
our later terminology, we call them \emph{triconnected regions}.

\begin{prop}\label{prop:3cr}
  Let $G$ be a graph and $R\subseteq V(G)$. The the following are
  equivalent.
  \begin{enumerate}
  \item $R$ is a triconnected region of $G$.
  \item $R$ is an inclusionwise maximal subset of $G$ such that
    $\torso GR$ is 3-connected and a topological subgraph of $G$.
  \item 
    $\torso GR$ is 3-connected and a topological subgraph of $G$, and
    for every connected component $C$ of $G\setminus R$ we have
    $|N(C)|\le 2$.
  \end{enumerate}
\end{prop}

\begin{proof}
  To prove that (1) implies (3), let $R$ be a triconnected region of
  $G$, that is, an inclusionwise maximal $2$-inseparable set
  $R\subseteq V(G)$ of cardinality $|R|\ge 4$. We have already noted
  that the torso of a proper $2$-block is 3-connected.

  Suppose for contradiction that $C$ is a connected component of
  $G\setminus R$ such that $|N(C)|\ge 3$. Let $C^+$ be the subgraph of
  $G$ with vertex set $V(C)\cup N(C)$ and all edges that have at least
  one endvertex in $C$ (that is, all edges of $C$ and all edges from
  $C$ to $N(C)$). By the maximality of $R$, for every $y\in V(C)$
  there is a separation $(Y,S,Z)$ of order at most $2$ such that $y\in
  Y$ and $R\subseteq S\cup Z$. Let $v_1,v_2,v_3\in N(C)$ be distinct
  and $w_1\in N(v_1)\cap V(C)$. Let $(Y_1,S_1,Z_1)\in\Sep_{\le 2}(G)$ such that $w_1\in
  Y_1$ and $R\subseteq S_1\cup Z_1$, and subject to these conditions,
  $(Y_1,S_1,Z_1)$ is $\preceq$-minimal. Then $v_1\in S_1$, and as $C$
  is connected, there is another vertex $x_1\in S_1$ that separates
  $w_1$ from $v_2,v_3$ in the graph $C^+$. It
  follows from the minimality $(Y_1,S_1,Z_1)$ that there is no $x_2$
  separating $x_1$ from $v_2,v_3$. Hence there are two internally
  disjoint paths $P_2,P_3\subseteq C^+$ from $x_1$ to $v_2,v_3$,
  respectively. Moreover, there is a path $P_1\subseteq C^+$ from
  $x_1$ to $v_1$ (via $w_1$), because $C^+$ is connected. $P_1$ is
  internally disjoint from $P_2$ and $P_3$, because otherwise $x_1$
  would not separate $w_1$ from $v_2,v_3$. But this implies that there
  is no $(Y_2,S_2,Z_2)\in\Sep_{\le 2}(G)$ such that $x_1\in Y_2$ and
  $R\subseteq S_2\cup Z_2$. This is a contradiction.

  Hence for every connected component $C$ of $G\setminus R$ we have
  $|N(C)|\le 2$. This directly implies that $\torso GR$ is a
  topological subgraph of $G$. 

  \medskip To prove that (3) implies (2), suppose that $R$ satisfies
  (3) and that there is an $R'\supset R$ such that $\torso G{R'}$ is
  3-connected. Let $C$ be a connected component of $G\setminus R$ that
  contains a vertex of $R'\setminus R$. Then $|N(C)|\le2$ and thus
  $(V(G)\setminus (V(C)\cup N(C)),N(C),V(C))\in\Sep_{\le 2}(G)$. This
  implies that
  $(R'\setminus (V(C)\cup N(C)),N(C),V(C)\cap R') \in\Sep_{\le
    2}(\torso G{R'})$,
  which contradicts $\torso G{R'}$ being 3-connected.

  \medskip
  Finally, to prove that (2) implies (1), suppose that $R$ satisfies
  (2). Then $|R|\ge 4$, because $\torso GR$ is 3-connected. For every
 separation $(Y,S,Z)\in\Sep_{\le 2}(G)$ with $Y\cap R\neq\emptyset$ and $Z\cap
  R\neq\emptyset$ the triple $(Y\cap R,S,Z\cap R)$ is a proper separation of
  $\torso GR$ of the same order, which implies that $R$ is $2$-inseparable. Suppose for
  contradiction that $R$ is not maximal 2-inseparable, and let $R'\supset
  R$ be the 2-block that contains $R$. Then by the implications
  (1)$\implies$(3)$\implies$(2), $R'$ satisfies (2) as well, and this
  is a contradiction.
\end{proof}

\subsection{Lift and Project}
We can ``lift'' a
tangle from a minor of a graph to the original graph.
Let $G$ be a graph, $H$ a minor of $G$, and $\CM$ a model of $H$ in
$G$, say, with branch sets $(M_w)_{w\in
  V(H)}$. For a separation
$(Y,S,Z)\in\Sep(G)$, the \emph{$\CM$-projection} of $(Y,S,Z)$ to $H$ is the
triple $\pi_\CM(Y,S,Z)=(Y',S',Z')$ of subsets of $V(H)$ defined by
\begin{align}
  \notag
  Y'&:=\big\{w\in V(H)\bigmid V(M_w)\subseteq Y\big\},\\
  \label{eq:proj}
  S'&:=\big\{w\in V(H)\bigmid V(M_w)\cap S\neq\emptyset\},\\
  \notag
  Z'&:=\big\{w\in V(H)\bigmid V(M_w)\subseteq Z\big\}.
\end{align}
It is easy to see that $(Y',S',Z')$ is a separation of $H$ of order
$|S'|\le |S|$.

\begin{lem}[\cite{gm10}]\label{lem:lift}
  Let $G$ be a graph, $H$ a minor of $G$, and $\CM$ a model of $H$ in
  $G$. Let $\CT'$ be an $H$-tangle of order $k$. Then the set $\CT$ of
  all separations $(Y,S,Z)\in\Sep_{<k}(G)$ such that
  $\pi_{\CM}(Y,S,Z)\in\CT'$ is a $G$-tangle of order $k$.
\end{lem}

We call $\CT$ be the \emph{lifting} of $\CT'$ to $G$ with respect to
the model $\CM$. 
Clearly, the lifting may depend on the model. This is even the
case if we only consider faithful minors and models.
It is easy to see that the lifting relation is transitive, that is, if
we have graphs $G,G',G''$ such that $G'$ a minor of $G$ and $G''$ a
minor of $G'$, tangles $\CT,\CT',\CT''$ of $G,G',G''$ such that $\CT'$
the lifting of $\CT''$ to $G'$ with respect to some model of $G''$ in
$G'$ and $\CT$ the lifting of $\CT'$ to $G$ with respect to some model
of $G'$ in $G$, then there is a model of $G''$ in $G$ such that $\CT$
is the lifting of $\CT''$ to $G$ with respect to this model.

So we can lift tangles from minors of a graph to the graph. What about
the converse: do tangles of a graph induce tangles of their minors?
Obviously not in general, but in the following lemma we identify a
useful special case where a tangle of graph induces a tangle of some
minor that is a torso of a triconnected region. We need some additional terminology.
Let $\CT,\CT'$ be $G$-tangles. If $\CT'\subseteq\CT$, we say that
$\CT$ is an \emph{extension} of $\CT'$ and $\CT'$ is a
\emph{truncation} of $\CT$. Observe that every $G$-tangle $\CT$ of
order $k$ has a unique truncation to every order $k'\le k$.

\begin{lem}\label{lem:induced-tangle}
  Let $\CT$ be a $G$-tangle of order $4$ such that the truncation
  of $\CT$ to order $3$ is $\CT^3(R)$ for some triconnected region $R$
  of $G$. Let $\torso{\CT}{R}$ be the set of all
  separations $(Y',S',Z')\in\Sep_{<4}(\torso GR)$ such that
  there is a separation $(Y,S,Z)\in\CT$ with $Y'=Y\cap X,S'=S\cap
  X,Z'=Z\cap X$.

  Then $\torso{\CT}{R}$ is a $\torso{G}{R}$-tangle of order $4$.
\end{lem}

\begin{proof}
    We note first that for all $(Y,S,Z)\in\CT$ we
  have $Z\cap R\neq\emptyset$. To see this, by Lemma~\ref{lem:reed}, we may assume
  without loss of generality that $Z$ is connected in $G$ and that
  $S=N(Z)$. Hence if $Z\cap R=\emptyset$, there is a connected
  component $C$ of $G\setminus R$ such that $Z\subseteq V(C)$. By
  Proposition~\ref{prop:3cr}, we have $|N(C)|\le 2$. Thus
  $(V(C),N(C),V(G)\setminus (V(C)\cup N(C)))\in\Sep_{<3}(C)$. As the truncation
  of $\CT$ to order $3$ is $\CT^3(R)$ and $R\cap V(C)=\emptyset$, we
  have $(V(C),N(C),V(G)\setminus (V(C)\cup N(C)))\in\CT$. But as
  $(Y,S,Z)\in \CT$ and $Z\subseteq V(C)$ and thus $Z\cup S\subseteq
  V(C)\cap N(C)$, this contradicts, this contradicts
  Corollary~\ref{cor:tangle-closure}.

  Let us now prove that $\CT':=\torso{\CT}{R}$ satisfies the tangle
  axioms. Let $G':=\torso GR$. 
  To prove that $\CT'$ satisfies \ref{li:t1}, let $(Y',S',Z')\in\Sep_{<4}(G')$. For every
  connected component $C$ of $G\setminus R$ the set $N(C)$ is a clique
  in $G'$ and thus either $N(C)\subseteq Y'\cup S'$ or
  $N(C)\subseteq Z'\cup S'$. We let $Y$ be the union of $Y'$ with
  the vertex sets of all connected components $C$ of $G\setminus R$
  such that $N(C)\subseteq Y'\cup S'$, and we let $Z$ be the union of $Z'$ with
  the vertex sets of all remaining connected components of $G\setminus
  R$. Then $(Y,S,Z)\in\Sep_{<4}(G)$. By
  \ref{li:t1}, either $(Y,S,Z)\in\CT$ or $(Z,S,Y)\in\CT$, and hence either $(Y',S',Z')\in \CT'$ or
  $(Z',S',Y')\in\CT'$.

  To prove that $\CT'$ satisfies \ref{li:t2}, let
  $(Y_i',S_i',Z_i')\in\CT'$ for $i=1,2,3$. Then there are
  $(Y_i,S_i,Z_i)\in\CT$ such that $Y_i\cap R=Y'_i$, $S_i\cap R=S_i'$,
  and $Z_i\cap R=Z_i'$. By Lemma~\ref{lem:reed}, we may assume that
  the sets $Z_i$ are connected in $G$. By our observation above, they
  have a nonempty intersection with $R$.  By \ref{li:t3}, either there
  is a vertex $v\in Z_1\cap Z_2\cap Z_3$ or an edge $e$ that has at
  least one endvertex in every $Z_i$. Assume the latter, the argument
  in the former case is similar (and simpler). Let $z_i$ be the
  endvertex of $e$ in $Z_i$. If all the $z_i$ are in $R$, the
  edge $e$ is also an edge of $G'$ which has an endvertex in
  every $Z_i'$. Otherwise, there is a connected
  component $C$ of $G\setminus X$ such that all the $z_i$ are in
  $V(C)\cap N(C)$. If $z_i\in V(C)$, then
  $Z_i\cap N(C)\neq\emptyset$, because $Z_i$ is connected in
  $G$ and has a nonempty intersection with $R$. Thus $Z_1', Z_2',Z_3'$
  all contain at least one of the at most two vertices in $N(C)$. Thus
  either they share a vertex, or the edge of $G'$ that connects
  the two vertices in $N(C)$ has an endvertex in all three $Z_i'$.

  Finally, \ref{li:t3} follows from the fact that for every vertex
  separation $(Y,S,Z)\in\CT$ we have
  $Z\cap R\neq\emptyset$. 
\end{proof}

\section{Tangles of Order 4}
\label{sec:4t}

Let us now look at tangles of order
$4$. Lemma~\ref{lem:induced-tangle}, in combination with
Theorem~\ref{theo:tangles_vs_3cc} allows us to focus on 3-connected
graphs. 
The main result of this section is a correspondence between
tangles of order $4$ and what we will call \emph{quasi-4-connected regions}
of a graph. This correspondence holds for all but a small number of
\emph{exceptional} regions, which we shall completely characterise. We
first state the theorem; the necessary definitions follow as we go along.

\begin{theo}[Correspondence Theorem]\label{theo:corr}
  Let $G$ be a 3-connected graph. Then with every quasi-4-connected
  region $R$ of $G$ we can associate a $G$-tangle $\CT_R$ of order $4$
  and with every $G$-tangle $\CT$ of order $4$ a quasi-4-connected
  region $R_{\CT}$ such that 
  \[
  \CT=\CT_{R_\CT}.
  \]
\end{theo}

We shall call the torsos $\torso G{R_{\CT}}$ for the $G$-tangles of
order $4$ the \emph{quasi-4-connected components} of $G$.

In general, the mapping  $R\mapsto \CT_R$ is
not injective; the mapping $\CT\mapsto R_\CT$ is (otherwise the theorem
could not hold). The mapping $R\mapsto\CT_R$ is \emph{canonical} (or
\emph{isomorphism invariant}). This
means that for any two graphs $G,G'$ and regions $R,R'$, if $f$ is an
isomorphism from $G$ to $G'$ that maps $R$ to $R'$ then $f$ also maps
$\CT_R$ to $\CT_{R'}$. This will be obvious from the construction. The
mapping $\CT\mapsto R_{\CT}$ is not canonical. However, the mapping
from $\CT$ to the quasi-4-connected component $\torso G{R_{\CT}}$,
viewed as an abstract graph, is (see
  Corollary~\ref{cor:q4c-can}).

\subsection{Quasi-4-Connected Graphs}
Recall from the introduction that a graph $G$ is \emph{quasi-4-connected} if %
$G$ is 3-connected and for all separations $(Y,S,Z)$ of $G$
of order $3$, either $|Y|\le 1$ or $|Z|\le 1$. In this section, we
will analyse tangles of order $4$ of quasi-4-connected graphs.

\begin{lem}\label{lem:q4c1}
  Let $G$ be a quasi-4-connected graph of order $|G|\ge 8$. Let 
  \[
  \CT:=\big\{(Y,S,Z)\in\Sep_{<4}(G)\bigmid |Y|<|Z|\big\}.
  \]
  Then  $\CT$ is a $G$-tangle of order $4$.
\end{lem}

\begin{proof}
  To see that $\CT$ satisfies \ref{li:t1}, let $(Y,S,Z)\in\Sep_{<4}(G)$. Without loss of
  generality, we assume that $|Y|\le|Z|$. Then $|Y|\le 1$ and thus
  $|Z|=V(G)\setminus(Y\cup S)\ge 4$. Hence $(Y,S,Z)\in\CT$.  Similarly,
  $\CT$
  satisfies \ref{li:t3}, because for all $(Y,S,Z)\in\CT$ it holds
  that $|Y\cup S|\le 4<|V(G)|$.

  It remains to prove that $\CT$ satisfies \ref{li:t2}. For $i=1,2,3$,
  let $(Y_i,S_i,Z_i)\in\CT$. Suppose for contradiction that $Z_1\cap
  Z_2\cap Z_3=\emptyset$ and that there is no edge that
  has an endvertex in each $Z_i$.

  \begin{claim}
    For distinct $i,j,k\in[3]$ and $x\in V(G)$, if $x\in Z_i\cap Z_j$
    then $x\in Y_k$.

    \proof 
    We have $x\not\in Z_k$, because $Z_1\cap Z_j\cap
    Z_k=\emptyset$. Suppose that $x\in S_k$, and let $z\in N(x)\cap
    Z_k$. Such a $z$ exists, because
    $Z_k\neq\emptyset$ and $N(Z_k)\subseteq S_k$, and as $|S_k|\le 3$
    and $G$ is 3-connected, this implies $N(Z_k)=S_k$. But the edge
    $xz$ has an endvertex in every $Z_i$, which contradicts our
    assumption that no such edge exists. 
    \uend
  \end{claim}

  We have $|Y_i|\le 1$ and thus
  $|V(G)\setminus(Y_1\cup Y_2\cup Y_3)|\ge 5$. A simple double counting
  argument shows that there is a vertex
  $x\in V(G)\setminus(Y_1\cup Y_2\cup Y_3)$ such that $x$ is only
  contained in one of the three sets $S_1,S_2,S_3$ (count pairs
  $(x,S_i)$ where $x\in V(G)\setminus(Y_1\cup Y_2\cup Y_3)$ and
  $i=1,2,3$). But an $x\in V(G)\setminus(Y_1\cup Y_2\cup Y_3)$ contained
  in at most one of the sets $S_i$ is contained in two of the sets
  $Z_i$, and this contradicts Claim~1.
\end{proof}

\begin{figure}
  \centering
  \begin{tikzpicture}
  [
  vertex/.style={draw,circle,fill=black,inner sep=0mm,minimum
    size=2mm},
  every edge/.style={draw,thick}
  ]

    \path
    (1.5,-1.5) node[vertex] (v1) {} node[below right] {$v_1$}
    (0,0) node[vertex] (v2) {} node[left] {$w_1$}
    (1.5,0) node[vertex] (v3) {}  node[below right] {$w_2$}
    (3,0) node[vertex] (v4) {} node[right] {$w_3$}
    (0,1.5) node[vertex] (v5) {}  node[left] {$v_2$}
    (1.5,1) node[vertex] (v6) {}  node[above] {$v_3$}
    (3,1.5) node[vertex] (v7) {} node[right] {$v_4$}
    ;

    \path 
    (v1) edge (v2) edge (v3) edge (v4) edge (v5) edge[bend left=15] (v6) edge (v7)
    (v2) edge (v5) edge (v6)
    (v3) edge (v5) edge (v7)
    (v4) edge (v6) edge (v7)
    (v5) edge (v6) edge (v7)
    (v6) edge (v7)
    ;
\end{tikzpicture}
  \caption{The Graph $\THT$}
  \label{fig:th3}
\end{figure}

\begin{exa}\label{exa:cube-1}
  Figure~\ref{fig:th3} shows a quasi-4-connected graph $\THT$ with seven
  vertices and no tangle of order $4$. The name $\THT$ is motivated by
  the fact that this graph can be seen as a tetrahedron with 3
  ``corners'' attached to it.

  To see that $\THT$ has no tangle of order $4$, suppose for
  contradiction that $\CT$ is a $\THT$-tangle
  of order $4$. For $i=1,2,3$, let $S_i=N(w_i)$ and
  $Z_i=V(\THT)\setminus(\{w_i\}\cup S_i)$. Then by Lemma~\ref{lem:t2},
  we have $(\{w_i\},S_i,Z_i)\in\CT$. However, $Z_1\cap
  Z_2\cap Z_3=\emptyset$, and for every edge $e=xy$ of $\THT$ there is an
  $i$ such that $x,y\in\{w_i\}\cup S_i$. Hence $\CT$ violates tangle
  axiom \ref{li:t2}.
  \uend
\end{exa}

\begin{figure}
  \centering
  \begin{tikzpicture}
  [
  vertex/.style={draw,circle,fill=black,inner sep=0mm,minimum
    size=2mm},
  every edge/.style={draw,thick}
  ]

    \path
    (0,-0.5) node[vertex] (x1) {} node[below left] {$v_1$}
    (1.5,0) node[vertex] (x2) {} node[below] {$v_2$}
    (3,-0.5) node[vertex] (x3) {} node[below right] {$v_3$}
    (0,1.5) node[vertex] (y1) {} node[above left] {$w_1$}
    (1.5,2) node[vertex] (y2) {} node[above] {$w_2$}
    (3,1.5) node[vertex] (y3) {} node[above right] {$w_3$}
  ;

    \path 
    (x1) edge (x2) edge (x3) edge (y1) edge (y2) edge (y3) 
    (x2) edge (x3) edge (y1) edge (y2) edge (y3) 
    (x3) edge (y1) edge (y2) edge (y3) 
    ;
\end{tikzpicture}
  \caption{The graph $\TRT$}
  \label{fig:tr3}
\end{figure}

We now give a precise characterisation of the quasi-4-connected graphs
with a tangle of order $4$. A quasi-4-connected graph is
\emph{exceptional} if it is either isomorphic to a subgraph of the
graph $\THT$ shown in Figure~\ref{fig:th3} or isomorphic to a subgraph
of the graph $\TRT$ shown in Figure~\ref{fig:tr3}.

\begin{lem}\label{lem:q4c2}
  Let $G$ be an exceptional quasi-4-connected graph. Then $G$ has no
  tangle of order $4$.
\end{lem}

\begin{proof}
  As all supergraphs of a graph that has a tangle of order $4$ also have
  a tangle of order $4$, we may
  assume without loss of generality that $G$ is an inclusionwise
  maximal exceptional quasi-4-connected graph, that is, $\THT$ or $\TRT$.

  We have already seen in Example~\ref{exa:cube-1} that $\THT$ has no tangle of order $4$.

  Suppose for contradiction that
  that $\CT$ is a $\TRT$-tangle of order $4$. By Lemma~\ref{lem:t2},
  the separations 
  $(Y_i,S,Z_i):=\big(\{w_i\},\{v_1,v_2,v_3\},\{w_2,w_2,w_3\}\setminus\{w_i\}\big)$
  of order $3$ are in $\CT$. However, we have $Z_1\cap Z_2\cap
  Z_2=\emptyset$, and for every edge $xy\in E(\TRT)$ there is an $i$ such
  that $x,y\in Y_i\cup S$. This contradicts~\ref{li:t2}.
\end{proof}

\begin{theo}\label{theo:q4c}
  Let $G$ be a quasi-4-connected graph. Then $G$ has a tangle of order
  $4$ if and only if it is not exceptional.

  Furthermore, if $G$ has a tangle of order $4$, it has exactly one
  such tangle, which consists of all separations
  $(Y,S,Z)\in\Sep_{<4}(G)$ such that $|Y|<|Z|$.
\end{theo}

\begin{proof}
  We have already seen that if $G$ is exceptional then it has no
  tangle of order $4$. To prove the converse, we assume that $G$ is not
  exceptional. Let 
  \[
  \CT:=\big\{(Y,S,Z)\in\Sep_{<4}(G)\bigmid |Y|<|Z|\big\}.
  \]
  We shall prove that $\CT$ is a $G$-tangle of order $4$.

  By Lemma~\ref{lem:q4c1}, we may assume that
  $|G|\le 7$. Note that $|G|\ge 5$, because the only quasi-4-connected graph of
  order at most $4$, the tetrahedron, is a subgraph of $\THT$ (and
  also of $\TRT$) and hence exceptional. If $G$ is
  4-connected, then it follows from Lemma~\ref{lem:t1} that $\CT$ is a
  $G$-tangle of order $4$. Hence we may assume that $|G|$ is not
  4-connected. But then $|G|\ge 6$, because all graphs of order $5$
  that are not 4-connected are subgraphs of $\THT$  (and
  also of $\TRT$) and hence exceptional. From
  from now on we assume that $6\le|G|\le 7$ and that $G$ is not 4-connected. 

  $\CT$ trivially \ref{li:t3}. To see that it satisfies
  \ref{li:t1}, let $(Y,S,Z)\in\Sep_{<4}(G)$. Without loss of generality, we may assume that $|Y|\le |Z|$. As
 $G$ is quasi-4-connected, we have $|Y|\le 1$. Thus
 $|Z|=|G|\setminus|Y\cup S|\ge 6-4=2>|Y|$.

 It remains to prove that $\CT$ satisfies \ref{li:t2}. For $i=1,2,3$,
  let $(Y_i,S_i,Z_i)\in\CT$. Suppose for contradiction that $Z_1\cap
  Z_2\cap Z_3=\emptyset$ and that there is no edge that
  has an endvertex in each $Z_i$.

  The next claim is the same as Claim~1 in the proof of
  Lemma~\ref{lem:q4c1}. 
  
  \begin{claim}
    For distinct $i,j,k\in[3]$ and $x\in V(G)$, if $x\in Z_i\cap Z_j$
    then $x\in Y_k$.
    \uend
  \end{claim}

  \begin{cs}
    \case1: $|G|=6$.\\
    Then by Claim~2, at least two of the sets $Y_i$ must be nonempty. Say, $Y_1=\{y_1\}$ and
    $Y_2=\{y_2\}$. Then $S_i=N(y_i)$ for $i=1,2$. 

    Suppose for contradiction $y_1\in S_2$. Then $y_2\in S_1$. A similar
    double counting argument as above shows that at least one of the
    remaining four vertices in $V(G)\setminus(Y_1\cup Y_2)$ is contained
    in at most one of the sets $S_i$: there are two pairs $(x,S_1)$
    with $x\in V(G)\setminus (Y_1\cup Y_2)$ and two pairs $(x,S_2)$
    with $x\in V(G)\setminus (Y_1\cup Y_2)$ and at most three pairs
    $(x,S_3)$; overall at most seven pairs $(x,S_i)$. But if every $x$
    was contained in two of the sets $S_i$, there would be eight such
    pairs. As above, an
    $x\in V(G)\setminus(Y_1\cup Y_2\cup Y_3)$ contained in at most one of
    the sets $S_i$ is contained in two of the sets $Z_i$, and this
    contradicts Claim~1.

    So $y_1\not\in S_2$ and $y_2\not\in S_1$. As $V(G)= 6$, we must
    have $|S_1\cap S_2|\ge 2$. 

    \begin{cs}
      \case{1a}
      $|S_1\cap S_2|=3$.\\
      Then $S_1=S_2$. Let $y_3$ be the unique vertex in
    $V(G)\setminus(S_1\cup\{y_1,y_2\})$. Then $N(y_3)=S_1$, because
    $y_1,y_2\not\in N(y_3)$ (otherwise $S_1=S_2$ would not separate
    $y_1,y_2$ from $y_3$). Thus $G$ is isomorphic to a subgraph of
    $\TRT$: the three vertices in $S_1$ can be mapped to
    $v_1,v_2,v_3$, and the vertices $y_1,y_2,y_3$ can be mapped to
    $w_1,w_2,w_3$.

    \case{1b}
    $|S_1\cap S_2|= 2$.\\
    Say, $S_1\cap S_2=\{x_1,x_2\}$. For
    $i=1,2$, let $x_{2+i}$ be the unique vertex in
    $S_i\setminus\{x_1,x_2\}$. Figure~\ref{fig:ex1}(a) shows the
    situation. As there is no edge from $y_1$ to $y_2$, this shows that
    $G$ is isomorphic to a subgraph of $\THT$: the four vertices $x_1,\ldots,x_4$
    can be mapped to $v_1,\ldots,v_4$ and $y_1,y_2$ to $w_1,w_2$, respectively.
    \end{cs}

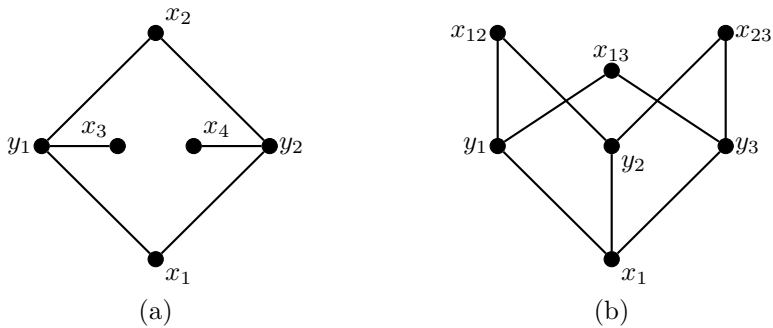
\begin{figure}
  \centering
  \begin{tikzpicture}
  [
  vertex/.style={draw,circle,fill=black,inner sep=0mm,minimum
    size=2mm},
  every edge/.style={draw,thick}
  ]

  \begin{scope}
      \path
      (0,0) node[vertex] (y1) {} node[left] {$y_1$}
      (1,0) node[vertex] (z1) {} node[above left] {$x_3$}
      (2,0) node[vertex] (z2) {} node[above right] {$x_4$}
      (3,0) node[vertex] (y2) {} node[right] {$y_2$}
      (1.5,1.5) node[vertex] (xx) {} node[above right] {$x_2$}
      (1.5,-1.5) node[vertex] (x) {} node[below right] {$x_1$}
      ;

  \path 
  (y1) edge (z1) edge (x) edge (xx) 
  (y2) edge (z2) edge (x) edge (xx) 
  ;

      \path (1.5,-2.2) node {(a)};
  \end{scope}

  \begin{scope}[xshift=6cm]

    \path
    (1.5,-1.5) node[vertex] (v1) {} node[below right] {$x_1$}
    (0,0) node[vertex] (v2) {} node[left] {$y_1$}
    (1.5,0) node[vertex] (v3) {}  node[below right] {$y_2$}
    (3,0) node[vertex] (v4) {} node[right] {$y_3$}
    (0,1.5) node[vertex] (v5) {}  node[left] {$x_{12}$}
    (1.5,1) node[vertex] (v6) {}  node[above] {$x_{13}$}
    (3,1.5) node[vertex] (v7) {} node[right] {$x_{23}$}
    ;

    \path 
    (v1) edge (v2) edge (v3) edge (v4)
    (v2) edge (v5) edge (v6)
    (v3) edge (v5) edge (v7)
    (v4) edge (v6) edge (v7)
    ;

     \path (1.5,-2.2) node {(b)};
 
  \end{scope}

\end{tikzpicture}
  \caption{Proof of Theorem~\ref{theo:q4c}}
  \label{fig:ex1}
\end{figure}

    \case2: $|G|=7$.\\
    Then all three sets $Y_i$ must be nonempty. Say, $Y_i=\{y_i\}$, and note
    that $S_i=N(y_i)$. 

    By essentially the same argument as in Case~1, we have $y_i\not\in
    S_j$ for all $i,j$. Furthermore, $S_i\neq S_j$ for all $i\neq j$,
    because otherwise $\big((\{y_i,y_j\},S_i,V(G)\setminus (S_i\cup\{y_i,y_j\})\big)$
    is separation of $G$ of order $3$ where both sides have
    cardinality at least $2$, which contradicts $G$ being
    quasi-4-connected. Hence $|S_i\cap S_j|\le 2$. By the usual
    argument based on Claim~1, each of the four vertices in $V(G)\setminus\{y_1,y_2,y_3\}$
    is contained in at least two of the sets $S_i$. It follows that
    there is one vertex $x_1\in S_1\cap S_2\cap S_3$ and for all
    distinct $i,j,k$
    a vertex $x_{ij}\in S_i\cap S_j\setminus S_k$. So far, the graph
    $G$ looks like the graph in Figure~\ref{fig:ex1}(b). This shows
    that it is a subgraph of $\THT$.
    \qedhere
 \end{cs}
\end{proof}

\subsection{Quasi-4-Connected Regions}
\label{sec:q4r}

For the rest of Section~\ref{sec:4t}, we make the following
assumption.

\begin{ass}\label{ass:4t1}
  $G$ is a 3-connected graph.
\end{ass}

A \emph{quasi-4-connected region} of $G$ is a
subset $R\subseteq V(G)$ satisfying the following conditions.
\begin{nlist}{Q}
  \item\label{li:q1} $\torso{G}{R}$ is a faithful minor of $G$.
  \item\label{li:q2} $\torso{G}{R}$ is quasi-4-connected.
  \item\label{li:q3} For every connected component $C$ of $G\setminus R$ it holds that $N(C)=3$.
\end{nlist}
While conditions \ref{li:q1} and \ref{li:q2} are, to some extent,
natural, condition~\ref{li:q3} may seems less so. It is a (weak) maximality
condition: if $R'\supset R$ such that $\torso{G}{R'}$ is
quasi-4-connected, then $R'\setminus R$ contains at most one vertex of
every connected component of $G\setminus R$ (unless $|R|=4$). 
Conditions \ref{li:q1}--\ref{li:q3} are motivated by the characterisation of 3-connected
components given in Proposition~\ref{prop:3cr}(3). The reason for
choosing these conditions instead of adding some maximality condition
is simply that it works best in combination with tangles and for the
Decomposition Theorem; it is condition \ref{li:q3} which guarantees
that our decomposition will have adhesion $3$.

In the remainder of Section~\ref{sec:q4r}, we shall prove that we can associate
a tangle of order $4$ with every quasi-4-connected region, up to a
finite number of small exceptional cases. These exceptional cases will
be derived from the exceptional quasi-4-connected graphs, but will
also have to take the surrounding graph into account. The following
example illustrates why.

\begin{figure}
  \centering
  \begin{tikzpicture}
  [
  vertex/.style={draw,circle,fill=black,inner sep=0mm,minimum
    size=2mm},
  every edge/.style={draw,thick}
  ]

    \path
    (1.5,-1.5) node[vertex] (x1) {} node[below right] {$v_1$}
    (0,0) node[vertex] (y1) {} node[left] {$w_1$}
    (1.5,0) node[vertex] (y2) {}  node[below right] {$w_2$}
    (3,0) node[vertex] (y3) {} node[right] {$w_3$}
    (0,1.5) node[vertex] (x2) {}  node[left] {$v_2$}
    (1.5,1) node[vertex] (x3) {}  node[above] {$v_3$}
    (3,1.5) node[vertex] (x4) {} node[right] {$v_4$}
   (1.5,2.5) node[vertex] (y4) {}  node[above] {$w_4$}
    ;

    \path 
    (x1) edge (y1) edge (y2) edge (y3) edge[dashed] (x2) edge[dashed,bend left=15] (x3) edge[dashed] (x4)
    (y1) edge (x2) edge (x3)
    (y2) edge (x2) edge (x4)
    (y3) edge (x3) edge (x4)
    (x2) edge[dashed] (x3) edge[dashed] (x4)
    (x3) edge[dashed] (x4)
    (y4) edge (x2) edge (x3) edge (x4)
    ;
\end{tikzpicture}
  \caption{The graph $\THF$}
  \label{fig:th4}
\end{figure}

\begin{exa}\label{exa:q4r1}
  Consider the graph $G=\THF$ in Figure~\ref{fig:th4}. The dashed edges
  may or not be present; it makes no difference for the example. Let
  $R:=\{v_1,v_2,v_3,v_4\}$. Then $R$ is a quasi-4-connected region of
  $G$. The torso $\torso GR$ is a tetrahedron, which is an exceptional
  quasi-4-connected graph. Yet the graph $G$ has a tangle $\CT$ of
  order $4$. This tangle $\CT$ consists of all separations
  $(Y,S,Z)\in\Sep_{<4}(G)$ such that $R\subseteq S\cup Z$, and hence it
  is fully justified to say that this tangle is ``associated with $R$''.

  To make the example more interesting, we may replace the vertices
  $w_i$ by larger 3-connected graphs. Then the resulting graph may have other
  tangles of order $4$. But $R$ remains a quasi-4-connected region and
  the set of all
  separations $(Y,S,Z)\in\Sep_{<4}(G)$ such that $R\subseteq S\cup Z$
  remains a tangle of order $4$.
  \uend
\end{exa}

Let us call subgraph $H$ of $\THF$ \emph{full} if it is obtained by
deleting some of the dashed edges in Figure~\ref{fig:th4}, that is,
$V(H)=V(\THF)=\{v_1,\ldots,v_4,w_1,\ldots,w_4\}$ and 
\[
\{v_iw_j\mid i,j\in[4],v_{i}w_j\in E(\THF)\}\subseteq E(H)\subseteq
E(\THF).
\]
Observe that every full subgraph of $\THF$ is non-exceptional quasi-4-connected.

Let $R$ be a quasi-4-connected region of $G$. A
\emph{non-exceptional extension} of $R$ is a graph $\hat H$ satisfying the following 
conditions.
\begin{nlist}{X}
\item\label{li:x1} $\hat H$ is a faithful minor of $G$.
\item\label{li:x2}  $\hat H$ is non-exceptional quasi-4-connected.
\item\label{li:x3}  $R\subseteq V(\hat H)$, and for each connected component $C$ of $G\setminus R$ it holds that
  $V(\hat H)\cap V(C)\le 1$.
\item\label{li:x4}
  Subject to \ref{li:x1}--\ref{li:x3}, $V(\hat H)$ is inclusionwise minimal.
\end{nlist}
Note that, by \ref{li:x1} and \ref{li:x3}, we have $R\subseteq V(\hat H)\subseteq V(G)$
We call the vertices in $V(\hat H)\setminus R$ the \emph{extension
  vertices} of $\hat H$. 
Further note that if $\torso GR$ is
non-exceptional, then we have $V(\hat H)=R$ for every
non-exceptional extension $\hat H$ of $R$. This implies $\hat
H\subseteq\torso GR$, but not necessarily $\hat
H=\torso GR$.

\begin{lem}\label{lem:q4r1}
  Let $R$ be a quasi-4-connected region of $G$ such that $\torso GR$ 
  is exceptional. Let $\hat H$ be a non-exceptional extension of
  $R$. Then $\hat H$ is isomorphic to a full subgraph of $\THF$.
\end{lem}

\begin{proof}
  Let $H:=\torso GR$ and $\hat R:=V(\hat H)$. By \ref{li:q3} and \ref{li:x3} and since $\hat H$ is
  3-connected, we have  $N^{\hat H}(z)\subseteq R$ and $|N^{\hat H}(z)|=3$ for every extension vertex
  $z\in\hat R\setminus R$. We observe next that there are no two
  extension vertices $z_1,z_2\in \hat R\setminus R$ such that
  $N^{\hat H}(z_1)=N^{\hat H}(z_2)$. Indeed, if
  $N^{\hat H}(z_1)=N^{\hat H}(z_2)=:N$ then
  $|\hat R\setminus (N\cup\{z_1,z_2\})|=1$, because $\hat H$ is
  quasi-4-connected, and thus $\hat H$ is isomorphic to a subgraph of $\TRT$, which
  contradicts $\hat H$ being non-exceptional.
 
  \begin{claim}
    $H$ is isomorphic to a subgraph of $\THT$.

    \proof Suppose for contradiction that $H$ is not isomorphic to a
    subgraph of $\THT$. As $H$ is exceptional, this means that
    $H$ is isomorphic to a subgraph of of $\TRT$ that contains the
    vertices $w_1,w_2,w_3$. Without loss of generality we assume that
    $H\subseteq\TRT$ with $w_1,\ldots,w_3\in R$. Then
    $v_1,\ldots,v_3\in R$, because otherwise $H$ is not
    3-connected. For every connected component $C$ of $G\setminus R$
    we have $N(C)\le 3$ and $N(C)$ is a clique in $H$. Thus $N(C)$
    contains at most one of the vertices $w_i$.  Let
    $z\in\hat R\setminus R$, and let $C$ be the connected
    component of $G\setminus R$ such that $z\in V(C)$. Without loss of
    generality we may assume that $w_1,w_2\not\in N(C)$. Then there is
    a separation $(Y,\{v_1,v_2,v_3\},Z)$ of $\hat H$ with
    $w_1,w_2\in Y$ and $w_3,z\in Z$, and this contradicts $\hat H$
    being quasi-4-connected.
    \uend
  \end{claim}
  
  It follows from Claim~1 that $H$ is a tetrahedron, possibly with some
 vertices of degree $3$ attached. The only way to turn such a graph
 into a non exceptional quasi-4-connected graph by attaching further
 vertices of degree $3$ with mutually non-adjacent neighbourhoods is
 to turn it into $\THF$, possibly with some of the dashed edges missing.
\end{proof}

Let us call a quasi-4-connected region $R$ \emph{non-exceptional} if
it has a non-exceptional extension $\hat H$. Let $R$ be
non-exceptional and $\hat H$ a non-exceptional extension of $R$. Let
$\CM$ be a faithful model of $\hat H$ in $G$, and let $\CHT$ be the
unique $\hat H$-tangle of order $4$. Then the lifting
$\CT(\hat H,\CM)$ of $\CHT$ with respect to $\CM$ is a $G$-tangle of
order $4$.

\begin{lem}\label{lem:q4r3}
  Let $R$ be a non-exceptional quasi-4-connected region of $G$ such
  that $\torso GR$ is exceptional. Then
  for all non-exceptional extensions $\hat H,\hat H'$ of $R$ and all faithful
  models $\CM$ of $\hat H$ and $\CM'$ of $\hat H'$ we have $\CT(\hat
  H,\CM)=\CT(\hat H',\CM')$.
\end{lem}

\begin{proof}
  Let $H:=\torso GR$.
  As $H$ is exceptional, by Lemma~\ref{lem:q4r1}, both
  $\hat H$ and $\hat H'$ are isomorphic to full subgraphs of
  $\THF$. 

  Without loss of generality we may assume that $\hat
  H\subseteq\THF$.
  Then $v_1,\ldots,v_4\in R$, because otherwise $H$ is not
  quasi-4-connected or we have no way of adding the remaining vertices
  without violating ~\ref{li:x3}.  Let $f$ be an isomorphism from $\hat H'$
  to a full subgraph of $\THF$, and let $v_i':=f^{-1}(v_i)$ and
  $w_j':=f^{-1}(w_j)$. Then $v_1',\ldots,v_4'\in R$, and by symmetry
  we may assume without loss of generality that $v_i'=v_i$ for all
  $i\in[4]$. As the $w_j$s and $w_j'$s are uniquely determined by
  their neighbours among the $v_i$, if $w_j\in R$, then $w_j=w_j'$,
  and if $w_j$ is in a component $C$ of $G\setminus R$ then $w_j'$
  is in a component $C'$ with $N(C)=N(C')=N^{\THF}(w_j)$.

  Let $\CT:=\CT(\hat H,\CM)$ and $\CT':=\CT(\hat H',\CM')$. Suppose
  for contradiction that $\CT\neq\CT'$. Then there is a 
  separation $(Y,S,Z)\in\Sep_{<4}(G)$ such that
  $(Y,S,Z)\in\CT$ and $(Z,S,Y)\in\CT'$. Let
  $(Y_M,S_M,Z_M):=\pi_{\CM}(Y,S,Z)$ and
  $(Y_M',S_M',Z_M'):=\pi_{\CM'}(Y,S,Z)$. Then
  \begin{equation}
    \label{eq:7}
    (Y_M,S_M,Z_M)\in\CHT
    \quad\text{and}\quad
    (Z_M',S_M',Y_M')\in\CHT.
  \end{equation}
  We have $Z_M,Z_M'\subseteq Z$ and
  $Y_M,Y_M'\subseteq Y$ (see the definition of the projection in
  \eqref{eq:proj}). Thus 
  \begin{equation}
    \label{eq:8}
    Z_M\cap
    Y_M'=\emptyset
   \quad\text{and}\quad
   Y_M\cap Z_M'=\emptyset,
  \end{equation}
  because $Y\cap Z=\emptyset$.

  \begin{claim}
    There is no $j$ such that $w_j\in Y$ and $w_j'\in Z$ or vice
    versa.

    \proof Suppose for contradiction that $w_1\in Y$ and $w_1'\in
    Z$.
    Let $C,C'$ be the connected components of $G\setminus R$ such that
    $w_1\in V(C)$ and $w_1'\in V(C')$ (possibly, $C=C'$). Then
    $N(C)=N(C')=\{v_1,v_2,v_3\}$. As $G$ is 3-connected, there are
    internally disjoint paths $P_1,P_2,P_3$ from $w_1$ to
    $v_1,v_2,v_3$, respectively, and internally disjoint paths
    $P_1',P_2',P_3'$ from $w_1'$ to $v_1,v_2,v_3$, respectively.  The
    vertex sets of all these paths are contained in
    $V(C)\cup V(C')\cup\{v_1,v_2,v_3\}$, and as $S$ separates $w_1\in Y$ from
    $w_1'\in Z$, we have $S\subseteq V(C)\cup V(C')\cup\{v_1,v_2,v_3\}$. This
    implies $S_M,S_M'\subseteq\{v_1,v_2,v_3\}$.

    If $v_4\in Z$ then 
    $v_4,w_2',w_3',w_4'\in Z_M'$. Thus $|Z_{M'}|\ge 4>|Y_{M'}|$, and this implies
    $(Y_M',S_M',Z_M')\in\CHT$, contradicting
    \eqref{eq:7}. Similarly, if $v_4\in Y$ then 
    $v_4,w_2,w_3,w_4\in Y_M$, and this implies
    $(Z_M,S_M,Y_M)\in\CHT$, contradicting
    \eqref{eq:7} again.
    \uend
  \end{claim}

  As $\hat H$ is quasi-4-connected, $(Y_M,S_M,Z_M)\in\CHT$ implies
  that either $Y_M=\emptyset$ or $Y_M=\{w_j\}$ for some $j$ and thus
  $|Z_M|\ge 4$. Similarly, $(Z_M',S_M',Y_M')\in\CHT$ implies
  $Z_M'=\emptyset$ or $Z_M'=\{w_j'\}$ for some $j$ and $|Y_M'|\ge 4$.

  \begin{claim}[resume]
    $Y_M\neq\emptyset$ and $Z_M'\neq\emptyset$.

    \proof
    Suppose for contradiction that $Y_M=\emptyset$. Then $|Z_M|\ge
    5$. As $|Y_M'|\ge 4$ and $|\THF|=8$, it follows that either $Z_M\cap
    Y_M'\neq\emptyset$, which contradicts \eqref{eq:8}, or there is
    a $j$ such that $w_j\in Z_M\subseteq Z$ and
    $w_j'\in Y_M'\subseteq Y$, which contradicts Claim~1. 
    \uend
  \end{claim}

  Thus $Y_M=\{w_j\}$ and $Z_M'=\{w_j'\}$, where $j\neq j'$ by
  \eqref{eq:8} and Claim~1. Without loss of generality we assume that
  $Y_M=\{w_1\}$ and $Z_M'=\{w_2'\}$. Then
  $S_M=N^{\THF}(w_1)=\{v_1,v_2,v_3\}$ and
  $S_M'=N^{\THF}(w_2)=\{v_1,v_2,v_4\}$. Hence $w_3\in Z_M\subseteq Z$
  and $w_3'\in Y_M'\subseteq Y$, and this contradicts Claim~1.
\end{proof}

If $\torso GR$ is non-exceptional, then there is no need for
non-exceptional extensions, and we can directly work with liftings of
the unique $\torso GR$-tangle of order $4$. Surprisingly, it is much
harder to prove the uniqueness of the lifting in this case. If
$H:=\torso GR$ is non-exceptional and $\CM$ is a faithful image of $H$
in $G$, then $\CT(H,\CM)$ is the lifting of the unique $H$-tangle of
order $4$ to $G$ with respect to $\CM$.

\begin{lem}\label{lem:q4r2}
  Let $R$ be a non-exceptional quasi-4-connected region of $G$ such that $H:=\torso GR$ is non-exceptional. Then for all faithful
  models $\CM, \CN$ of $H$ in $G$ we have $\CT(H,\CM)=\CT(H,\CN)$.
\end{lem}

\begin{proof}
   Let $\CHT$ be the unique $H$-tangle of order $4$. Let $(M_w)_{w\in R}$ and
  $(N_w)_{w\in R}$ be the branch sets of faith models $\CM$ and $\CN$
  of $G$ in $G$. Suppose for contradiction that 
$\CT(H,\CM)\neq\CT(H,\CN)$. Then there is a separation
  $(Y,S,Z)\in\Sep_{<4}(G)$ such that 
  \begin{equation*}
    (Y,S,Z)\in \CT(H,\CM)
  \quad\text{and}\quad
  (Z,S,Y)\in \CT(H,\CN).
  \end{equation*}
  Let $(Y_M,S_M,Z_M):=\pi_{\CM}(Y,S,Z)$ and
  $(Y_N,S_N,Z_N):=\pi_{\CN}(Y,S,Z)$. Then 
  \begin{equation}
    \label{eq:9}
    (Y_M,S_M,Z_M)\in\CHT
    \quad\text{and}\quad
    (Z_N,S_N,Y_N)\in\CHT.
  \end{equation}
It
  follows from the definition of the projections in \eqref{eq:proj} and
  the assumption that the models $\CM$ and $\CN$ be faithful that
  $Y_M\subseteq Y\cap {R}\subseteq Y_M\cup S_M$ and
  $Z_M\subseteq Z\cap {R}\subseteq Z_M\cup S_M$ and
  $Y_N\subseteq Y\cap {R}\subseteq Y_N\cup S_N$ and
  $Z_N\subseteq Z\cap {R}\subseteq Z_N\cup S_N$. Hence 
  \begin{equation}
    \label{eq:10}
    Y_M\cap
    Z_N=\emptyset
    \quad\text{and}\quad
    Z_M\cap Y_N=\emptyset.
  \end{equation}
  (see Figure~\ref{fig:q4r1}(a)).

  \begin{figure}
    \centering
    \begin{tikzpicture}
  \begin{scope}
        \fill[black!20] (-1.75,-1.75) rectangle (-0.25,0.25);
        \fill[black!20] (-0.25,-1.75) rectangle (0.25,1.75);
        \fill[black!20] (0.25,-0.25) rectangle (1.75,1.75);

      \draw[thick] (-1.75,-1.75) rectangle (1.75,1.75)
                   (-0.25,-1.75) rectangle (0.25,1.75) 
                   (-1.75,-0.25) rectangle (1.75,0.25);
      
      \path (-2,1) node {$Z_M$}             
            (-2,0) node {$S_M$}        
            (-2,-1) node {$Y_M$}  
            (-1,2) node {$Y_N$}
            (0,2) node {$S_N$}
            (1,2) node {$Z_N$}
      ;

      \path (0,-2.3) node {(a)};
      \end{scope}

  \begin{scope}[xshift=4.5cm]
        \fill[black!20] (-1.75,-1.75) rectangle (-0.25,0.25);
        \fill[black!20] (-0.25,0.25) rectangle (1.75,1.75);

      \draw[thick] (-1.75,-1.75) rectangle (1.75,1.75)
                   (-0.25,-1.75) rectangle (0.25,1.75) 
                   (-1.75,-0.25) rectangle (1.75,0.25);
      
      \path (-2,1) node {$Z_M$}             
            (-2,0) node {$S_M$}        
            (-2,-1) node {$Y_M$}  
            (-1,2) node {$Y_N$}
            (0,2) node {$S_N$}
            (1,2) node {$Z_N$}
      ;

      \draw[fill=black] (-1.3,0) circle (2pt) node[left] {$y_1$};
      \draw[fill=black] (-1,0) circle (2pt) node[fill=black!20,inner sep=1pt,below=4pt] {$y_2$};
      \draw[fill=black] (-0.7,0) circle (2pt) node[right] {$y_3$};
      \draw[fill=black] (0,1.3) circle (2pt) node[above] {$z_1$};
      \draw[fill=black] (0,1) circle (2pt) node[fill=black!20,inner sep=1pt,right=4pt] {$z_2$};
      \draw[fill=black] (0,0.7) circle (2pt) node[below] {$z_3$};
      
      \path (0,-2.3) node {(b)};
      \end{scope}

  \begin{scope}[xshift=9cm]
        \fill[black!20] (-1.75,-1.75) rectangle (-0.25,0.25);
        \fill[black!20] (-0.25,-1.75) rectangle (0.25,-0.25);
        \fill[black!20] (-0.25,0.25) rectangle (0.25,1.75);

      \draw[thick] (-1.75,-1.75) rectangle (1.75,1.75)
                   (-0.25,-1.75) rectangle (0.25,1.75) 
                   (-1.75,-0.25) rectangle (1.75,0.25);
      
      \path (-2,1) node {$Z_M$}             
            (-2,0) node {$S_M$}        
            (-2,-1) node {$Y_M$}  
            (-1,2) node {$Y_N$}
            (0,2) node {$S_N$}
            (1,2) node {$Z_N$}
      ;

      \draw[fill=black] (-1.3,0) circle (2pt) node[left] {$y_1$};
      \draw[fill=black] (-1,0) circle (2pt) node[fill=black!20,inner sep=1pt,below=4pt] {$y_2$};
      \draw[fill=black] (-0.7,0) circle (2pt) node[right] {$y_3$};
      \draw[fill=black] (0,1.25) circle (2pt) node[above] {$z_1$};
      \draw[fill=black] (0,0.75) circle (2pt) node[below] {$z_2$};
      
      \path (0,-2.3) node {(c)};
      \end{scope}

 \begin{scope}[yshift=-5cm]
        \fill[black!20] (-1.75,-1.75) rectangle (-0.25,0.25);
        \fill[black!20] (-0.25,-0.25) rectangle (0.25,1.75);
        \fill[black!20] (0.25,0.25) rectangle (1.75,1.75);

      \draw[thick] (-1.75,-1.75) rectangle (1.75,1.75)
                   (-0.25,-1.75) rectangle (0.25,1.75) 
                   (-1.75,-0.25) rectangle (1.75,0.25);
      
      \path (-2,1) node {$Z_M$}             
            (-2,0) node {$S_M$}        
            (-2,-1) node {$Y_M$}  
            (-1,2) node {$Y_N$}
            (0,2) node {$S_N$}
            (1,2) node {$Z_N$}
      ;

      \draw[fill=black] (-1.25,0) circle (2pt) node[left] {$y_1$};
      \draw[fill=black] (-0.75,0) circle (2pt) node[right] {$y_2$};
      \draw[fill=black] (0,1.25) circle (2pt) node[above] {$z_1$};
      \draw[fill=black] (0,0.75) circle (2pt) node[below] {$z_2$};
     \draw[fill=black] (-0.05,0.05) circle (2pt) node[below right=-2pt] {$x$};
      
      \path (0,-2.3) node {(d)};
      \end{scope}

 \begin{scope}[xshift=4.5cm,yshift=-5cm]
        \fill[black!20] (-1.75,-0.25) rectangle (-0.25,0.25);
        \fill[black!20] (-0.25,-1.75) rectangle (0.25,-0.25);
        \fill[black!20] (-0.25,0.25) rectangle (0.25,1.75);
        \fill[black!20] (0.25,-0.25) rectangle (1.75,0.25);

      \draw[thick] (-1.75,-1.75) rectangle (1.75,1.75)
                   (-0.25,-1.75) rectangle (0.25,1.75) 
                   (-1.75,-0.25) rectangle (1.75,0.25);
      
      \path (-2,1) node {$Z_M$}             
            (-2,0) node {$S_M$}        
            (-2,-1) node {$Y_M$}  
            (-1,2) node {$Y_N$}
            (0,2) node {$S_N$}
            (1,2) node {$Z_N$}
      ;

      \draw[fill=black] (-1.25,0) circle (2pt) node[left] {$y_1$};
      \draw[fill=black] (-0.75,0) circle (2pt) node[right] {$y_2$};
      \draw[fill=black] (0,1.25) circle (2pt) node[above] {$z_1$};
      \draw[fill=black] (0,0.75) circle (2pt) node[below] {$z_2$};
      
      \path (0,-2.3) node {(e)};
      \end{scope}

\end{tikzpicture}
    \caption{Proof of Lemma~\ref{lem:q4r2}}
    \label{fig:q4r1}
  \end{figure}

  By
  \eqref{eq:9} and Lemma~\ref{lem:t2} we have
  $|Z_M\cup S_M|\ge 5$ and $|Y_N\cup S_N|\ge 5$ and thus
  $|Z_M|,|Y_N|\ge 2$. As ${H}$ is quasi-4-connected, it follows
  that
  \begin{equation}
    \label{eq:11}
    |Y_M|\le 1
    \quad\text{and}\quad
    |Z_N|\le 1.
  \end{equation}

  \begin{claim}
  \begin{equation}
    \label{eq:12}
    |(S_M\cap Y_N)\cup (Z_M\cap S_N)|\ge 4.
  \end{equation}
  \proof
  Let $X:=(S_M\cap Y_N)\cup (Z_M\cap S_N)$, and suppose for
  contradiction that $|X|\le 3$. Then
  $(\emptyset,X,{R}\setminus X)\in\CHT$. Hence by \eqref{eq:9}
  and \ref{li:t2}, either $Z_M\cap
  Y_N\cap {R}\setminus X\neq\emptyset$ or there is an edge that has
  an endvertex in $Z_M$, $Y_N$, and ${R}\setminus X$. 
However, it
  follows from \eqref{eq:10} that neither is the case.
  \uend
  \end{claim}

  Without loss of generality we assume
  \begin{equation}
    \label{eq:13}
    |(S_M\cap Y_N)|\ge|(Z_M\cap S_N)|
  \end{equation}

  Let us call an edge $yz\in E(H)$ with $y\in Y$ and $u\in Z$ a
  \emph{yz-edge} and a connected component $C$ of $G\setminus R$ with
  $N(C)\cap Y\neq\emptyset$ and $N(C)\cap Z\neq\emptyset$ a
  \emph{yz-component}. If $yz$ is a yz-edge, we have $yz\not\in
  E(G)$.
  Thus there must be a yz-component $C$ such that $y,z\in N(C)$. If
  this is the case, we say that the yz-component $C$ \emph{covers} the
  edge $yz$. Note that every yz-component $C$ has a nonempty
  intersection with $S$, because if $y\in N(C)\cap Y$ and
  $z\in N(C)\cap Z$ then there is a path from $y$ to $z$ with all
  internal vertices in $C$, and this path must have a nonempty
  intersection with $S$. This means that there are at most three
  yz-components. It follows from \ref{li:q3} that each yz-component
  covers at most two yz-edges, and if it covers two edges, they have
  one endvertex in common.

  \begin{cs}
    \case1
    $|S_M\cap Y_N|=3$.\\
    Then $S_M\subseteq Y_N$ and thus $S_M\cap S_N=\emptyset$. Suppose
    that $S_M=y_1,y_2,y_3$. As
    $|Y_M\cap S_N|\le 1$ by \eqref{eq:11}, we have $|Z_M\cap S_N|\ge
    2$.
    \begin{cs}
      \case{1a}
      $|Z_M\cap S_N|=3$.\\
      Then $S_N\subseteq Z_M$. Suppose that $S_N=\{z_1,z_2,z_3\}$ (see
      Figure~\ref{fig:q4r1}(b)).
      Whenever $y_iz_j\in E(H)$, it is a
      yz-edge, and hence there is a yz-component $C_{ij}$ that
      covers it.

      \begin{claim}[resume]
        There is a perfect matching between
        $\{y_1,y_2,y_3\}$ and $\{z_1,z_2,z_3\}$ in ${H}$.

        \proof We first note that every $y_i$ has at least one $z_j$ as a neighbour,
        because if, say, $y_1$ has no neighbour among the $z_j$s,
          then $\{y_2,y_3\}$ is a separator of $H$. Similarly, every
          $z_j$ has a neighbour among the $y_i$s. 

          Now let $Y\subseteq \{y_1,y_2,y_3\}$, and let
          $Z:=N^{{H}}(Y)\cap\{z_1,z_2,z_3\}$ be the set of neighbours of
          $Y$. We shall prove that $|Y|\le|Z|$. Then the claim follows
          from Halls's Marriage Theorem.

          If $|Y|=1$ we have $|Z|\ge 1$, because every $y_i$ has a
          neighbour among the $z_j$. If $|Y|=3$ we have $|Z|=3$,
          because if there is a $z\in \{z_1,z_2,z_3\}\setminus Z$ this
          $z$ has no neighbour among the $y_i$s. Suppose that
          $|Y|=2$, and let $y$ be the unique element of
          $\{y_1,y_2,y_3\}\setminus Y$. If
          $Z\neq\{z_1,z_2,z_3\}$, then $Z\cup\{y\}$ is a separator of
          $H$, and this implies $|Z|\ge 2$.
          \uend
        \end{claim}

        Without loss of generality we assume that
        $y_1z_1,y_2z_2,y_3z_3\in E(H)$. It follows from
        \ref{li:q3} that the yz-components
        $C_{11},C_{22},C_{33}$ covering these yz-edges are distinct. Thus these are the only
        yz-components. Let $s_i\in S\cap V(C_{ii})$.

        \begin{claim}[resume]\label{cl:q4r20}
          There is no $y_i$ such that $z_1,z_2,z_3\in N^H(y_i)$ and
          not $z_j$ such that $y_1,y_2,y_3\in N^H(z_j)$. 

          \proof
          Suppose for contradiction $z_1,z_2,z_3\in N^H(y_1)$. 

          If $C_{11}=C_{12}$ then $C_{11}$ covers the edges $y_1z_1$
          and $y_1z_2$, but not $y_1z_3$. Hence $C_{33}=C_{13}$. It
          follows $z_2\not\in N(C_{11})\cup N(C_{33})$.

          Now we have to analyse the models $\CM,\CN$. As $y_2\in
          S_M$, we must have $s_2\in M_{y_2}$. As $y_3\not\in N(C_{11})$, we have
          $s_3\in M_{y_3}$. But then the edge $y_1z_3$ cannot be
          realised in $\CM$, because $C_{33}=C_{13}$ is the only
          component that covers the edge, and $s_3$ separates $y_1$
          from $z_3$ in $C_{33}$.

          The case $C_{11}=C_{13}$ is symmetric.

          So suppose that $C_{11}\neq C_{12}, C_{13}$. Then we have $C_{22}=C_{12}$
          and $C_{33}=C_{13}$, because we need to cover the edges
          $y_1z_2$ and $y_1z_3$, and we cannot have $C_{22}=C_{13}$ or
          $C_{33}=C_{12}$ by \ref{li:q3}.

          Without loss of generality we may assume that
          $y_2\not\in N(C_{11})$ (the other case $y_3\not\in
          N(C_{11})$ is symmetric). As we also have
          $y_2\not\in N(C_{33})$, we must have $s_2\in
          M_{y_2}$.
          This implies that the edge $y_1z_2$ cannot be
          realised in $\CM$, because $s_2$ separates $y_1$ from $z_2$
          in $C_{12}$.
          \uend
        \end{claim}

        By symmetry, we may assume that $C_{11}=C_{12}$. Suppose that
        $C_{22}=C_{21}$. Then if $C_{33}=C_{13}$, we have
        $z_1,z_2,z_3\in N^H(y_1)$, which contradicts
        Claim~\ref{cl:q4r20}. Similarly, if $C_{33}=C_{23}$, we have
        $z_1,z_2,z_3\in N^H(y_2)$. If $C_{33}=C_{31}$, we have
        $y_1,y_2,y_3\in N^H(z_1)$, and if $C_{33}=C_{32}$, we have
        $y_1,y_2,y_3\in N^H(z_2)$. All this contradicts
        Claim~\ref{cl:q4r20}.

        Suppose next that $C_{22}=C_{32}$. Then $y_1,y_2,y_3\in
        N^H(z_2)$, which again contradicts
        Claim~\ref{cl:q4r20}.

        So we must have $C_{22}=C_{23}$. By symmetry, this implies
        $C_{33}=C_{31}$ (just as $C_{11}=C_{12}$ implies
        $C_{22}=C_{23}$).

        Then 
        \begin{align*}
          N^H(C_{11})&=\{y_1,z_1,z_2\},\\
          N^H(C_{22})&=\{y_2,z_2,z_3\},\\
          N^H(C_{11})&=\{y_3,z_3,z_1\}.
        \end{align*}
        Looking at the model $\CN$, we have either $s_1\in N_{z_1}$ or
        $s_3\in N_{z_1}$. If $s_1\in N_{z_1}$ then the edge $y_1z_2$
        cannot be realised in the model $\CN$. Thus $s_3\in
        N_{z_1}$.
        But then the edge $y_3z_3$ cannot be realised in the model
        $\CN$.  Either way we have a contradiction.  

        \case{1b}
        $|Z_M\cap S_N|=2$.\\
        Then $Z_N=\emptyset$, because otherwise ${H}$ is not
        3-connected. Let $Z_M\cap S_N=\{z_1,z_2\}$ (see
        Figure~\ref{fig:q4r1}(c)). As $|Y_M|\le 1$, we have $|R|\le 6$. 

        \begin{claim}[resume]
          $N^H(z_i)\supseteq\{y_1,y_2,y_3\}$ for $i=1,2$.
          
          \proof
          Suppose for contradiction that $y_3\not\in N^H(z_1)$. Then
          $N^H(z_1)=\{y_1,y_2,z_2\}$, and the mapping $\pi$ defined by
          $\pi(y_i):=v_i$ for $i=1,2,3$, $\pi(z_2):=v_4$,
          $\pi(z_1):=w_2$, and if there is a vertex $x\in Y_M$,
          $\pi(x):=w_1$, is an embedding of $H$ into $\THT$. Thus
          $H$ is exceptional, which is a contradiction.
          \uend
        \end{claim}

        The six edges $y_iz_j$ are yz-edges. Thus each edge $y_iz_j$
        needs to be covered by a yz-component $C_{ij}$. As there are
        six yz-edges and at most three yz-components and each
        yz-component covers at most two yz-edges, each yz-component
        must cover exactly two of the yz-edges $y_iz_j$.

        \begin{claim}[resume]
          There is an $i$ such that $C_{i1}=C_{i2}$.

          \proof
          Suppose not. Then $C_{11}\neq C_{12}$ and thus either
          $C_{11}=C_{21}$ or $C_{11}=C_{31}$. By symmetry, we may
          assume $C_{11}=C_{21}$. Then the four edges
          $y_1z_2,y_2z_2,y_3z_1,y_3z_2$ must be covered by the
          remaining two yz-components. We either have $C_{12}=C_{22}$
          or $C_{12}=C_{32}$. If $C_{12}=C_{32}$, the two edges
          $y_2z_2$ and $y_3z_1$ must be covered by the same
          yz-component, which is impossible. Hence $C_{12}=C_{22}$ and
          thus $C_{31}=C_{32}$. This proves the claim for $i=3$.
          \uend
        \end{claim}

        By symmetry, we may assume that $C_{31}=C_{32}$. Then the four
        edges $y_1z_1$, $y_1z_2$, $y_2z_1$, $y_2z_2$ must be covered by the
        remaining two yz-components. Suppose first
        that $C_{11}=C_{12}$. Then $C_{22}=C_{21}$, and we have
        $N(C_{11})=\{y_1,z_1,z_2\}$ and $N(C_{22})=\{y_2,z_1,z_2\}$. 
        Let $s_i$ be the unique element of $S\cap V(C_{ii})$.

        We analyse the model $\CN$. If $s_1\in N_{z_1}$, then we cannot realise
        the edge $z_2y_1$ in the model $\CN$, because $s_1$ separates
        $y_1$ and $z_2$ in $C_{11}$ and $y_1\not\in N(C_{22})$.
        Similarly, if $s_2\in N_{z_1}$, then we cannot realise the
        edge $z_2y_2$ in the model $\CN$, because $s_2$ separates
        $y_2$ and $z_2$ in $C_{22}$ and $y_2\not\in N(C_{11})$.

        If $C_{11}=C_{21}$ and $C_{22}=C_{12}$, then we can argue
        similarly with the model $\CM$.
    \end{cs}

    \case2
    $|S_M\cap Y_N|=2$.\\
    Then it follows from Claim~1 and \eqref{eq:13} that $|Z_M\cap
    S_N|=2$. Let $S_M\cap Y_N=\{y_1,y_2\}$ and $Z_M\cap
    S_N=\{z_1,z_2\}$. 
    \begin{cs}
      \case{2a}
      $S_M\cap S_N\neq\emptyset$.\\
      Then $|S_M\cap S_N|=1$. Let $x$ be the unique vertex in
      $S_M\cap S_N$ (see Figure~\ref{fig:q4r1}(d)). Then $|R|\le
      7$.
      If there is a vertex in $Y_M$, we call it $y_3$, and if there is
      a vertex in $Z_N$ we call it $z_3$.

        \begin{claim}[resume]
          $y_iz_j\in E({H})$ for $i,j=1,2$.

        \proof Suppose for contradiction that that
        $y_1z_1\not\in E(H)$. 
        Then $\{x,y_2,z_2\}$ is a separator of $H$ that separates
        $y_1$ from $z_1$. If both $y_3$ and $z_3$ exists, it even
        separates $\{y_1,y_3\}$ from $\{z_1,z_3\}$, which contradicts
        $H$ being quasi-4-connected. 

        Hence without loss of generality we may assume that $y_3$ does
        not exist, that is, $Y_M=\emptyset$. Then the following
        mapping $\pi$ is an embedding of $H$ into the graph $\THT$:
      \[
      \begin{array}{c|c|c|c|c|c|c}
         u      & x  &y_1&y_2&z_1&z_2&z_3\\
         \hline
         \pi(u)&v_1&w_1&v_2&v_4&v_3&w_3.
       \end{array}
       \]
       Hence $H$ is exceptional, which is a contradiction.
        \uend
      \end{claim}

      For all $i,j$ the edge $y_iz_j\in E(H)$ is a yz-edge. Hence
      there is a yz-component $C_{ij}$ covering it. As one
      yz-component covers at most two yz-edges, we need at least
      two yz-components to cover the four edges $y_iz_j$.
      Furthermore, the
      yz-components $C_{11}$, $C_{22}$ and the
      yz-components $C_{12}$, $C_{21}$ are distinct. Let
      $s_1\in S\cap V(C_{11})$ and $s_2\in S\cap C_{22}$.

      \begin{claim}[resume]
        $x\not\in S$.

        \proof Suppose for contradiction that $x\in S$. Then
        $|S\setminus R|\le 2$ and hence there are at most two
        yz-components. We can argue exactly as in Case~1b. Let me
        repeat the argument for the reader's convenience.

        Suppose first that $C_{11}=C_{12}$. Then
        $C_{22}=C_{21}$, and we have $N(C_{11})=\{y_1,z_1,z_2\}$ and
        $N(C_2)=\{y_2,z_1,z_2\}$. We analyse the model $\CN$. If $s_1\in N_{z_1}$, then we cannot realise
      the edge $z_2y_1$ in the model $\CN$, because $s_1$ separates
      $y_1$ and $z_2$ in $C_{11}$ and $y_1\not\in
      N(C_{22})$.
      Similarly, if $s_2\in N_{z_1}$, then we cannot realise the edge
      $z_2y_2$ in the model $\CN$, because $s_2$ separates $y_2$ and
      $z_2$ in $C_{22}$ and $y_2\not\in N(C_{11})$.

      It remains to consider the case $C_{11}=C_{12}$. But this case
      is symmetric, and we argue with the roles of $\CM$ and $\CN$
      swapped.
      \uend
      \end{claim}

      Thus either $x\in Y$ or $x\in Z$. By symmetry, we may assume
      that $x\in Y$. 

      \begin{claim}[resume]
        $Z_N=\emptyset$, that is, $z_3$ does not exist.

        \proof
        Suppose for contradiction that $z_3$ exists.  Then $xz_3$ is a
        yz-edge, and we
      need a yz-component $C\neq C_{11},C_{12},C_{21},C_{22}$ to cover
      it. We have $y_1,y_2\not\in N(C)$, and hence all edges $y_iz_j$
      must be covered by the components
      $C_{11},C_{12},C_{21},C_{22}$. Now we can argue as in the proof
      of Claim~7 to derive a contradiction.
      \end{claim}

       Now we are in the same situation
        as in Case~1b with $x$ playing the role of $y_3$, and we can
        argue exactly as we did there.

        \case{2b}
        $S_M\cap S_N=\emptyset$.\\
        Then $Y_M\cap Y_N=\emptyset$, because otherwise $Y_M\cap
        S_N=\emptyset$, and 
        $\{y_1,y_2\}$ separates $Y_M\cap Y_N$ from $\{z_1,z_2\}$,
        which contradicts ${H}$ being 3-connected. Similarly,
        $Z_M\cap Z_N=\emptyset$ (see
        Figure~\ref{fig:q4r1}(e)).

        If $Y_M=\emptyset$ then $Z_N=\emptyset$, because otherwise
        $\{z_1,z_2\}$ is a separator of $H$. It follows that $|R|=4$,
        which contradicts $H$ being non-exceptional. Thus $Y_M\neq\emptyset$ and, by
        symmetry, $Z_N\neq\emptyset$. Let $y_3$ be the unique vertex
        in $Y_M$ and $z_3$ the unique vertex in $Z_N$. Then $y_3\in
        Y_M\cap S_N$ and $N^H(y_3)=\{y_1,y_2,z_3\}$. Similarly, $z_3\in
        S_M\cap Z_N$ and $N^H(z_3)=\{z_1,z_2,y_3\}$.

        The yz-edge $y_3z_3$ must be covered by some yz-component $C$.

        If $y_iz_j\not\in E(H)$ for some $i,j\le 2$, then $H$ can be
        embedded into $\THT$ and is exceptional.
        Thus $y_iz_j\in E(H)$ for all $i,j\le 2$. The yz-component $C$
        covers none of these four yz-edges, and we need to cover them
        with the two remaining yz-components. We have seen (several
        times) that this is impossible.
        \qedhere
    \end{cs}
  \end{cs}
\end{proof}

\begin{theo}\label{theo:q4r}
  Let $R$ be a non-exceptional quasi-4-connected region of $G$. If
  $\torso GR$ is non-exceptional, we let $\CT_R:=\CT(\torso GR,\CM)$
  for some faithful model $\CM$ of $\torso GR$ in $G$, and if
  $\torso GR$ is non-exceptional, we let $\CT_{R}:=\CT(\hat H,\CM)$
  for some non-exceptional extension $\hat H$ of $R$ and some faithful
  model $\CM$ of $\hat H$ in $G$.

  Then the mapping $R\mapsto\CT_R$ is well-defined (that is, $\CT_R$ only
  depends on $R$ and not on the extension $\hat H$ or the model
  $\CM$), and $\CT_R$ is a $G$-tangle of order $4$.
\end{theo}

\begin{proof}
  If $\torso GR$ is non-exceptional, this follows from
  Lemma~\ref{lem:q4r2}. If $\torso GR$ is exceptional, it follows from
  Lemma~\ref{lem:q4r3}.
\end{proof}

This proves the first half of the Correspondence
Theorem~\ref{theo:corr}. The remainder of Section~\ref{sec:4t} is
devoted to the second half.

\subsection{Degenerate 3-Separations}
\label{sec:q3b}

We continue assuming that $G$ is a 3-connected graph. Let us call a
proper separation $(Y,S,Z)$ \emph{degenerate} if $|Y|=1$ and $S$ is an
independent set in $G$.

\begin{figure}
  \centering
  \begin{tikzpicture}
  [
  vertex/.style={draw,circle,fill=black,inner sep=0mm,minimum
    size=2mm},
  every edge/.style={draw,thick}
  ]

    \path
    (0,0,0) node[vertex] (v1) {}
    (2,0,0) node[vertex] (v2) {}
    (0,2,0) node[vertex] (v3) {}
    (0,0,2) node[vertex] (v4) {}
    (2,2,0) node[vertex] (v5) {}
    (2,0,2) node[vertex] (v6) {}
    (0,2,2) node[vertex] (v7) {}
    (2,2,2) node[vertex] (v8) {};

    \path 
    (v1) edge (v2) edge (v3) edge (v4)
    (v2) edge (v5) edge (v6)
    (v3) edge (v5) edge (v7)
    (v4) edge (v6) edge (v7)
    (v5) edge (v8)
    (v6) edge (v8)
    (v7) edge (v8)
    ;
  
\end{tikzpicture}
  \caption{The cube graph}
  \label{fig:cube}
\end{figure}
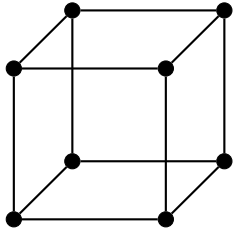

\begin{exa}
  Let $G$ be a hexagonal grid (see
  Figure~\ref{fig:hexgrids}) or a cube (see
  Figure~\ref{fig:cube}). Then all proper separations
  $(Y,S,Z)\in\Sep_{\le 3}(G)$ with $|Y|\le |Z|$ are degenerate.
\end{exa}

\begin{lem}\label{lem:degsep}
  Let $(Y,S,Z)\in\Sep_{=3}(G)$ be a non-degenerate proper
  separation. Then $\torso G{Z\cup S}$ is a faithful minor of $G$.
\end{lem}

\begin{proof}
  Suppose that $S=\{s_1,s_2,s_3\}$ and let $H:=\torso{G}{Z\cup S}$.
  Note that $E(H)\setminus E(G)\subseteq\{s_1s_2,s_1s_3,s_2s_3\}$. We
  shall define a faithful model
  $\big((M_v)_{v_\in V(H)}, (e_f)_{f\in E(H)}\big)$ of $H$ in $G$. For
  all $z\in Z$ we let $V(M_z):=\{z\}$. For all edges in
  $f\in E(G)\cap E(H)$ we let $e_f:=f$. It only remains to define the
  $V(M_{s_i})$ and $e_{s_is_j}$.

  Suppose first that $S$ is not independent. Say, $s_1s_2\in
  E(G)$. Let $C$ be a connected component of $G[Y]$. As $G$ is
  3-connected, we have $N(C)=S$. We let
  $V(M_{s_3}):=V(C)\cup\{s_3\}$, and for $i=1,2$ we let $e_{s_is_3}$
  be an edge from $V(C)$ to $s_i$.

  Suppose next that $S$ is independent and $G[Y]$ is not connected. Let $C_1$ and $C_2$ be
  two connected components of $G[Y]$. We let $V(M_{s_2}):=\{s_2\}$. We contract $C_1$ onto $s_1$ to
  create an edge from $s_1$ to $s_2$. That is, we let
  $V(M_{s_1}):=V(C_1)\cup\{s_1\}$, and we let $e_{s_1s_2}$ be an
  arbitrary edge from $V(C_1)$ to $s_2\in N(C_1)$. Then we contract
  $C_2$ onto $s_3$ to create edges from $s_3$ to $s_1$ and
  $s_2$. Formally, we let  $V(M_{s_3}):=V(C_3)\cup\{s_1\}$, and for
  $i=1,2$ we let
  $e_{s_is_3}$ be an
  arbitrary edge from $V(C_2)$ to $s_i\in N(C_2)$.

  Finally, suppose that $S$ is an independent set and $G[Y]$. Let $v\in Y$.  As $G$ is
  3-connected, there are internally disjoint paths $P_i$, for
  $i\in[3]$, from $v$ to $s_i$. At least one of these paths, say,
  $P_1$, can be chosen to have length at least $2$. To see this,
  suppose that $P_1,P_2,P_3$ have length $1$. Let $w\in N(v)\cap Y$;
  such a $w$ exists because $G[Y]$ is connected and $|Y|\ge 2$. Then there is a path $Q$
  from $w$ to $S$ in $G\setminus\{v\}$. Let $s_i$ be the endvertex of
  $Q$ in $S$. As for all $j\neq i$ the path $P_j$ only consist of a single
  edge, $Q$ and $P_j$ have an empty intersection. We can replace $P_i$
  by the path from $v$ to $w$ and then along $Q$ to $s_i$; this path
  has length at least $2$.

  In the following we assume without loss of generality that $P_1$ has
  length at least $2$. Then  $ V(P_1)\setminus\{v,s_1\}\neq\emptyset$. 
  Let $Q'$ be a path from
  $V(P_1)\setminus\{v,s_1\}$ to $\big(V(P_2)\cup
  V(P_3)\big)\setminus\{v\}$ in $G\setminus\{v,s_1\}$. Such a path
  exists because $G$ is 3-connected. Without loss of generality we may
  assume that the endvertex of $Q'$ is on $P_2$ and that $Q'$ has no
  internal vertex on $V(P_1)\cup V(P_2)\cup V(P_3)$. For $i=1,2$, let
  $w_i$ be the endvertex of $Q'$ on $P_i$, and let $P_i'$ be the segment of
  $P_i$ from  $s_i$ to $w_i$. We let $V(M_{s_1}):=V(P_1')\cup
  V(Q')\setminus\{w_2\}$ and $V(M_{s_2}):=V(P_2')$, and we let
  $e_{s_1s_2}$ be the edge of $Q'$ incident with $w_2$.
  We let $V(M_{s_3}):=V(P_3)\cup (V(P_1)\setminus V(P_1'))\cup
  (V(P_2)\setminus V(P_2'))$. For $i=1,2$, let $x_i$ be the neighbour
  of $w_i$ in $V(P_i)\setminus V(P_i')$, and let $e_{s_is_3}$ be the
  edge from $w_i$ to $x_i$.
\end{proof}

\begin{rem}\label{rem:degsep}
  The converse of Lemma~\ref{lem:degsep} holds as well:
  \emph{if $(Y,S,Z)\in\Sep_{\le 3}(G)$ is a proper separation such
    that $\torso G{Z\cup S}$ is a minor of $G$, then $(Y,S,Z)$ is non-degenerate.}

    To see this, suppose that $(Y,S,Z)\in\Sep_{\le 3}(G)$ is a proper
    degenerate separation and $H:=\torso G{S\cup Z}$ is a minor of
  $G$. Note that $|H|=|G|-1$ and
  $|E(H)|=|E(G)|$, because $S$ is an independent set in $G$. Let
  $\big((M_w)_{w\in V(H)},(e_f)_{f\in E(H)}\big)$ be a model of $H$ in
  $G$. Then $\{e_f\mid f\in E(H)\}=E(G)$, and this implies that
  $|M_w|=1$ for all $w\in V(H)$, because if $|M_w|\ge 2$ then at least
  one edge would appear in the connected graph $G[M_w]$. Thus
  $\sum_{w\in V(H)}|M_w|=|G|-1$, and hence there is a vertex
  $v\in V(G)\setminus\bigcup_{w\in V(H)}V(M_w)$. As $G$ is connected,
  $v$ is incident with at least one edge $e$, and this edge $e$ is not
  among the $e_f$ for $f\in E(H)$. This is a contradiction.
  \uend
\end{rem}

\subsection{Crossing Separations}
\label{sec:4t1}

\begin{figure}
  \centering
  \input{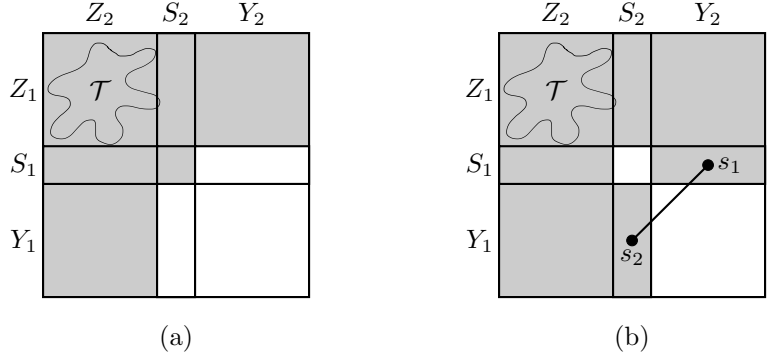}
  \caption{Orthogonal and crossing separations}
  \label{fig:orth}
\end{figure}

Let us call two separations $(Y_1,S_1,Z_1),(Y_2,S_2,Z_2)\in\Sep(G)$
\emph{orthogonal} if 
\[
(Y_1\cup S_1)\cap(Y_2\cup S_2)\subseteq S_1\cap S_2.
\]
(see Figure~\ref{fig:orth}(a)).  It is not hard to show that the minimal
separations of a tangle of order $3$ in a graph are mutually
orthogonal. This is the key to the proof of
Theorem~\ref{theo:tangles_vs_3cc}. The minimal separations of a tangle
of order $4$ are not necessarily orthogonal, but in this section, we
shall prove that they can only ``cross'' in a very restricted way.

We continue to assume that $G$ is a 3-connected graph and, in addition
make the following assumption, which will stay in place until the end
of Section~\ref{sec:4t}.

\begin{ass}\label{ass:4t2}
  $\CT$ is a $G$-tangle of order $4$.
\end{ass}

For all sets $S\subseteq V(G)$ of cardinality  $|S|\le 3$ we let
$\CZ(S):=\CZ_{\CT}(S)$ and $\CY(S):=\CY_{\CT}(S)$ (see \eqref{eq:3}
and \eqref{eq:4}).

\begin{lem}[Crossing Lemma]\label{lem:cross}
  Let $(Y_1,S_1,Z_1),(Y_2,S_2,Z_2)\in\CMT$ be distinct. Then
  either $(Y_1,S_1,Z_1)$ and $(Y_2,S_2,Z_2)$ are orthogonal or
  $Y_1\cap Y_2=\emptyset$ and $S_1\cap S_2=\emptyset$ and there is an edge $s_1s_2\in E(G)$ such
  that for $i=1,2$ we have $S_i\cap Y_{3-i}=\{s_i\}$ (see Figure~\ref{fig:orth}(b)).

  In the latter case, we call the edge $s_1s_2$ the \emph{crossedge}
  of $(Y_1,S_1,Z_1)$ and $(Y_2,S_2,Z_2)$.
\end{lem}

\begin{proof}
  We observe first that $Y_1$ and $Y_2$ are
  nonempty. If $Y_i=\emptyset$ then $S_i=\emptyset$ by the minimality
  of $(Y_i,S_i,Z_i)$ and thus
  $(Y_i,S_i,Y_i)=(\emptyset,\emptyset,V(G))\succeq(Y_{3-i},S_{3-i},Z_{3-i})$.
  Again by the minimality of $(Y_i,S_i,Z_i)$ this implies that the
  two separations are equal. This contradicts our assumption
  that they be distinct.

  By Corollary~\ref{cor:reed3} we have
  \begin{equation}
    \label{eq:14}
    |(S_1\cap Z_2)\cup(S_1\cap S_2)\cap(Z_1\cap S_2)|\ge 4
  \end{equation}
  and
  \begin{equation}
    \label{eq:15}
    |S_i\cap Z_{3-i}|\ge 1\hspace{2cm}\text{for }i=1,2.
  \end{equation}
  An easy calculation based on these inequalities and $|S_i|\le 3$
  shows that
  \begin{equation*}
    |(S_1\cap Y_2)\cup (S_1\cap S_2)\cup (Y_1\cap S_2)|\le 2.
  \end{equation*}
  As $G$ is 3-connected and $(S_1\cap Y_2)\cup (S_1\cap S_2)\cup (Y_1\cap S_2)$ separates $Y_1\cap Y_2$ from the nonempty set
  $S_1\cap Z_2$, it follows that
  \begin{equation*}
    Y_1\cap Y_2=\emptyset.
  \end{equation*}
  Thus if $S_1\cap Y_2=S_2\cap Y_1=\emptyset$, then 
  $(Y_1,S_1,Z_1)$ and $(Y_2,S_2,Z_2)$ are orthogonal.
  
  In the following, we assume without loss of generality that 
  \begin{equation}
    \label{eq:18}
    |S_1\cap
    Y_2|\ge 1.
  \end{equation}
  Suppose for contradiction that $Y_1\cap S_2=\emptyset$. Then
  $Y_1\subseteq Z_2$. As there is no edge from $Z_2$ to
  $Y_2$, we have $N(Y_1)\subseteq S_1\setminus Y_2$. Thus
  $(Y_1,S_1\setminus Y_2,Z_1\cup(S_1\cap Y_2))$ is a separation
  of order less than $3$, which contradicts $G$ being 3-connected. Thus
  \begin{equation}
    \label{eq:19}
    |S_2\cap
    Y_1|\ge 1.
  \end{equation}
  The only solution to the inequalities $|S_i|\le 3$ and \eqref{eq:14},
  \eqref{eq:18}, and \eqref{eq:19} is
  \begin{equation*}
    |S_i\cap Z_{3-i}|= 2
    \quad\text{and}\quad
     |S_i\cap Y_{3-i}|=1
   \qquad\text{for }i=1,2.
  \end{equation*}
  Let $s_i$ be the unique vertex in $S_i\cap Y_{3-i}$. We have
  $s_1s_2\in E(G)$, because otherwise $S_1\cap Z_2$ (and also $Z_1\cap
  S_2$) separates $s_1$ from $s_2$, which contradicts $G$ being 3-connected.
\end{proof}

We say that two separations
$(Y_1,S_1,Z_1),(Y_2,S_2,Z_2)\in\CMT$ \emph{cross} if they are not
orthogonal. We call $(Y,S,Z)\in\CMT$ \emph{crossed} if there is some
$(Y',S',Z')\in\CMT$ that crosses $(Y,S,Z)$.  We denote the set of all
non-degenerate separations $(Y,S,Z)\in\CMT$ by $\CNDT$.

By Corollary~\ref{cor:reed2}, the minimal separations
are determined by their separators: for every $3$-separator $S$ of
$G$, if $(Y,S,Z)\in\CMT$ for some $Y,Z\subseteq V(G)$, then
$Z=\CZ(S)$ and $Y=\CY(S)$. Therefore, we call $S$ $\CT$-minimal if
$(\CY(S),S,\CZ(S))\in\CMT$, and we denote the set of all $\CT$-minimal
3-separations by $\CMS$. We say that $S_1,S_2$ are \emph{orthogonal} (\emph{crossing})
if the corresponding separations $(\CY(S_i),S_i,\CZ(S_i))$ are orthogonal
(crossing, respectively). The \emph{crossedge} of two crossing
separators $S_1,S_2\in\CMS$ is the crossedge of the corresponding
separations.  We say that $S\in\CMS$ is \emph{crossed} if
$(\CY(S),S,\CZ(S))$ is crossed.
We call $S\in\CMS$ \emph{degenerate} if $(\CY(S),S,\CZ(S))$ is
degenerate and denote
the set of all non-degenerate 
$S\in\CMS$ by $\CNDS$.

\begin{lem}\label{lem:4t1a}
  Let $S\in\CMS$ be crossed. Then $\CY(S)$ is connected in $G$. 
\end{lem}

\begin{proof}
  Let $Y:=\CY(S)$ and $Z:=\CZ(S)$, and let $(Y',S',Z')\in\CT$ be a vertex
  separation that crosses $(Y,S,Z)$. Let $ss'$ with $s\in S$ and
  $s'\in S'$ be the crossedge. Then $S'\cap Y=\{s'\}$ and $S\cap
  Y'=\{s\}$. Let $s''\in S\cap Z'$.

  Suppose for contradiction that $G[Y]$ is not connected. Then there
  is a connected component $C$ of $G[Y]$ such that $s'\not\in V(C)$
  and hence $S'\cap V(C)=\emptyset$. As $G$ is 3-connected, we have
  $N(C)=S$. Thus there is a path from $s''\in Z'$ to $s\in Y'$ with
  all internal vertices in $V(C)$. This path has an empty intersection
  with $S'$, which contradicts $S'$ separating $Z'$ from $Y'$.
\end{proof}

\begin{lem}[Crossedge Independence Lemma]\label{lem:CI}
  Let $S,S_1,S_2\in\CMS$ be distinct such that both
  $S,S_1$ and $S,S_2$ cross, and let $e_1=s_1s_1',e_2=s_2s_2'$ with
  $s_i\in S$ and $s_i'\in S_i$ be the
  respective crossedges. 
  \begin{enumerate}
  \item $S$ is an independent set.
  \item If $S$ is degenerate, then $s_1,s_2,s_1'$ are mutually distinct and $s_1'=s_2'$.
  \item If $S$ is non-degenerate, then $s_1,s_2,s_1',s_2'$ are
    mutually distinct.%
  \end{enumerate}
\end{lem}

\begin{proof}
  Let
  $(Y,S,Z),(Y_i,S_i,Z_i)\in\CMT$ be the separations
  corresponding to our separators.  By the Crossing Lemma~\ref{lem:cross} we have
  \begin{gather}
    \label{eq:21}
    Y\cap Y_1=Y\cap
    Y_2=Y_1\cap Y_2=\emptyset\\
     \label{eq:22}
   S\cap Y_1=\{s_1\}\quad\text{and}\quad S\cap Y_2=\{s_2\}\\
     \label{eq:23}
   S_1\cap Y=\{s_1'\}\quad\text{and}\quad S_2\cap Y=\{s_2'\}\\
   \label{eq:24}
   S\cap S_1=S\cap S_2=\emptyset.
  \end{gather}
  Observe that $s_1\neq s_2$, because $s_1\in Y_1$ and $s_2\in Y_2$ and $Y_1\cap
    Y_2=\emptyset$ by \eqref{eq:21}. Furthermore, $s_1\neq s_2'$, because $s_1\in S$
    and $s_2'\in Y$ and, similarly, $s_2\neq s_1'$.

    To prove (1), note that $s_i\not\in S_1\cup S_2$ by
    \eqref{eq:24}. As $s_i\in Y_i$ by \eqref{eq:22}, we have
    $s_i\in Y_i\setminus(S_1\cup S_2)$. It follows
    that $s_1s_2\not\in E(G)$, because $Y_1$ and $Y_2$ are disjoint
    and $N(Y_i)\subseteq S_i$. Let $s$ be the unique element in
    $S\setminus\{s_1,s_2\}$. By
    \eqref{eq:22} and \eqref{eq:24} we have $s\in Z_1\cap
    Z_2$. As $S_i$ separates $Z_i$ from $Y_i$, we have
    $s_is\not\in E(G)$. Hence $S$ is an independent set.
 
    To prove (2), suppose that $S$ is degenerate. Then $|Y|=1$, and as
    $s_1',s_2'\in Y$ this implies $s_1'=s_2'$. 

    To prove (3), suppose that $S$ is non-degenerate. Then
    $|Y|\ge 2$, because $S$ is an independent set.
    We need to prove that 
    $
      s_1'\neq s_2'.
    $
    We have $Y\cap Y_1=\emptyset$ and $Y\cap
    S_1=\{s_1'\}$. Hence $Y\cap Z_1\neq\emptyset$. Let $C$ be a
    connected component of $G[Y\cap Z_1]$. As $G$ is 3-connected, we
    have $|N(C)|\ge 3$, and as 
    \[
    N(C)\subseteq N(Y\cap
    Z_1)\subseteq(S\cap Z_1)\cup (S\cap S_1)\cup(Y\cap
    S_1)=\{s_2,s,s_1'\},
    \]
    it follows that $N(C)=\{s_2,s,s_1'\}$. Hence there is a path
    from $s_2$ to $s$ with all internal
    vertices in $C$. As $S_2$ separates $s\in Z_2$ from
    $s_2\in Y_2$, it holds that $V(C)\cap S_2\neq\emptyset$. As
    $V(C)\subseteq Y$ and $Y\cap S_2=\{s_2'\}$, we have
    $s_2'\in V(C)$. Thus $s_2'\neq s_1'\in N(C)$.
\end{proof}

\begin{figure}
  \centering 
  \begin{tikzpicture}
  [
  current point is local=true,
  line width=0.3mm,
  vertex/.style={draw,circle,fill=black,inner sep=0mm,minimum
    size=1.5mm},
  tnode/.style={circle,fill=white,inner sep=0mm,minimum size=6mm,},
  every edge/.style={draw},
  ]
  
  \newcommand{\trup}{{+(90:3mm) node[vertex] {}-- +(210:3mm) node[vertex] {} -- +(330:3mm)
     node[vertex] {} -- cycle}}
  \newcommand{\trdn}{{+(270:3mm) node[vertex] {} -- +(30:3mm) node[vertex] {} -- +(150:3mm)
     node[vertex] {} -- cycle}}
  
 \newcommand{\rad}{10mm} 
 \newcommand{\grd}{17.3205mm} 

 \draw ++(90:10mm) {
    +(90:\rad) node[tnode] (u1) {} -- 
    +(150:\rad) node[tnode] (d1) {} -- 
    +(210:\rad) node[tnode] (u2) {} -- 
    +(270:\rad) node[tnode] (d2) {} -- 
    +(330:\rad) node[tnode] (u3) {} -- 
    +(30:\rad) node[tnode] (d3) {} -- 
    cycle
  }
  ++(240:\grd) {
    +(90:\rad) node[tnode] {} -- 
    +(150:\rad) node[tnode] (d4) {} -- 
    +(210:\rad) node[tnode] (u4) {} -- 
    +(270:\rad) node[tnode] (d5) {} -- 
    +(330:\rad) node[tnode] (u5) {} -- 
    +(30:\rad) node[tnode] {} -- 
    cycle
  }
  ++(0:\grd) {
    +(90:\rad) node[tnode] {} -- 
    +(150:\rad) node[tnode] {} -- 
    +(210:\rad) node[tnode] {} -- 
    +(270:\rad) node[tnode] (d6) {} -- 
    +(330:\rad) node[tnode] (u6) {} -- 
    +(30:\rad) node[tnode] (d7) {} -- 
    cycle
  }
  ++(60:\grd) { [rounded corners]
    +(90:\rad) node[tnode] (u7) {} -- 
    +(150:\rad) node[tnode] {} -- 
    +(210:\rad) node[tnode] {} -- 
    +(270:\rad) node[tnode] {} -- 
    +(330:\rad) node[tnode] {} -- 
    +(30:\rad) -- 
    cycle
  }
  ++(120:\grd) { [rounded corners]
    +(90:\rad) --
    +(150:\rad) node[tnode] (d8) {} -- 
    +(210:\rad) node[tnode] {} -- 
    +(270:\rad) node[tnode] {} -- 
    +(330:\rad) node[tnode] {} -- 
    +(30:\rad) --
    cycle
  }
  ++(180:\grd) { [rounded corners]
    +(90:\rad) -- 
    +(150:\rad) -- 
    +(210:\rad) node[tnode] (u8) {} -- 
    +(270:\rad) node[tnode] {} -- 
    +(330:\rad) node[tnode] {} -- 
    +(30:\rad) node[tnode] {} -- 
    cycle
  }
  ++(240:\grd) { [rounded corners]
    +(90:\rad) node[tnode] {} -- 
    +(150:\rad) -- 
    +(210:\rad) node[tnode] (u9) {} -- 
    +(270:\rad) node[tnode] {} -- 
    +(330:\rad) node[tnode] {} -- 
    +(30:\rad) node[tnode] {} -- 
    cycle
  }
  ++(240:\grd) { [rounded corners]
    +(90:\rad) node[tnode] {} -- 
    +(150:\rad) -- 
    +(210:\rad) -- 
    +(270:\rad) node[tnode] (d9) {} -- 
    +(330:\rad) node[tnode] {} -- 
    +(30:\rad) node[tnode] {} -- 
    cycle
  }
  ++(300:\grd) { [rounded corners]
    +(90:\rad) node[tnode] {} -- 
    +(150:\rad) node[tnode] {} -- 
    +(210:\rad) -- 
    +(270:\rad) -- 
    +(330:\rad) node[tnode] (u10) {} -- 
    +(30:\rad) node[tnode] {} -- 
    cycle
  }
  ++(0:\grd) { [rounded corners]
    +(90:\rad) node[tnode] {} -- 
    +(150:\rad) node[tnode] {} -- 
    +(210:\rad) node[tnode] {} -- 
    +(270:\rad) -- 
    +(330:\rad) node[tnode] (u11) {} -- 
    +(30:\rad) node[tnode] {} -- 
    cycle
  }
  ++(0:\grd) { [rounded corners]
    +(90:\rad) node[tnode] {} -- 
    +(150:\rad) node[tnode] {} -- 
    +(210:\rad) node[tnode] {} -- 
    +(270:\rad) -- 
    +(330:\rad) -- 
    +(30:\rad) node[tnode] (d10) {} -- 
    cycle
  }
  ++(60:\grd) { [rounded corners]
    +(90:\rad) node[tnode] (u12) {} -- 
    +(150:\rad) node[tnode] {} -- 
    +(210:\rad) node[tnode] {} -- 
    +(270:\rad) node[tnode] {} -- 
    +(330:\rad) -- 
    +(30:\rad) -- 
    cycle
  }
  ;

\foreach \n in {u1,u2,u3,u4,u5,u6,u7,u8,u9,u10,u11,u12}
\draw (\n) \trup; 

\foreach \n in {d1,d2,d3,d4,d5,d6,d7,d8,d9,d10}
\draw (\n) \trdn; 

\end{tikzpicture}
  \caption{Hexagonal grid with triangles}
  \label{fig:thex}
\end{figure}
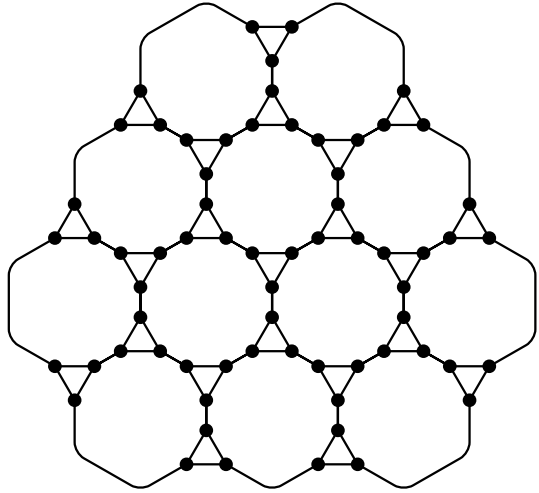

\begin{exa}
  The two cases of the Crossedge Independence Lemma are
  nicely illustrated by a hexagonal grid (see
  Figure~\ref{fig:hexgrids}), where all 3-separators are degenerate,
  and the graph in Figure~\ref{fig:thex}, where we have crossing
  non-degenerate 3-separators. In fact, a lot of my intuition draws from
  these two examples.
  \uend
\end{exa}

We call a crossedge of two 3-separators in $S_1,S_2\in\CMS$
\emph{non-degenerate} if both $S_1$ and $S_2$ are non-degenerate. Let us denote the set of crossedges of $\CT$ by
$\EC(\CT)$ and the subset of all non-degenerate crossedges by $\ECND(\CT)$

\begin{cor}\label{cor:CI}
  $\ECND$ is a matching. That is, distinct $e,e'\in\ECND(\CT)$ 
  have no endvertex in common.
\end{cor}

\begin{proof}
  Let $e=st,e'=s't'\in\ECND(\CT)$. Let $S_1,S_2\in\CNDS$ such that $e$
  is the crossedge of $S_1$ and $S_2$ and 
  $s\in S_1$ and $t\in S_2$,  and let $S_1',S_2'\in\CNDS$ such that  $e'$
  is the crossedge of $S_1'$ and $S_2'$ and
  $s'\in S_1'$ and $t'\in S_2'$. Suppose for contradiction that $t=t'$.
  Then $t\in \CY(S_1)\cap S_2\cap S_2'$, and thus
  both $S_2$ and $S_2'$ cross $S_1$, and $t$ is an endvertex of both
  crossedges. This contradicts the Crossedge Independence Lemma~\ref{lem:CI}.
\end{proof}

\subsection{Contracting a Crossedge}
\label{sec:cce}

In this section, we will study what happens if we contract a single
non-degenerate crossedge of $G$. We shall prove that the resulting
graph $G'$ is still 3-connected and has a tangle $\CT'$ of order $4$
that is ``induced'' by $\CT$. Technically, we will prove that $\CT$
is the lifting of $\CT'$ (see Lemma~\ref{lem:lift}). Moreover, the
minimal separations of $\CT'$ are the same as those of $\CT$, except
of course for the two separations whose crossedge we contract. Before
we go to the details, let us put this in a wider
perspective. Since the non-degenerate crossedges form a matching, the contraction
of one crossedge leaves the remaining ones intact, and we can
contract them all, one at a time. This leaves us with a graph and
tangle that has no crossedges, which means that all minimal
separations in the tangle are orthogonal. We will then remove the
``$Y$-parts'' of all non-degenerate minimal separations in the tangle
to obtain the quasi-4-connected region associated with our tangle.

The technically most difficult part is the contraction of one
crossedge, that is, the present section. The main insight is that if we a
have separator $S$ of order $4$ of our graph that contains both
endvertices of a crossedge, then each connected component of
$G\setminus S$ will only have one endvertex of the crossedge in its
neighbourhood. Thus there is a subset $S^\circ$ of $S$ of size $3$
that is still a separator of $G$. This will allow us to establish a correspondence
between the separations of order $3$ of the graphs $G$ and  $G'$ (obtained from $G$ by
contracting the crossedge). It still leaves us with many questions on
how exactly we match the separations, and in fact we will only get a
reasonable correspondence for the minimal separations in the tangle,
but this will be good enough to define a tangle of $G'$.

We still assume that $G$ is a 3-connected graph and $\CT$ a $G$-tangle
of order $4$, and we continue to use the notation and terminology of
the previous subsections. In addition, we make the following assumption.

\begin{ass}\label{ass:4t3}
  $e=s_1s_2$ is the crossedge of $S_1,S_2\in\CNDS$ and $G'$ is the
  graph obtained from $G$ by contracting the edge $e$ onto $s_1$. That
  is, $V(G'):=V(G)\setminus\{s_2\}$ and
  $E(G'):=\big(E(G)\setminus\{vs_2\mid v\in V(G)\}\big)\cup \{vs_1\mid
  vs_2\in E(G)\setminus\{e\}\}$.
\end{ass}

Let us fix some notation:
  for $i=1,2$, let $Y_i:=\CY(S_i)$ and $Z_i:=\CZ(S_i)$, and we assume
 $S_i=\{s_i,s_i',s_i''\}$ (see Figure~\ref{fig:4t2}).

\begin{figure}
  \centering
  \input{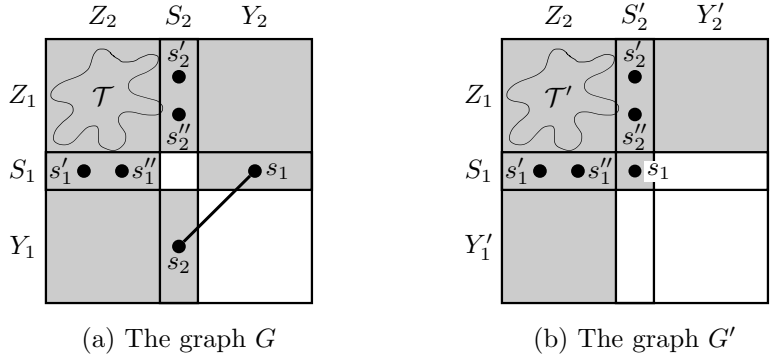}
  \caption{The setting of Section~\ref{sec:cce}}
  \label{fig:4t2}
\end{figure}

\begin{lem}\label{lem:4t15a}
  For all $S\in\CMS\setminus\{S_1,S_2\}$, 
  \begin{equation}
    \label{eq:25}
    s_1,s_2\not\in S\cup \CY(S),
  \end{equation}
  and $S$ is a separator of $G'$.
\end{lem}

\begin{proof}
  We have $s_j\in Y_{3-j}$ and $\CY(S)\cap Y_{3-j}=\emptyset$ by the
  Crossing Lemma~\ref{lem:cross}. Hence $s_j\not\in\CY(S)$. If
  $s_j\in S$ then $S$ crosses $S_{3-j}$ and the crossedge has an
  endvertex with the crossedge of $S_1$ and $S_{2}$ in common. As
  $S_1$ and $S_2$ are non-degenerate, this contradicts the Crossedge
  Independence Lemma~\ref{lem:CI}.

  To see that $S$ is a separator of $G'$, just note that
  $(\CY(S),S,\CZ(S)\setminus\{s_2\})$ is a proper separation of $G'$.
\end{proof}

\begin{lem}\label{lem:4t15}
  Let $S$ be a separator of $G$ such that
  $s_1,s_2\in S$ and
  $|S\setminus\{s_1,s_2\}|\le 2$. Then
  $|S\setminus\{s_1,s_2\}|= 2$, and there is a a separator
  $S^{\circ}\subseteq S$ such that $|S^{\circ}\cap\{s_1,s_2\}|=1$ and
  $S^{\circ}\setminus\{s_1,s_2\}=S\setminus\{s_1,s_2\}$ and there are at least two
  connected components of $G\setminus S^{\circ}$ that have a nonempty
  intersection with $V(G)\setminus S$.
\end{lem}

\begin{proof} 
  Choose a set $S^{\circ}\subseteq S$ such that there are at least two
  connected components of $G\setminus S^{\circ}$ that have a nonempty
  intersection with $V(G)\setminus S$, and subject to this condition
  $|S^{\circ}\cap\{s_1,s_2\}|$ is minimum. We shall prove that
  \begin{equation}
    \label{eq:26} |S^{\circ}\cap\{s_1,s_2\}|=1.
  \end{equation} 
  As $G$ is 3-connected and $S^{\circ}$ is a separator, this will imply
  $|S^{\circ}\setminus\{s_1,s_2\}|=|S\setminus\{s_1,s_2\}|=2$
  and thus $S^{\circ}\setminus\{s_1,s_2\}=S\setminus\{s_1,s_2\}$.

  Suppose for contradiction that $|S^{\circ}\cap\{s_1,s_2\}|= 2$, that is,
  $s_1,s_2\in S^{\circ}$.
  
  Let us call a connected component $C$ of $G\setminus S^{\circ}$
  \emph{relevant} if $V(C)\setminus S\neq\emptyset$.

  \begin{claim} 
    For every relevant component $C$,
    \[ 
    S^{\circ}\cap\{s_1,s_2\}\subseteq N(C).
    \]

    \proof Suppose that $C$ is a relevant component with
    $S^{\circ}\cap\{s_1,s_2\}\not\subseteq N(C)$. Let
    $S^{\bullet}:=N(C)$. Then $S^{\bullet}\subseteq S^{\circ}\subseteq S$, and
    $S^{\bullet}$ is a separator of $G$ with at least two relevant components. However,
    $|S^{\bullet}\cap\{s_1,s_2\}|<|S^{\circ}\cap\{s_1,s_2\}|$. This
    contradicts the minimality of $|S^{\circ}\cap\{s_1,s_2\}|$.  \uend
  \end{claim}

  \begin{claim}[resume] 
    Let $C$ be a relevant component. Then $V(C)\setminus(Y_1\cup
    Y_2)\neq\emptyset$. 

    \proof If
    $V(C)\subseteq Y_j$, then $S^{\circ}=N(C)\subseteq Y_j\cup
    S_j$. We have $s_{3-j}\in S^\circ$, and as $s_{3-j}$ is the only
    neighbour of $s_j$ in $S_j\cup Y_j$, it follows that $s_j\not\in
    N(C)=S^\circ$. This is a contradiction.

    Hence $V(C)\setminus Y_1\neq\emptyset$ and $V(C)\setminus
    Y_2\neq\emptyset$. Then  $V(C)\setminus(Y_1\cup
    Y_2)\neq\emptyset$, because $C$ is connected and the only edge from
    $Y_1$ to $Y_2$ is the crossedge $s_1s_2$, which is not in $ E(C)$
    by our assumption that $s_1,s_2\in S^{\circ}\subseteq V(G)\setminus V(C)$.
    \uend
  \end{claim}

  Let $C_1$ and $C_2$ be two distinct relevant components. We shall
  prove that there is path in $G\setminus S^{\circ}$ from a vertex in $C_1$
  to a vertex in $C_2$; this will be a
  contradiction. For $i=1,2$, let
  $v_i\in V(C_i)\setminus (Y_1\cup Y_2)$. As $N(C_i)=S^{\circ}$, there is
  path $P_{i,j}$ from $v_i$ to $s_j$ such that
  $P_{i,j}\setminus\{s_j\}\subseteq C_i$. This path has a nonempty
  intersection with $S_{3-j}$, because $v_i\not\in Y_{3-j}$ and
  $s_j\in Y_{3-j}$. As $s_{3-j}\not\in V(C)$, either
  $s_{3-j}'\in V(P_{i,j})$ or $s_{3-j}''\in V(P_{i,j})$ (recall that
  $S_{3-j}=\{s_{3-j},s'_{3-j},s''_{3-j}\}$) . Note
  furthermore that $\big(V(P_{i,j})\setminus\{s_j\}\big)\cap \big(V(P_{3-i,j'})\setminus\{s_{j'}\}\big)=\emptyset$, because
  $C_1\cap C_2=\emptyset$. Hence we may assume without loss of
  generality that $s_{3-j}'\in V(P_{1,j})$ and $s_{3-j}''\in V(P_{2,j})$. Thus
  $s_1',s_2'\in V(C_1)$ and $s_1'',s_2''\in V(C_2)$.

  For $j=1,2$, this implies that $s_j's_j''\not\in E(G)$. As there
  is no edge from $s_j\in Y_{3-j}$ to $s_j',s_j''\in Z_{3-j}$,
  it follows that $S_j$ is an independent set. Since $S_j$ is
  non-degenerate, we have $|Y_j|>1$. Furthermore, $s_{3-j}$ is
  the only neighbour of $s_j$ in $Y_j$, because $s_js_{3-j}$ is the
  crossedge of $S_j$ and $S_{3-j}$. This implies that
  $N(Y_j\setminus\{s_{3-j}\})\subseteq\{s_j',s_j'',s_{3-j}\}$, and in
  fact equality holds because $G$ is 3-connected and
  $Y_j\setminus\{s_{3-j}\}\neq\emptyset$. This implies that there is a
  path $Q_j$ from $s_j'$ to $s_j''$ with all internal vertices in
  $Y_j\setminus\{s_{3-j}\}$. As $s_j'\in V(C_1)$ and
  $s_j''\in V(C_2)$, this path has a nonempty intersection with
  $S^{\circ}$. Say,
  $t_j\in V(Q_j)\cap S^{\circ}\subseteq S^{\circ}\setminus \{s_1,s_2\}$.

  Thus $S^{\circ}\setminus\{s_1,s_2\}=\{t_1,t_2\}\subseteq Y_1\cup Y_2$,
  because $|S^{\circ}\setminus\{s_1,s_2\}|\le|S\setminus\{s_1,s_2\}|\le
  2$. It follows that 
  \[
  S^{\circ}\cap Z_1\cap Z_2=\emptyset.
  \]
  
  \begin{claim}[resume]
    $Z_1\cap Z_2=\emptyset$.

    \proof
    Suppose for contradiction that $Z_1\cap Z_2\neq\emptyset$. Let $C$
    be a connected component of $G[Z_1\cap Z_2]$. Then $N(C)\subseteq
    N(Z_1\cap Z_2)=\{s_1',s_1'',s_2',s_2''\}$ and $|N(C)|\ge
    3$. Hence for some $j\in\{1,2\}$ we have $s_j',s_j''\in
    N(C)$. Then there is a path from $s_j'\in V(C_1)$ to $s_j''\in
    V(C_2)$ with all internal vertices in $V(C)$, and as $V(C)\cap
    S^{\circ}=\emptyset$, this is a contradiction.
    \uend
  \end{claim}

  As $s_1',s_2'\in V(C_1)$ and $s_1'',s_2''\in V(C_2)$, it follows
  that the graph $G[(Z_1\cup S_1)\cap (Z_2\cup S_2)]$ has vertex set
  $\{s_1',s_1'',s_1'',s_2''\}$ and edges set
  $\{s_1's_2',s_1''s_2'' \}$.

    \begin{claim}[resume]
      $(Y_1',S_1',Z_1'):= (Y_1\cup\{s_1'\},\{s_1'',s_2',s_1\},Z_1\setminus\{s_2'\})\in\CT$.

      \proof
      The
      shape of the graph $G[(Z_1\cup S_1)\cap (Z_2\cup S_2)]$
      described before the claim implies that $(Y'_1,S'_1,Z_1')$ is a
      separation of $G$. Suppose
      for contradiction that it is
      not in $\CT$. Then $(Z_1',S_1',Y_1')\in\CT$.       We also have $(Y_1,S_1,Z_1)\in\CT$ and
      $(\emptyset,\{s_1',s_2'\},V(G)\setminus\{s_1',s_2'\})\in\CT$. As
      \[
      Y_1'\cap Z_1\cap V(G)\setminus\{s_1',s_2'\}=\big(Y_1\cup\{s_1'\}\big)\cap Z_1\cap
      \big(V(G)\setminus\{s_1',s_2'\}\big)\subseteq (Y_1\cup S_1)\cap Z_1=\emptyset,
      \]
      by \ref{li:t2} there must be an edge that has an endvertex in
      $Y_1'=Y_1\cup\{s_1'\}$ and $Z_1$ and
      $V(G)\setminus\{s_1',s_2'\}$. The only edge that has an
      endvertex in $Y_1\cup\{s_1'\}$ and $Z_1$ is $s_1's_2'$. However,
      this edge has no endvertex in $V(G)\setminus\{s_1',s_2'\}$. This
      is a contradiction.
      \uend
    \end{claim}

    As $(Y_1,S_1',Z_1')$ is
    strictly smaller than $(Y_1,S_1,Z_1)$ with respect to the order $\preceq$, this contradicts the
    minimality of $(Y_1,S_1,Z_1)$.
\end{proof}
  
\begin{cor}\label{cor:4t15}
  $G'$ is 3-connected. 
\end{cor}

\begin{proof}
  Let $S'$ be a separator of $G'$ of order at most $2$. If
  $s_1\not\in S'$, then $S'$ is a separator of $G$ of order at most
  $2$, which contradicts $G$ being 3-connected.

  If $s_1\in S'$, let $S:=S'\cup\{s_2\}$. Then $S$ is a
  separator of $G$ with $|S\setminus\{s_1,s_2\}|\le 1$. This
  contradicts the assertion of Lemma~\ref{lem:4t15} that $|S\setminus\{s_1,s_2\}|= 2$.
\end{proof}

For a separator $S$ of $G$ such that
$s_1,s_2\in S$ and
$|S\setminus\{s_1,s_2\}|\le 2$, we call a subset
$S^{\circ}\subseteq S$ such that $|S^{\circ}\cap\{s_1,s_2\}|=1$ and
$|S^{\circ}\setminus\{s_1,s_2\}|= 2$ and there are at least two
connected components of $G\setminus S^{\circ}$ that have a nonempty
intersection with $V(G)\setminus S$ an \emph{essential subseparator}
of $S$. Note that $S$ has at most two essential subseparators:
$S\setminus\{s_1\}$ and $S\setminus\{s_2\}$.

\begin{lem}\label{lem:4t16}
  Let $S$ be a separator of $G$ such that
$s_1,s_2\in S$ and
$|S\setminus\{s_1,s_2\}|\le 2$. Then there is a unique
  connected component $C^{\circ}$ of $G\setminus S$ such that
  $V(C^{\circ})\subseteq Z^\bullet$ for every separation $(Y^\bullet,S^\bullet,Z^\bullet)\in\CT$  with $S^\bullet\subseteq
  S$.

  Furthermore, for every essential subseparator $S^{\circ}$ of $S$ it holds
  that $V(C^{\circ})=\CZ(S^{\circ})\setminus\{s_1,s_2\}$.
\end{lem}

\begin{proof}
  Let $S^{\circ}$ be an essential subseparator of $S$, which exists by
  Lemma~\ref{lem:4t15}. Without loss of generality we may assume that
  $s_1\in S^{\circ}$.  Let $Z^{\circ}:=\CZ(S^{\circ})$ and $Y^{\circ}:=\CY(S^{\circ})$.  Moreover, let
  $(Y',S',Z')\in\CMT$ such that $(Y',S',Z')\preceq(Y^{\circ},S^{\circ},Z^{\circ})$. Then
  $Y^{\circ}\cup S^{\circ}\subseteq Y'\cup S'$ and thus $s_1\in Y'\cup S'$. By
  \eqref{eq:25}, there is a $j\in\{1,2\}$ such that $S'=S_j$.

  \begin{claim}
    $Z^{\circ}\setminus\{s_1,s_2\}=Z^{\circ}\setminus\{s_2\}$ is connected.
    
    \proof
    Note that the equality
    $Z^{\circ}\setminus\{s_1,s_2\}=Z^{\circ}\setminus\{s_2\}$ holds because
    $s_1\in S^{\circ}\subseteq V(G)\setminus Z^{\circ}$.
    
    We have $Z_j=Z'\subseteq Z^{\circ}$. As $Z_j$ is connected, it suffices to
    prove that for every $z\in Z^{\circ}\setminus (\{s_2\}\cup Z_j)$ there is a path in
    $G[Z^{\circ}\setminus\{s_2\}]$ from
    $z$ to a vertex in $Z_j$. So let 
    \[
    z\in
    Z^{\circ}\setminus (\{s_2\}\cup Z_j)\subseteq (Y_j\cup
    S_j)\setminus\{s_1,s_2\}.
    \]
    As $N(Z_j)=S_j$, it suffices to find a path from $z$ to a vertex in
    $S_j$. Let $P$ be a shortest path from $z$ to a vertex in $s\in S_j$ in the
    connected graph $G[Z^{\circ}]$. Then $V(P)\setminus\{s\}\subseteq Y_j$. We
    need to prove that $s_2\not\in V(P)$. 
    \begin{cs}
      \case1
      $j=1$.\\
      As $N(Y^{\circ})=S^{\circ}$, the vertex $s_1\in S^{\circ}$ has a neighbour in
      $Y^{\circ}\subseteq Y_1$, and as $s_2$ is the only neighbour of $s_1$ in
      $Y_1$, we have $s_2\in Y^{\circ}$. Hence $s_2\not\in V(P)$.
      
      \case2
      $j=2$.\\
      The only neighbour of $s_2\in S_2$ in $Y_2$ is $s_1$, and as
      $s_1\not\in V(P)$ and $V(P)\setminus\{s\}\subseteq Y_2$, we have
      $s_2\not\in V(P)$.
      \uend
    \end{cs}
  \end{claim}

We let $C^{\circ}:=G[Z^{\circ}\setminus S]$. It follows
from Claim~1 that this is indeed a connected component of $G\setminus
S$. Now let $(Y^\bullet,S^\bullet,Z^\bullet)\in\CT$ with $S^\bullet\subseteq S$. We need to prove that
\begin{equation}
  \label{eq:27}
  V(C^{\circ})=Z^{\circ}\setminus S=Z^{\circ}\setminus\{s_2\}\subseteq Z^\bullet
\end{equation}
Without loss of generality we may assume that $Z^\bullet=\CZ(S^\bullet)$.
It follows from Lemma~\ref{lem:t2} that $Z^\bullet\not\subseteq S$,
because then $|Z^\bullet\cup S^\bullet|\le|S|\le 4$. If $Y^\bullet\subseteq S$ then $Z^\bullet\supseteq V(G)\setminus S\supseteq
Z^{\circ}\setminus S=V(C^{\circ})$. 

Let us assume that neither $Z^\bullet\subseteq S$ nor $Y^\bullet\subseteq
S$. Then $S^\bullet$ is an essential subseparator of $S$. Note that
the analogue of Claim~1 holds for $S^\bullet$: the set $Z^\bullet
\setminus S$ is connected. 
If $S^{\circ}=S^\bullet$ then $Z^{\circ}\subseteq Z^\bullet$ and thus $V(C^{\circ})\subseteq
Z^\bullet$. So suppose that 
$S^\bullet\neq S^{\circ}$. We shall prove the following claim,
which implies $V(C^{\circ})\subseteq Z^\bullet$.

\begin{claim}[resume]
  $Z^{\circ}\setminus S=Z^{\bullet}\setminus S$.

  \proof It follows from Lemma~\ref{lem:4t15} that
  $|S^{\bullet}\setminus\{s_1,s_2\}|\ge 2$ and thus
  \begin{equation}
    \label{eq:28}
    S^{\bullet}\setminus\{s_1,s_2\}=S\setminus\{s_1,s_2\}=S^{\circ}\setminus\{s_1,s_2\}.
  \end{equation}
  Let $(Y'',S'',Z'')\in\CMT$ such that
  $(Y'',S'',Z'')\preceq(Y^{\bullet},S^{\bullet},Z^{\bullet})$. Then $S''=S_1$ or
  $S''=S_2$.  Moreover, $S'=S''$, because $S^{\circ}\subseteq Y'\cup S'$ and
  $S^{\bullet}\subseteq Y''\cup S''$ and if $S'\neq S''$ we have
  $S^{\circ}\cap S^{\bullet}\subseteq (Y'\cup S')\cap(Y''\cup S'')=\{s_1,s_2\}$,
  which contradicts \eqref{eq:28}.

  Since $s_1\in S^\circ$ we have $s_2\in S^\bullet$. Note the the
  setting is completely symmetric with respect to $S^\circ$ and
  $S^\bullet$, as both are essential subseparators of $S$. Hence
  without loss of generality we may assume that $S'=S''=S_1$. 

  We prove next that $s_2\in Y^\circ$ (by a similar argument as in the
  proof of Claim~1): the vertex $s_1\in S^\circ$ has a neighbour in
  $Y^\circ\subseteq Y_1$, and as the only neighbour of $s_1$ in $S_1$
  is $s_2$, we have $s_2\in Y^\circ$.

  Thus
  $S^{\bullet}=(S^\circ\setminus\{s_1\})\cup\{s_2\}\subseteq
  Y^{\circ}\cup S^{\circ}$. As $s_1\in S''\subseteq Z^\bullet\cup
  S^\bullet$, we have $s_1\in Z^\bullet$, and now 
$S^{\bullet}\subseteq
  Y^{\circ}\cup S^{\circ}$ implies $Z^{\circ}\subseteq Z^{\bullet}$,
  because $Z^\circ\cup\{s_1\}$ is connected. Hence
  $Z^{\circ}\setminus S\subseteq Z^{\bullet}\setminus S$. As both
  $Z^{\circ}\setminus S$ and $Z^{\bullet}\setminus S$ are vertex sets
  of connected components of $G\setminus S$, equality holds.  \uend
\end{claim}

The uniqueness of $C^{\circ}$ is immediate,
because $(Y^{\circ}, S^{\circ},Z^{\circ})\in\CT$ and $C^{\circ}$ is the only connected component of
$G\setminus S$ with $V(C^\circ)\subseteq Z^\circ$.
\end{proof}

We define the \emph{expansion} of a set $S'\subseteq V(G')$ to be the
set 
\[
S'_\wedge:=\begin{cases}
  S'\cup\{s_2\}&\text{if }s_1\in S',\\
  S'&\text{otherwise}.
\end{cases}
\]
Note that if $S'$ is a 3-separator of $G'$ then either $S'_\wedge=S'$ is a
3-separator of $G$ or $S'_\wedge=S'\cup\{s_2\}$ is a separator of
$G$ satisfying the assumptions of Lemmas~\ref{lem:4t15} and
\ref{lem:4t16}. 

Next, we define the \emph{contraction} of a set $S\subseteq
V(G)$ to be the set
\[
S^\vee:=
\begin{cases}
  (S\cup\{s_1\})\setminus\{s_2\}&\text{if }\{s_1,s_2\}\cap S\neq\emptyset,\\
  S&\text{if }\{s_1,s_2\}\cap S=\emptyset.
\end{cases}
\]
Note that $(S'_\wedge)^\vee=S'$ for all $S'\subseteq V(G')$, but only
$(S^\vee)_\wedge\supseteq S$ for $S\subseteq V(G)$, and the inclusion may be strict.

For every set
$S'\subseteq V(G')$ of order $|S'|\le 3$ we define a set  $\CZ'(S)$  as follows.
\begin{itemize}
\item If $S'$ is a separator with $s_1\in S'$, we let $S:=S'_\wedge$ be the
  expansion of $S'$ and $C^{\circ}$ the connected component of $G\setminus
  S$ obtained from Lemma~\ref{lem:4t16}. Then we let $\CZ'(S'):=V(C^{\circ})$.
\item If $S'$ is a separator with $s_1\not\in S'$, we let $\CZ'(S'):=\CZ(S')^\vee$ be
  the contraction of the set $\CZ(S')$.
\item If $S'$ is not a separator of $G'$, we let $\CZ'(S'):=V(G')\setminus 
  S'$. 
\end{itemize}
Observe that $\CZ'(S')$ is the vertex set of a connected component of
$G'\setminus S'$. We let $\CY'(S')=V(G')\setminus(S'\cup \CZ'(S'))$.

We define $\CT'$ to be the set of all separations
$(Y',S',T')\in\Sep_{<4}(G')$ such that $\CZ'(S')\subseteq Z'$.

\begin{lem}
  \label{lem:4t17}
  $\CT'$ is a $G'$-tangle of order $4$.
\end{lem}

\begin{proof}
  It follows immediately from the definition that $\CT'$ satisfies
  \ref{li:t1} and \ref{li:t3}. 

  To see that it satisfies \ref{li:t2}, let $(Y^i,S^i,Z^i)\in\CT'$
  for $i=1,2,3$. Suppose for contradiction that $Z^1\cap Z^2\cap
  Z^3=\emptyset$ and there is no edge that has an endvertex in every
  $Z^i$. For every $i$, we let $S^{i,0}:=S^i_\wedge$.
  We define a separation
  $(Y^{i,1},S^{i,1},Z^{i,1})$ of $G$
  as follows.
  \begin{eroman}
  \item If $S^i$ is a separator of $G'$ and $s_1\in S^i$, then we let
    $S^{i,1}$ be an essential subseparator of $S^{i,0}$
    and $Z^{i,1}:=\CZ(S^{i,1})$ and $Y^{i,1}:=\CY(S^{i,1})$.
  \item If $S^i$ is a separator of $G'$ and $s_1\not\in S^i$, then we let
    $S^{i,1}:=S^i$ and $Z^{i,1}:=\CZ(S^{i,1})$ and
    $Y^{i,1}:=\CY(S^{i,1})$. 
  \item If $S^i$ is not a separator of $G'$ and $s_1\in S^i$ then we
    let 
    \[
    S^{i,1}:=S^i=S^{i,0}\setminus\{s_2\},
    \]
    and we let
    $Z^{i,1}:=\CZ(S^{i,1})$ and $Y^{i,1}:=\CY(S^{i,1})$. 
  \item  If $S^i$ is not a separator of $G'$ and $s_1\not\in S^i$, then we let
    $S^{i,1}:=S^i$ and $Z^{i,1}:=V(G)\setminus S^i$ and
    $Y^{i,1}:=\emptyset$. 
  \end{eroman}
  Note that in cases (i) and (iii) we have $Z^{i,1}\setminus\{s_1,s_2\}=Z^i$. In
  cases (ii) and (iv), $Z^i$ is the contraction of $Z^{i,1}$. 
  Thus in all four cases we have 
  \begin{equation}
    \label{eq:29}
    Z^{i,1}\setminus\{s_1,s_2\}= Z^i\setminus\{s_1\}.
  \end{equation}
  Moreover, we have 
  $(Y^{i,1},S^{i,1},Z^{i,1})\in\CT$. We let
  $(Y^{i,2},S^{i,2},Z^{i,2})\in\CMT$ such that $(Y^{i,2},S^{i,2},Z^{i,2})\preceq(Y^{i,1},S^{i,1},Z^{i,1})$.

  \begin{claim}
    If $s_1\in S^i$ then $S^{i,2}\in\{S_1,S_2\}$. In particular,
    $s_1,s_2\in S^{i,2}\cup Y^{i,2}$.

  \proof This follows from \eqref{eq:25}.
  \uend
  \end{claim}

  \begin{claim}[resume]
    If $s_1\not\in S^i$, then $s_1\in Z^{i}$ and
    $s_1,s_2\in Z^{i,1}$.

    \proof
    If $s_1\in Y^{i}$ then $s_1,s_2\in Y^{i,1}\subseteq Y^{i,2}$. Thus
    for $j=1,2$ we have $s_j\in
    Y^{i,2}\cap Y_{3-j}$.  By the Crossing Lemma~\ref{lem:cross}, this
    implies $S^{i,2}=S_{3-j}$ for $j=1,2$, which is impossible.
    \uend
  \end{claim}
  
  By \ref{li:t2}, either $Z^{1,2}\cap Z^{2,2}\cap
  Z^{3,2}\neq\emptyset$ or there is an edge that has an endvertex in
  every $Z^{i,2}$.
  \begin{cs}
    \case1 $Z^{1,2}\cap Z^{2,2}\cap
    Z^{3,2}\neq\emptyset$.\\
    Let $v\in Z^{1,2}\cap Z^{2,2}\cap Z^{3,2}$. If
    $v\in V(G)\setminus\{s_1,s_2\}$ then
    $v\in Z^1\cap Z^2\cap Z^3$ by \eqref{eq:29}. Thus we may assume that
    $v=s_j$.

    By Claim~1, we have $s_1\not\in  S^i$ for $i=1,2,3$, because
    otherwise $s_j\not\in Z^{i,2}$. Then by Claim~2 we have $s_1\in
    Z^1\cap Z^2\cap Z^3$.

  \case2 $Z^{1,2}\cap Z^{2,2}\cap
  Z^{3,2}=\emptyset$.\\
  Then there is an edge $e=v_1v_2$ that has an endvertex in every
  $Z^{i,2}$. As $Z^{1,2}\cap Z^{2,2}\cap
  Z^{3,2}=\emptyset$, we have $v_1,v_2\in S^{1,2}\cup S^{2,2}\cup
  S^{3,2}$. For $i=1,2$, let $J_i$ be the set of $j$ such that $v_i\in
  Z^{j,2}$. Then $J_1\cup J_2=[3]$.

  For $i=1,2$, if $v_i\not\in \{s_1,s_2\}$
  then for all $j\in J_i$ we have $v_i\in Z^j$ by \eqref{eq:29}.
  Thus if $\{v_1,v_2\}\cap\{s_1,s_2\}=\emptyset$, then $e\in E(G')$
  and $e$ has an endvertex in each $Z^i$.

  So let us assume that $v_2\in\{s_1,s_2\}$. Say, $v_2=s_2$. Then for all $j\in J_2$ we have
  $s_1\not\in S^{j}$ by Claim~1 and thus $s_1\in Z^j$ by Claim~2.
  \begin{cs}
    \case{2a}
    $v_1=s_1$.\\
    Then for all $j\in J_1$ we have $s_1\not\in S^{j}$ by Claim~1 and
    thus $s_1\in Z^j$ by Claim~2. It follows that
    $s_1\in Z^1\cap Z^2\cap Z^3$.  \case{2a}
    $v_1\neq s_1$.\\
    Then $v_1\not\in\{s_1,s_2\}$ and thus for all $j\in J_1$ we have
    $v_1\in Z^j$ by \eqref{eq:29}. Furthermore, we have $e':=v_1s_1\in
    E(G')$. This edge $e'$ has an endvertex in every $Z^j$.
    \qedhere
  \end{cs}
\end{cs}
\end{proof}

\begin{lem}\label{lem:4t18}
  $\CT$ is the lifting of $\CT'$ to $G$ with respect to the
  contraction of the edge $e=s_1s_2$.
\end{lem}

\begin{proof}
  Let $(Y,S,Z)$ be a separation of $G$ of order at most
  $3$. Let $(M_{w})_{w\in V(G'}$ be the branch sets of the model $\CM$
  of
  $G'$ in $G$ that corresponds to the contraction of $e$. Then
  $M_{s_1}:=\{s_1,s_2\}$ and $M_w:=\{w\}$ for all $w\neq s_1$.
  Let $(Y',S',Z'):=\pi_{\CM}(Y,S,Z)$ (see \eqref{eq:proj}). Then
  \begin{equation}
    \label{eq:30}
      (Y',S',Z')=\big(Y^\vee\setminus S^\vee,S^\vee,Z^\vee\setminus 
      S^\vee\big).
  \end{equation}
  We need to prove 
  \begin{equation}
    \label{eq:31}
    (Y,S,Z)\in \CT\iff(Y',S',Z')\in\CT'. 
  \end{equation}
  We prove the backward direction first.
  Suppose that $(Y',S',Z')\in\CT'$. If $S'$
  is not a separator of $G'$, $\CZ'(S')=V(G')\setminus S'$, and either
  $Z=V(G)\setminus S$ or $Z=V(G)\setminus (S\cup\{s_2\})$. In both
  cases, $(Y,S,Z)\in\CT$. (If $Z=V(G)\setminus (S\cup\{s_2\})$, this
  follows from Lemma~\ref{lem:t2}.) 

  So suppose that $S'$ is a
  separator of $G'$. If $s_1\not\in S'$, then $S=S'$ is a separator of
  $G$, and $Z'=Z^\vee$ is the contraction of $Z$. As
  $(Y',S',Z')\in\CT'$, we have 
  $\CZ'(S')\subseteq Z'$. Now $\CZ'(S')$ is the contraction of $\CZ(S')=\CZ(S)$,
  and thus we have $\CZ(S)\subseteq Z$. This implies $(Y,S,Z)\in\CT$. 

  Finally, suppose that $s_1\in S'$. Then $Z'\supseteq \CZ'(S')=V(C^{\circ})$
  for the connected component $C^{\circ}$ of $G\setminus
  S$ obtained from Lemma~\ref{lem:4t16}, and we have
  $V(C^{\circ})\subseteq Z$ and thus $(Y,S,Z)\in\CT$.

  To prove the forward direction of \eqref{eq:31}, we just note that
  \[
  (Y,S,Z)\in
  \CT\iff(Z,S,Y)\not\in\CT\implies(Z',S',Y')\not\in\CT'\iff(Y',S',Z')\in\CT',
  \]
  where the middle implication follows from the backward direction.
\end{proof}

\begin{lem}\label{lem:4t19}
  Either $G'$ is 4-connected and
  $\CMT'=\big\{(\emptyset,\emptyset,V(G'))\big\}$ or 
  \begin{multline}\label{eq:32}
  \CMT'=\Big\{\big(Y^\vee\setminus S^\vee,S^\vee,Z^\vee\setminus
  S^\vee\big)\Bigmid\\
  (Y,S,Z)\in\CMT\text{ such that $S^\vee$ is a separator of }G'\Big\}.
  \end{multline}
  Furthermore, for all $S\in\CMS$ such that $S^\vee$ is a
  separator of $G'$,
  \begin{equation}
    \label{eq:33}
    S\in\CNDS\iff S^\vee\in\CNDS'.
  \end{equation}
\end{lem}

Recall that, by Lemma~\ref{lem:4t15a}, $S^\vee$ is a separator of $G'$
for all $S\in\CMS\setminus\{S_1,S_2\}$. Thus the clause ``such that
$S^\vee$ is a separator of $G'$'' only refers to the separators
$S_1^\vee=S_1$ and $S_2^\vee=(S_2\setminus\{s_2\})\cup\{s_1\}$, which
may not be separators of $G'$.

\begin{proof}[Proof of Lemma~\ref{lem:4t19}]
  If $G'$ is 4-connected then 
  \[
  \CT'=\{(\emptyset,S',V(G')\setminus
  S')\mid S'\subseteq V(G)\text{ width }|S'|\le 3\big\}.
  \]
  The unique minimal element of this set
  $(\emptyset,\emptyset,V(G')$. In the following, let us assume that
  $G'$ is not 4-connected.

  To prove the inclusion ``$\subseteq$'' of \eqref{eq:32}, let
  $(Y',S',Z')\in\CMT'$. Then $Y'\neq\emptyset$, because $G'$ is not
  4-connected. Hence
  the expansion $S'_\wedge$ is a separator of $G$.
  \begin{cs}
    \case1
    $s_1\in S'$.\\
    Let $S^{\circ}$ be an
    essential subseparator of $S'_\wedge$. Then
    $|S^{\circ}\cap\{s_1,s_2\}|=1$. Say, $s_i\in S^{\circ}$. Let $Z^{\circ}:=\CZ(S^\circ)$ and
    $Y^{\circ}:=\CY(S^\circ)$. Then
    $Z'=Z^{\circ}\setminus\{s_1,s_2\}$ and $Y'=Y^{\circ}\setminus\{s_1,s_2\}$ and $S'=(S^{\circ})^\vee$.

    Let $(Y,S,Z)\in\CMT$ such that $(Y,S,Z)\preceq(Y^{\circ},S^{\circ},Z^{\circ})$. Then
    $s_i\in S^{\circ}\subseteq Y\cup S$ and thus by \eqref{eq:25},
    $S=S_1$ or $S=S_2$. Say, $S=S_j$. Note that $Z_j=Z\subseteq
    Z^{\circ}\setminus\{s_1,s_2\}=Z'$. Thus
    \[
    (Y',S',Z')\succeq
    (Y_j\setminus\{s_{3-j}\},S_j^\vee,Z_j)=(Y_j^\vee\setminus
    S_j^\vee,S_j^\vee,Z_j^\vee\setminus S_j^\vee).
    \]
    By the minimality of $(Y',S',Z')$, equality holds. 

    \case2
    $s_1\not\in S'$.\\
    Then $S'_\wedge=S'$ and $(Y'_\wedge,S'_\wedge,Z'_\wedge)\in\CT$. Suppose for
    contradiction that $(Y'_\wedge,S'_\wedge,Z'_\wedge)$ is not
    minimal in $\CT$ and let $(Y,S,Z)\in\CMT$ such that
    $(Y,S,Z)\prec(Y'_\wedge,S'_\wedge,Z'_\wedge)$.
    \begin{cs}
      \case{2a}
      $S\cap\{s_1,s_2\}=\emptyset$.\\
      Then $(Y^\vee,S,Z^\vee)\in\CT'$ is strictly smaller than 
      $(Y',S',Z')$, which contradicts the minimality
      of $(Y',S',Z')$.
     \case{2b}
      $S\cap\{s_1,s_2\}\neq\emptyset$.\\
      Then $S\in\{S_1,S_2\}$. Say, $S=S_2$. Then
        $(Y_2\setminus\{s_1\},S_2^\vee,Z_2)\in\CT'$ is strictly smaller than 
      $(Y',S',Z')$, again contradicting the minimality of the latter.
    \end{cs}
  \end{cs}
  
  To prove the converse inclusion of \eqref{eq:32}, let $(Y,S,Z)\in\CMT$.
  \begin{cs}
    \case1
    $S\in\{S_1,S_2\}$.\\
    Say, $S=S_2$. Then $Y=Y_2$ and $Z=Z_2$. Clearly,
    \[
    (Y^\vee\setminus S^\vee,S^\vee,Z^\vee\setminus
    S^\vee)=(Y_2\setminus\{s_1\},(S_2\setminus\{s_2\})\cup\{s_1\},Z_2)\in\CT'.
    \]
    Suppose for contradiction that it is not minimal. Let
    $(Y',S',Z')\in\CMT'$ such that
    $(Y',S',Z')\prec(Y_2\setminus\{s_1\},(S_2\setminus\{s_2\})\cup\{s_1\},Z_2)$. By
    the converse inclusion (proved above), there is a
    $(Y'',S'',Z'')\in\CMT$ such that
    \[
    \big((Y'')^\vee\setminus(S'')^\vee,(S'')^\vee,(Z'')^\vee\setminus(S'')^\vee\big)=(Y',S',Z').
    \]
    Then $(Z'')^\vee\setminus(S'')^\vee\subset Z_2$, which implies
    $Z''\setminus\{s_1,s_2\}\subset Z_2$. By the minimality of
    $(Y_2,S_2,Z_2)$, we have $Z''\not\subset Z_2$. Thus
    $Z''\cap\{s_1,s_2\}\neq\emptyset$. Say, $s_j\in Z''$. Furthermore, 
    $(Y'')^\vee\cup (S'')^\vee\supset (Y_2\cup S_2)\setminus\{s_2\}$
    and thus $(Y''\cup S'')\cap\{s_1,s_2\}\neq\emptyset$. As $s_j\in
    Z''$ and $s_1s_2\in E(G)$ it follows that
    $s_{3-j}\in S''$. By \eqref{eq:25}, it follows that
    $S''=S_{3-j}$. But then $s_j\in Y''=Y_{3-j}$, which is a contradiction.

    \case2
     $S\not\in\{S_1,S_2\}$.\\
     Then $s_1,s_2\not\in Y\cup S$ by \eqref{eq:25}, and thus
     \[
      (Y^\vee\setminus S^\vee,S^\vee,Z^\vee\setminus
      S^\vee)=(Y,S,Z^\vee)\in\CT'.
      \]
      Suppose for contradiction that $(Y,S,Z^\vee)$ is not minimal in
      $\CT'$. Let
    $(Y',S',Z')\in\CMT'$ such that
    $(Y',S',Z')\prec(Y,S,Z^\vee)$. By
    the converse inclusion (proved above), there is a
    $(Y'',S'',Z'')\in\CMT$ such that
    \[
    \big((Y'')^\vee\setminus(S'')^\vee,(S'')^\vee,(Z'')^\vee\setminus(S'')^\vee\big)=(Y',S',Z').
    \]
    Then $(Y'')^\vee\cup(S'')^\vee=Y'\cup S'\supset Y\cup S$, which implies
    $Y''\cup S''\supset Y\cup S$ and thus
    $(Y'',S'',Z'')\prec(Y,S,Z)$. This contradicts the minimality of $(Y,S,Z)$.
  \end{cs}

  \medskip
  It remains to prove \eqref{eq:33}. Let $S\in\CMS$ such that $S^\vee$ is a
  separator of $G'$. Then $S^\vee\in\CMS'$. Let $Y:=\CY(S)$ and
  $Z:=\CZ(S)$. Then $(Y,S,Z)\in\CMT$ and
  $(Y^\vee\setminus S^\vee,S^\vee,Z^\vee\setminus S^\vee)\in\CMS'$. 
  
  Suppose first that $S\not\in\{S_1,S_2\}$. Then $S^\vee=S$ and
  $Y^\vee\setminus S^\vee=Y^\vee=Y$ and
  $Z^\vee\setminus S^\vee=Z^\vee$, because $s_1,s_2\in Z$ by
  \eqref{eq:25}. It follows that $S$ is degenerate in $G$ if and only
  if $S^\vee$ is degenerate in $G'$.

  If $S\in\{S_1,S_2\}$, say, $S=S_j$, we need
  to prove that $S^\vee$ is non-degenerate.
  We have
  \[
  (Y^\vee\setminus S^\vee,S^\vee,Z^\vee\setminus
  S^\vee)=(Y_j\setminus\{s_{3-j}\},S_j^\vee,Z_j).
  \]
  As $S^\vee$ is a separator of $G'$, we have
  $Y_j\setminus\{s_{3-j}\}\neq\emptyset$.  If
  $|Y_j\setminus\{s_{3-j}\}|>1$, then
  $(Y_j\setminus\{s_{3-j}\},S_j^\vee,Z_j)$ is non-degenerate.  Suppose
  that $|Y_j\setminus\{s_{3-j}\}|=1$. Then $s_{3-j}$ is adjacent to
  $s_j'$ or $s_j''$, because $s_{3-j}$ has degree at least $3$ in
  $G$. Say, $s_{3-j}s_j'\in E(G)$. Then $s_1s_j'\in E(G')$, and thus
  $S_j^\vee$ is not an independent set in $G'$, which again means that
  $(Y_j\setminus\{s_{3-j}\},S_j^\vee,Z_j)$ is non-degenerate.
\end{proof}

$S_i^\vee$ is not necessarily a separator of $G'$. However, if $S_i$ is crossed
by some separator $S\in\CMS\setminus\{S_1,S_2\}$ then $S_i^\vee$
remains a separator of $G'$. Hence we get the following corollary.

\begin{cor}\label{cor:4t19}
  $\EC(\CT')=\EC(\CT)\setminus\{e\}$ and $\ECND(\CT')=\ECND(\CT)\setminus\{e\}$.
\end{cor}

\subsection{The Region of the Tangle}
\label{sec:tangleregion}

We still assume that $G$ is a 3-connected graph and $\CT$ is a $G$-tangle
of order $4$ (but we drop Assumption~\ref{ass:4t3}). 
Let $e^{1},\ldots,e^{m}$ be an enumeration of all non-degenerate
crossedges of $\CT$. Recall that, by Corollary~\ref{cor:CI},
$\{e^{1},\ldots,e^{m}\}$ is a matching of $G$. Say, $e^{i}=s_1^{i}s_2^{i}$.
Let $G^{(0)}:=G$, and for $i\in[m]$, let $G^{(i)}$ be the graph
obtained from $G$ by contracting the edges $e^{1},\ldots,e^{i}$ to
the vertices $s_1^{j}$. We inductively define for all $i$ a
$G^{(i)}$-tangle $\CT^{(i)}$ of order $3$ as follows.

We let $\CT^{(0)}:=\CT$. To define $\CT^{(i+1)}$, we assume that
$G^{(i)}$ is 3-connected and $\CT^{(i)}$ is a $G^{(i)}$-tangle of order
  $3$ and $e^{i+1}$ is a non-degenerate crossedge of $\CT^{(i)}$. Then Assumptions~\ref{ass:4t1}, \ref{ass:4t2},
\ref{ass:4t3} are satisfied with $G:=G^{(i)}, \CT:=\CT^{(i)},
e:=e^{i+1}, G':=G^{(i+1)}$, and we can apply the results of
Section~\ref{sec:cce}. We let $\CT^{(i+1)}$ be the tangle $(\CT^{(i)}\big)'$
(see Lemma~\ref{lem:4t17}).

For every $i\in[m]$ and $v\in V(G)$ we let 
\[
v^{\slashes i}:=
\begin{cases}
  s_1^{j}&\text{if }v\in\{s_1^{j},s_2^{j}\}\text{ for some }j\le i,\\
  v&\text{otherwise}.
\end{cases}
\]
For $W\subseteq V(G)$, we let $W^{\slashes i}:=\{w^{\slashes i}\mid w\in W\}$.

\begin{lem}\label{lem:4t20}
  Let $i\in[m]$.
  \begin{enumerate}
  \item $G^{(i)}$ is 3-connected.
  \item $\CT^{(i)}$ is a $G^{(i)}$-tangle of order $3$ with
    \begin{multline*}
      \CMT^{(i)}=\Big\{\big(Y^{\slashes i}\setminus S^{\slashes i},S^{\slashes i},Z^{\slashes i}\setminus S^{\slashes i}\big)\Bigmid\\
      (Y,S,Z)\in\CMT\text{ such that $S^{\slashes i}$ is a separator of }G^{(i)}\Big\}.
    \end{multline*}
    For $i=m$, it may also happen that $G^{(m)}$ is 4-connected and
    $(\CT^{(m)})_{\min}=\{(\emptyset,\emptyset,V(G^{(m)}))\}$.
  \item $\ECND(\CT^{(i)})=\{e^{i+1},\ldots,e^{m}\}$.
  \item $\CT$ is the lifting of $\CT^{(i)}$ from $G^{(i)}$ to $G$ with 
    respect to the contraction of $e^{1},\ldots,e^{i}$. 
  \item The graph $G^{(i)}$ and the tangle $\CT^{(i)}$ do not depend
    on the order in which the edges $e^{1},\ldots,e^{i}$ are contracted.
    
    Up to isomorphism, $G^{(i)}$ and $\CT^{(i)}$ also do not depend on
    whether $e^{j}$ is contracted to $s_1^{j}$ or $s_2^{j}$.
  \end{enumerate}
\end{lem}

\begin{proof}
  Assertions (1)--(4) follow by induction from
  Corollary~\ref{cor:4t15}, Lemma~\ref{lem:4t17} and
  Lemma~\ref{lem:4t19}, Corollary~\ref{cor:4t19},
  Lemma~\ref{lem:4t18}, respectively. Assertion (5) is obvious as far
  as the graph $G^{(i)}$ is concerned, and for the tangle $\CT^{(i)}$
  it follows from (4).
\end{proof}

We let
\[
R^{(0)}:=\bigcap_{(Y,S,Z)\in\CNDT}Z
\cup\bigcup_{S\in\CNDS}S
\]
and 
\[
R^{(i)}:=R^{(0)}\setminus\{s_2^{1},\ldots,s_2^{i}\}.
\]
We shall prove that $R^{(m)}$ is a non-exceptional quasi-4-connected
region of $G$ and that $\CT$ is equal to $\CT_{R^{(m)}}$, the tangle associated with
$R^{(m)}$.

For all $i\in[m]$ we let 
\[
H^{(i)}:=\torso{G}{R^{(i)}}.
\]

The \emph{fence} of a separator $S\in\CNDS$ is the set
$\operatorname{fc}(S)$ consisting of all vertices in $S$
that are not endvertices of a non-degenerate crossedge of $S$ and all vertices in
$\CY(S)$ that are endvertices of non-degenerate crossedges. For example, if
$S=\{s_1,s_2,s_3\}$, and $S$ is crossed by $S_1,S_2\in\CNDS$ with
crossedges $s_1s_1'$ and $s_2s_2'$, respectively, then $\fc(S)=\{s_1',s_2',s_3\}$. %
Note that $\fc(S)\subseteq R^{(0)}$ and that $\fc(S)\setminus
S=\CY(S)\cap R^{(0)}$.

\begin{lem}\label{lem:4t21}
  For every connected component $C$ of $G\setminus R^{(0)}$ there is a
  unique $S\in\CNDS$ such that $V(C)\subseteq \CY(S)\setminus\fc(S)$ and $N(C)=\fc(S)$.
\end{lem}

\begin{proof}
  It follows from the definition of $R^{(0)}$ and the
  connectedness of $C$ that there is an $S\in\CNDS$ such that
  $V(C)\subseteq \CY(S)$. This $S$ is unique, because $\CY(S)\cap
  \CY(S')=\emptyset$ for distinct $S,S'\in\CMS$ (by
  the Crossing Lemma~\ref{lem:cross}).

  Let $S_1,\ldots,S_k\in\CNDS$ be the non-degenerate separators crossing $S$, and let
  $s_is_i'$ with $s_i\in S$ and $s_i'\in S_i$ be the crossedge of $S$
  and $S_i$. Then $k\le3$. Let $s_{k+1},\ldots,s_3$ be the elements of
  $S\setminus\{s_1,\ldots,s_k\}$. Then $\fc(S)=\{s_1',\ldots,s_k',s_{k+1},\ldots,s_3\}$ and $C$ is a connected
  component of $G[\CY(S)]\setminus\{s_1',\ldots,s_k'\}$. Hence
  $N(C)\subseteq\{s_1,s_2,s_3\}\cup\{s_1',\ldots,s_k'\}$. However, it
  follows from the Crossing Lemma~\ref{lem:cross} that $s_i'$ is the only neighbour
  of $s_i$ in $\CY(S)$, for $i\le k$. Hence $s_i\not\in N(C)$, and we
  have $N(C)\subseteq\fc(S)$. As $G$ is 3-connected and $\fc(S)=3$,
  we have $N(C)=\fc(S)$
\end{proof}

\begin{lem}\label{lem:4t22}
  ${H}^{(0)}$ is a faithful minor of $G$.
\end{lem}

\begin{proof}
  By Lemma~\ref{lem:4t21}, for every connected component $C$ of
  $G\setminus R^{(0)}$ there is an $S\in\CNDS$ such that
  $V(C)\subseteq\CY(S)\setminus\fc(S)$ and $N(C)=\fc(S)$. Thus by
  Lemma~\ref{lem:degsep}, it suffices to prove that for every
  $S\in\CNDS$ such that
  $\CY(S)\setminus\fc(S)\neq\emptyset$, either $\fc(S)$ is not
  an independent set or $|\CY(S)\setminus\fc(S)|\ge 2$. 

  So let $S=\{s_1,s_2,s_3\}\in\CNDS$ such that
  $\CY(S)\setminus\fc(S)\neq\emptyset$, and let $k$ be the number of $S'\in\CMS$
  crossing $S$. Then $0\le k\le 3$. If $k\ge 1$, let
  $S_1,\ldots,S_k\in\CMS$ be the separations crossing $S$ and let
  $s_is_i'$ be the crossedges.
  \begin{cs}
  \case1 $k=0$.\\
  Then ${\fc(S)}=S$ and $\CY(S)\setminus\fc(S)=\CY(S)$ and thus either ${\fc(S)}$ is not
  an independent set or $|\CY(S)\setminus\fc(S)|\ge 2$, because $S$ is non-degenerate.

  \case2 $k=1$.\\
  Then ${\fc(S)}=\{s_1',s_2,s_3\}$ and thus $\CY(S)\setminus\fc(S)=\CY(S)\setminus\{s_1'\}$. Suppose that $|\CY(S)\setminus\fc(S)|=1$. The degree of $s_1'$ is
  at least $3$, and as all its neighbours are in $S\cup{\CY}(S)$, it
  must have at least one neighbour in $\{s_2,s_3\}$. Thus ${\fc(S)}$ is
  not an independent set.
  \case3 $k=2$.\\
  Then ${\fc(S)}=\{s_1',s_2',s_3\}$ and thus $\CY(S)\setminus\fc(S)=\CY(S)\setminus\{s_1',s_2'\}$. Suppose that $|\CY(S)\setminus\fc(S)|=1$. The degree of $s_1'$ is
  at least $3$, and all its neighbours are in $S\cup{\CY}(S)$. It
  follows from the Crossing Lemma~\ref{lem:cross} that the only neighbour of $s_2$ in
  $\CY(S)$ is $s_2'$. Thus $s_2$ is not a neighbour of $s_1'$, and
  $s_1'$ must have at least one neighbour in $\{s_2',s_3\}$. Hence ${\fc(S)}$ is
  not an independent set.
  \case4 $k=3$.\\
  Similar to the previous cases.\qedhere
  \end{cs}
\end{proof}

\begin{lem}\label{lem:4t23}
  For all $S\in\CNDS$ the set ${\fc(S)}$ is a clique in ${H}^{(0)}$.
\end{lem}

\begin{proof}
  Let $S=\{s_1,s_2,s_3\}\in\CNDS$. As in the proof of the previous
  lemma, let $k$ be the number of $S'\in\CNDS$
  crossing $S$. Let
  $S_1,\ldots,S_k\in\CNDS$ be the separations crossing $S$ and let
  $s_is_i'$ be the crossedges.
  If $\CY(S)\setminus\fc(S)\neq\emptyset$, the claim is immediate from the
  definition of the torso. So suppose that $\CY(S)\setminus\fc(S)=\emptyset$. Then
  $\CY(S)\subseteq {\fc(S)}\setminus S=\{s_1',\ldots,s_k'\}$, and as
  $\CY(S)\neq\emptyset$, this implies that $k\ge 1$. Now we argue similarly to
  the proof of the previous lemma.
 \begin{cs}
   \case1 $k=1$.\\
   Then ${\fc(S)}=\{s_1',s_2,s_3\}$ and $\CY(S)=\{s_1'\}$ and hence
   $N^G(s_1')=\{s_1,s_2,s_3\}$. In particular,
   $s_1's_2,s_1's_3\in E(G)$. As $|\CY(S)|=1$ and $S$ is
   non-degenerate, $S$ is not an independent set. There is no edge
   from $\{s_2,s_3\}\subseteq\CZ(S_1)$ to $s_1\in\CY(S_1)$. Hence
   $s_2s_3\in E(G)$, which implies $s_2s_3\in E(H^{(0)})$. Thus
   ${\fc(S)}$ is a clique in ${H}^{(0)}$.
   \case2 $k=2$.\\
   Then ${\fc(S)}=\{s_1',s_2',s_3\}$ and thus
   $\CY(S)=\{s_1',s_2'\}$. The degree of
   $s_1'$ is at least $3$, and all its neighbours are in
   $S\cup\CY(S)\setminus\fc(S)$. It follows from the Crossing
   Lemma~\ref{lem:cross} that the only neighbour of $s_2$ in $\CY(S)$
   is $s_2'$. Thus $s_2$ is not a neighbour of $s_1'$, and therefore
   $N(s_1')=\{s_1,s_2',s_3\}$. Similarly,
   $N(s_2')=\{s_1',s_2,s_3\}$. Hence ${\fc(S)}=\{s_1',s_2',s_3\}$ is a
   clique.
   \case3 $k=3$.\\
   Then ${\fc(S)}=\{s_1',s_2',s_3'\}$ and thus
   $\CY(S)=\{s_1',s_2',s_3'\}$. The
   degree of $s_1'$ is at least $3$, and all its neighbours are in
   $S\cup\CY(S)\setminus\fc(S)$. As $s_2,s_3$ are not neighbours of
   $s_1'$, we have $N(s_1')=\{s_1,s_2',s_3'\}$. Similarly,
   $N(s_2')=\{s_1',s_2,s_3'\}$. Hence ${\fc(S)}$ is a clique.\qedhere
  \end{cs}
\end{proof}

The next lemma is a little bit surprising. It says that instead of
\emph{deleting} the endpoints $s_2^j$ of the crossedges $e^j$, we
could have \emph{contracted} the crossedges to $s_1^j$ with the same result.

\begin{lem}\label{lem:4t24}
  For all $i\ge 0$, the graph $H^{(i)}$ is 
  equal to the
  graph obtained from ${H}^{(0)}$ by contracting the edges
  $e^{1},\ldots,e^{i}$ to $s_1^{1},\ldots,s_1^{i}$, respectively.
\end{lem}

\begin{proof}
  The proof is by induction on $i$. The base step $i=0$ is trivial.

  For the inductive step $i\to i+1$, suppose that $e^{i+1}$ is the
  crossedge of $S_1,S_2\in\CNDS$ with $s_j:=s^{i+1}_j\in S_j$. Then
  $s_j\in \fc(S_{3-j})$. Suppose that $\fc(S_j)=\{s_{3-j},t_j,u_j\}$.
  By a similar analysis as in the proofs of
  the previous two lemmas, we see that
  \[
  N^{{H}^{(0)}}(s_j)=\{s_{3-j}\}\cup\big(\fc(S_{3-j})\setminus\{s_j\}\big)=\{s_{3-j},t_{3-j},u_{3-j}\}.
  \]
  By the induction hypothesis, ${H}^{(i)}$ is the graph obtained
  from ${H}^{(0)}$ by contracting $e^{1},\ldots,e^{i}$ to
  $s_1^{1},\ldots,s_1^{i}$, respectively. Hence
  \[
  N^{{H}^{(i)}}(s_j)=\big(N^{{H}^{(0)}}(s_j)\big)^{\slashes
    i}
  =\{s_{3-j},t_{3-j}^{\slashes i},u_{3-j}^{\slashes i}\}.
  \]
  Here we use the fact that $s_{3-j}^{\slashes i}=s_{3-j}$ because
  the edges $e^1,\ldots,e^{i+1}$ form a matching.

  Hence the neighbours of $s_1$ in the graph obtained from
  ${H}^{(i)}$ by contracting the edge $e^{i+1}=s_1s_2$ to $s_1$ are
  $t_1^{\slashes i},u_1^{\slashes i},t_2^{\slashes i},u_2^{\slashes
    i}$.
  I claim that the neighbours of $s_1$ in the graph ${H}^{(i+1)}$ are
  $t_1^{\slashes i},u_1^{\slashes i},t_2^{\slashes i},u_2^{\slashes
    i}$
  as well. Observe that
  \[
  {H}^{(i+1)}=\torso{{H}^{(i)}}{R^{(i+1})}=\torso{H^{(i)}}{R^{(i)}\setminus\{s_2\}}.
  \]
  The vertices $t_2^{\slashes i},u_2^{\slashes
    i}$ are neighbours of $s_1$ in ${H}^{(i)}$ and hence in $ {H}^{(i+1)}=\torso{H^{(i)}}{R^{(i)}\setminus\{s_2\}}$. To see that $t_1^{\slashes i},u_1^{\slashes
    i}$ are neighbours of $s_1$ in $ {H}^{(i+1)}$, note that the unique connected component of
  ${H}^{(i)}\setminus R^{(i+1)}$ consists of the vertex $s_2$, and
  $N^{{H}^{(i)}}(s_2)=\{s_1,t_1^{\slashes i},u_1^{\slashes
    i}\}$. Thus the set $\{s_1,t_1^{\slashes i},u_1^{\slashes
    i}\}$ is a clique in the torso $ {H}^{(i+1)}=\torso{H^{(i)}}{R^{(i+1)}}$. This implies that $s_1t_1^{\slashes i}$ and
  $s_1u_1^{\slashes i}$ are edges of $ {H}^{(i+1)}$.

  Hence indeed ${H}^{(i+1)}$ is the graph obtained from ${H}^{(i)}$ by
    contracting the edge $e^{i+1}$ to $s_1=s_1^{i+1}$.
\end{proof}

\begin{cor}\label{cor:4t24a}
  For $0\le i\le m$, the graph ${H}^{(i)}$ is a faithful minor of $G$.
\end{cor}

\begin{cor}\label{cor:4t24b}
  For $0\le i\le m$, 
  \[
  {H}^{(i)}=\torso{G^{(i)}}{R^{(i)}}.
  \]
\end{cor}

\begin{lem}\label{lem:4t25}
  \[
  R^{(m)}=V(G^{(m)})\setminus\bigcup_{(Y,S,Z)\in\CNDT^{(m)}}Y.
  \]
\end{lem}

\begin{proof}
  This follows from Lemmas~\ref{lem:4t20}(2) and \ref{lem:4t24} and the
  definition of $R^{(m)}$.
\end{proof}

\begin{cor}\label{cor:4t25b}
    For all connected components $C$ of $G\setminus R^{(m)}$ there is
  an $S\in\CNDS$ such that $N(C)=S^{\slashes m}$.
\end{cor}

\begin{lem}\label{lem:4t25a}
  Let $(Y,S,Z)\in\Sep_{<4}(H^{(m)})$ be a proper separation. Then
  $(Y,S,Z)$ or $(Z,S,Y)$ is degenerate. Furthermore, 
  either $\big(Y,S,V(G^{(m)})\setminus(Y\cup S)\big)$ or
  $\big(Z,S,V(G^{(m)})\setminus(Z\cup S)\big)$ is a degenerate
  separation of $G^{(m)}$ contained in $\CMT^{(m)}$.
\end{lem}

\begin{proof}
  $(Y,S,Z)$ gives rise to a separation $(Y',S,Z')$ of $G^{(m)}$ with
  $Y=Y'\cap R^{(m)}$ and $Z=Z'\cap R^{(m)}$. Without loss of generality we
  assume that $(Y',S,Z')\in\CT^{(m)}$. Let $(Y'',S'',Z'')\in\CMT$ such
  that $(Y'',S'',Z'')\preceq (Y',S,Z')$. Then $Y\subseteq Y'\subseteq
  Y''$. If $(Y'',S'',Z'')$ is
  non-degenerate, then by Lemma~\ref{lem:4t25}, $Y''\cap
  R^{(m)}=\emptyset$, and thus $Y\subseteq Y''\cap R^{(m)}=\emptyset$,
  which contradicts $(Y,S,Z)$ being a proper separation. Thus
  $(Y'',S'',Z'')$ is degenerate and therefore $|Y''|=1$. But then
  $Y=Y'=Y''$ and thus $S'=S$ and 
  $Z'=Z''=V(G^{(m)})\setminus(Y\cup S)$.
\end{proof}

\begin{lem}\label{lem:4t26}
  $R^{(m)}$ is a quasi-4-connected region of $G$.
\end{lem}

\begin{proof}
  We have seen that $H^{(m)}=\torso G{R^{(m)}}$ is a faithful minor
  of $G$ (Corollary~\ref{cor:4t24a}). Thus $R^{(m)}$ satisfies
  \ref{li:q1}. By Lemma~\ref{lem:4t25a}, $H^{(m)}$ is quasi-4-connected. Thus $R^{(m)}$
  satisfies \ref{li:q2}. It follows from Corollary~\ref{cor:4t25b} that $R^{(m)}$
  satisfies \ref{li:q3}. 
\end{proof}

It remains to prove that $R^{(m)}$ is non-exceptional and that
$\CT=\CT_{R^{(m)}}$. (Recall the definition of $\CT_{R^{(m)}}$ from Theorem~\ref{theo:q4r}.)

\begin{lem}\label{lem:4t27}
  If $H^{(m)}$ is non-exceptional then $\CT=\CT_{R^{(m)}}$.
\end{lem}

\begin{proof}
  Suppose that $H^{(m)}$ is non-exceptional. 
  Let $\hat\CT$ be the
  unique $H^{(m)}$-tangle of order $4$. By Theorem~\ref{theo:q4c},
  $\CHT$ consist of all separations $(Y',S',Z')\in\Sep_{<4}(H^{(m)})$
  such that $|Y'|<|Z'|$.
  By the transitivity of the lifting relation,
  Lemma~\ref{lem:4t20}(4), and Corollary~\ref{cor:4t24b}, it suffices
  to prove that $\CT^{(m)}$ is the lifting of $\hat\CT$ to $G^{(m)}$
  with respect to some faithful model of $H^{(m)}$ in $G^{(m)}$.

  So let $\CM$ be a faithful model of $H^{(m)}$
  in $G^{(m)}$. Let $(Y,S,Z)\in\CT^{(m)}$ and
  $(Y',S',Z'):=\pi_{\CM}(Y,S,Z)$. Recall from the definition
  \eqref{eq:proj} of the
  projection operation that $Y'\subseteq Y\cap R^{(m)}$ and $Z'\subseteq Z\cap
  R^{(m)}$. We need to prove that
  $(Y',S',Z')\in\CHT$. Let $(Y'',S'',Z'')\in\CMT^{(m)}$ such that
  $(Y'',S'',Z'')\preceq(Y,S,Z)$. Then $Y'\subseteq Y\subseteq Y''$. If $(Y'',S'',Z'')$
  is non-degenerate, then by Lemma~\ref{lem:4t25} we have $Y''\cap
  R^{(m)}=\emptyset$. Thus $Y'=\emptyset$, which implies
  $(Y',S',Z')\in\CHT$. If $(Y'',S'',Z'')$ is degenerate, then
  $|Y''|=1$ and thus $|Y'|\le 1$. Again, this implies
  $(Y',S',Z')\in\CHT$: either $Y'=\emptyset$, or $|Y'|=1$, but then
  $|R^{(m)}|\ge 6$, because  
  $(Y',S',Z')$ is a proper separation of $H^{(m)}$ and the only
  non-exceptional quasi-4-connected graph of order $5$ is the complete
  graph, which has no proper separation.
\end{proof}

\begin{lem}\label{lem:4t28}
  Suppose that $H^{(m)}$ is exceptional. Then it is isomorphic to a
  subgraph of $\THT$. Furthermore, there is an embedding $f$ of $H^{(m)}$
  into $\THT$ such that $v_1,\ldots,v_4\in f(R^{(m)})$, and if $w_j\in
  f(R^{(m)})$, say with $w_j=f(w_j')$, then
  $N^{G^{(m)}}(w_j')=N^{H^{(m)}}(w_j')$ and 
  \[
  \Big(\{w_j'\},N^{G^{(m)}}(w_j'),V(G^{(m)})\setminus\big(\{w_j'\}\cup
  N^{G^{(m)}}(w_j')\big)\Big)
  \]
  is a degenerate separation of $G^{(m)}$ in $\CMT^{(m)}$.
\end{lem}

\begin{proof}
  Suppose that $H^{(m)}$ is not isomorphic to a subgraph of $\THT$. Then,
  without loss of generality, it is a subgraph of $\TRT$ with
  $w_1,w_2,w_3\in V(H^{(m)})$.  We apply Lemma~\ref{lem:4t25a} to the
  separation $(\{w_1\},\{v_1,v_2,v_3\},\{w_2,w_3\})$ of
  $H^{(m)}$. Then 
  \[
  (\{w_1\},\{v_1,v_2,v_3\},V(G^{(m)}\setminus\{v_1,v_2,v_3,w_1\})\in\CMT.
  \]
  However, $V(G^{(m)}\setminus\{v_1,v_2,v_3,w_1\})$ is not connected
  in $G^{(m)}$, because $w_2$ and $w_3$ belong to different connected
  components. This contradicts Corollary~\ref{cor:reed2}.

  Thus $H^{(m)}$ is isomorphic to a subgraph of $\THT$. Without
  loss of generality, we may assume that $H^{(m)}\subseteq\THT$. We
  may further assume that $v_1,\ldots,v_4\in V(H^{(m)})$, because the
  only quasi-4-connected subgraphs of $\THT$ that do not contain
  $v_1,\ldots,v_4$ contain one $w_i$ and the three adjacent $v_i$s,
  and these subgraphs are isomorphic to the subgraph induced by
  $v_1,\ldots,v_4$. 

  Suppose that $w_j\in V(H^{(m)})$ for some $j$, say,
  $j=1$. Let $S:=N^{H^{(m)}}(w_1)=\{v_1,v_2,v_3\}$. Then
  $(\{w_1\},S,V(H^{(m)})\setminus\big(S\cup\{w_1\}\big)$ is a proper separation of
  $H^{(m)}$. By 
  Lemma~\ref{lem:4t25a}, either
  \[
  \Big(\{w_1\},S,V(G^{(m)})\setminus\big(S\cup\{w_1\}\big)\Big)
  \]
  is a degenerate separation of $G^{(m)}$ in $\CMT^{(m)}$, or $|V(H^{(m)})\setminus(S\cup\{w_1\})|=1$,
  which implies that $V(H^{(m)})\setminus(S\cup\{w_1\})=\{v_4\}$, and
   \[
  \Big(\{v_4\},S,V(G^{(m)})\setminus\big(S\cup\{v_4\}\big)\Big)
  \]
  is a degenerate separation of $G^{(m)}$ in $\CMT^{(m)}$. 
 In the former case, we are done, and in the latter case the mapping $f:V(H^{(m)})\to
  V(\THT)$ defined by $f(w_1):=v_4$, $f(v_4)=w_1$, $f(v_i)=v_i$ for
  $i=1,2,3$ is an embedding of $H^{(m)}$ into $\THT$ with the desired properties.
\end{proof}

\begin{lem}\label{lem:4t29}
  Suppose that $H^{(m)}$ is exceptional. Then there is a
  non-exceptional extension $\hat H$ of $R^{(m)}$ and a faithful model $\CM$
  of $\hat H$ in $G$ such that $\CT=\CT(\hat H,\CM)$.
\end{lem}

\begin{proof}
  We first prove the assertion for $G^{(m)}$ instead of $G$; it will then be
  easy to lift it to $G$.

  \begin{claim}
    There is a non-exceptional extension $\hat H$ of $R^{(m)}$ in
    $G^{(m)}$ and a faithful model $\CM$ of $\hat H$ in $G^{(m)}$ such
    that $\CT^{(m)}$ is the lifting of the unique $\hat H$-tangle of
    order $4$ to $G^{(m)}$ with
    respect to $\CM$.

    \proof By Lemma~\ref{lem:4t28}, we may assume without loss of
    generality that $H^{(m)}\subseteq\THT$ with
    $\{v_1,\ldots,v_4\}\in R^{(m)}$.  We actually view $H^{(m)}$ as a
    subgraph of $\THF$, and for $j=1,\ldots,4$, let
    $S_j:=N^{\THT}(w_j)$. Then, by Lemma~\ref{lem:4t28}, if
    $w_j\in R^{(m)}$ then
    $\big(\{w_j\},S_j,V(G)\setminus(\{w_j\}\cup S_j\big)$ is a
    degenerate 3-separation of $G^{(m)}$.

    For $i\in[4]$, let $Z_i$ be the
    connected component of $G^{(m)}\setminus S_i$ that contains the
    unique element in $\{v_1,\ldots,v_4\}\setminus S_i$, and we let
    $Y_i:=V(G^{(m)})\setminus(S_i\cup Z_i)$. Then
    $(Y_i,S_i,Z_i)\in\CT^{(m)}$. This can be seen as follows. If
    $w_i\in R^{(m)}$, then 
    $\big(\{w_i\},S_i,V(G)\setminus(\{w_i\}\cup S_i)\big)$ is a
    degenerate 3-separation of $G^{(m)}$ in $\CMT^{(m)}$. Thus
    $V(G)\setminus(\{w_i\}\cup S_i)$ is connected, which implies
    $Z_i=V(G)\setminus(\{w_i\}\cup S_i)$ and $Y_i=\{w_1\}$ and 
    $(Y_i,S_i,Z_i)=\big(\{w_i\},S_i,V(G)\setminus(\{w_i\}\cup
    S_i)\big)\in\CT^{(m)}$. Otherwise,
    $Y_i\cap R^{(m)}=\emptyset$, and either $Y_i=\emptyset$, which
    trivially implies $(Y_i,S_i,Z_i)\in\CT^{(m)}$, or it follows from
    Lemma~\ref{lem:4t25} that $(Y_i,S_i,Z_i)\in\CNDT^{(m)}$.

    Now it follows from \ref{li:t3} that if $w_j\not\in R^{(m)}$ then
    there is a connected component $C_j$ of $G^{(m)}\setminus R^{(m)}$
    such that $N(C_j)=S_j$, because otherwise for the three
    separations $(Y_i,S_i,Z_i)\in\CT^{(m)}$, where
    $i\in[4]\setminus\{j\}$, the intersection of the $Z_i$ is empty,
    and there is no edge that has an endvertex in each $Z_i$.
    Contracting $C_j$ to a single vertex gives us the desired faithful
    model $\CM$ of a full subgraph $\hat H$ of $\THF$ in
    $G^{(m)}$. 

    Let $\CHT$ be the unique $\hat H$-tangle of order $4$.
    It remains to prove that $\CT^{(m)}$ is the lifting of $\CHT$ with
    respect to $\CM$. So let $(Y,S,Z)\in\CT^{(m)}$, and let $(\hat
    Y,\hat S,\hat Z):=\pi_{\CM}(Y,S,Z)$. We need to prove that $(\hat
    Y,\hat S,\hat Z)\in\CHT$. As $\hat Y\subseteq Y$, we may assume
    without loss of generality that $|Y|\ge 2$. Let $(Y',S',Z')\in\CMT^{(m)}$
    such that $(Y',S',Z')\preceq(Y,S,Z)$. Then $Y\subseteq Y'$, and thus
    $(Y',S',Z')$ is non-degenerate. By Lemma~\ref{lem:4t25}, $Y'\cap
    R^{(m)}=\emptyset$ and thus 
    \[
    |\hat Y|\le|Y\cap V(\hat H)|\le 1
    \]
   By Lemma~\ref{lem:t2}, it follows that $(\hat
    Y,\hat S,\hat Z)\in\CHT$.
    \uend
  \end{claim}

  Now choose $\hat H$ according to the claim, and let $\CHT$ be the
  unique $\hat H$-tangle of order $4$. Then $\hat H$ is also a
  non-exceptional extension of $H$ in $G$, and by
  Lemma~\ref{lem:4t20}(4) and the transitivity of
  the lifting relation, $\CT$ is the lifting of
  $\CHT$ to $G$ with respect to some faithful model of $\hat H$ in $G$.
\end{proof}

We let
\[
R_{\CT}:=R^{(m)}.
\]

\begin{proof}[Proof of the Correspondence Theorem~\ref{theo:corr}.]
  The theorem follows from Theorem~\ref{theo:q4r} and
  Lemmas~\ref{lem:4t27} and \ref{lem:4t29}.
\end{proof}

So far, we have only talked about the quasi-4-connected regions of a
graph. It is natural to call the graphs $\torso G{R_{\CT}}$ for the
$G$-tangles $\CT$ of order $4$ the \emph{quasi-4-connected components}
of $G$. While the regions $R_{\CT}$ are not canonical, the following
corollary (to 
Lemma~\ref{lem:4t20}) says that the quasi-4-connected components
$\torso G{R_{\CT}}$ are canonical if viewed as abstract graphs (that
is, up to isomorphism). 

\begin{cor}\label{cor:q4c-can}
  Let $G,G'$ be a 3-connected graphs, and let $\CT,\CT'$ be tangles of
  order $4$ of the these graphs. Suppose that there is an isomorphism
  $f$ from $G$ to $G'$ that maps $\CT$ to $\CT'$, that is, such that
  for all $(Y,S,Z)\in\Sep_{<4}(G)$ we have
  $(Y,S,Z)\in\CT\iff(f(Y),f(S),f(Z))\in\CT'$. Then there is an
  isomorphism from $\torso G{R_{\CT}}$ to $\torso {G'}{R_{\CT'}}$.
\end{cor}

\section{Decomposition into Quasi-4-Connected Components}

With the Correspondence Theorem at hand, it is now relatively easy to
prove the Decomposition Theorem~\ref{theo:dec}.
 
\begin{theo}\label{theo:q4dec}
  Let $G$ be a 3-connected graph. Then $G$ has a tree decomposition
  $(T,\beta)$ of adhesion at most $3$ such that for all $t\in V(T)$,
  the torso $\torso G{\beta(t)}$ is either is a complete graph $K_3$ or $K_4$ or
 a quasi-4-connected component of $G$.

  Furthermore, such a decomposition can be computed in time $O(n^2(n+m))$.
\end{theo}

Here, and throughout this section, we denote the numbers of vertices and edges of the 
 input graph $G$ of our algorithms by $n$ and $m$, respectively.

The Decomposition Theorem~\ref{theo:dec} follows by combining the
decomposition of Theorem~\ref{theo:q4dec} with the standard decomposition of a
graph into its three connected components.

The proof of Theorem~\ref{theo:q4dec} requires some preparation.For
the rest of this section, we assume that $G$ is a 3-connected
graph. 
Let $(Y,S,Z)\in\Sep_{=3}(G)$ be non-degenerate. A
\emph{split vertex} of $(Y,S,Z)$ is a vertex $z\in Z$ such that for every
connected component $C$ of $G\setminus(S\cup\{z\})$ it holds that
$|N(C)|=3$. 

\begin{lem}\label{lem:q4dec1}
  Let $(Y_0,S_0,Z_0) \in\Sep_{=3}(G)$ be a non-degenerate proper
  separation such that
  $Z_0$ is connected and $(Y_0,S_0,Z_0)$ has no split vertex. Then the
  set $\CT(Y_0,S_0,Z_0)$ of all separation $(Y,S,Z)\in\Sep_{<4}(G)$
  such that either $Z_0\subseteq Z$ or $|Z\cap S_0|>|Y\cap S_0|$ is a
  $G$-tangle of order $4$.
\end{lem}

\begin{proof}
  Let $\CT:=\CT(Y_0,S_0,Z_0)$. To see that $\CT$ satisfies
  \ref{li:t1}, let $(Y,S,Z)\in\Sep_{<4}(G)$. If $S\subseteq Y_0\cup S_0$, then the connected set $Z_0$ is
  either a subset of $Z$ or of $Y$, and thus either $(Y,S,Z)\in\CT$
  or $(Z,S,Y)\in\CT$. Suppose next that $|S\cap Z_0|=1$. Let $z$ be
  the unique vertex in $S\cap Z_0$. Then $z$ is not a split vertex of $(Y_0,S_0,Z_0)$,
  and hence there is a connected component $C$ of
  $G\setminus(S_0\cup\{z\})$ such that $N(C)=S_0\cup\{z\}$. Then
  $V(C)\subseteq Z_0$, because $z\in Z_0$, and thus $V(C)\cap
  S=\emptyset$. It follows that either $V(C)\subseteq Y$ or $V(C)\subseteq
  Z$. Without loss of generality we may assume that $V(C)\subseteq
  Z$. As $S_0\subseteq N(C)$, this implies $S_0\setminus S\subseteq
  Z$. As $S_0\setminus S\neq\emptyset$, it follows that
  $(Y,S,Z)\in\CT$. Finally, suppose that $|S\cap Z_0|\ge 2$. If
  $S\cap S_0=\emptyset$, then either $|Z\cap S_0|\ge 2$ or $|Y\cap
  S_0|\ge 2$, and thus either $(Y,S,Z)\in\CT$
  or $(Z,S,Y)\in\CT$. If $|S\cap S_0|=1$, then $S\cap Y_0=\emptyset$,
  and as $G$ is 3-connected and $Y_0\neq\emptyset$, the vertices in
  $S_0\setminus S$ belong to the same connected component of
  $G\setminus S_0$. Hence either both are in $Z$ or both are in $Y$,
  and again it follows that  either $(Y,S,Z)\in\CT$
  or $(Z,S,Y)\in\CT$. 

  Observe next $|V(G)|\ge 6$, because $|Y_0|\ge 1$ and $|S_0|=3$
  and $|Z_0|\ge 2$ (otherwise the unique vertex in $Z_0$ would be a
  split vertex).

  \begin{claim}
    For all $(Y,S,Z)\in\CT$ we have $|S\cup Z|\ge 4$.

    \proof
    It follows from the definition of $\CT$ that $Z\neq\emptyset$. If $Y=\emptyset$,
    then $|S\cup Z|=|V(G)|\ge 6$. Otherwise, $(Y,S,Z)$ is a proper
    separation and thus $|S|=3$, which implies $|S\cup Z|\ge
    4$.
    \uend
  \end{claim}

  The claim implies that $\CT$ satisfies \ref{li:t3}.

  To prove that $\CT$ satisfies \ref{li:t2}, let
  $(Y_i,S_i,Z_i)\in\CT$ for $i=1,2,3$. Suppose for contradiction $Z_1\cap
  Z_2\cap Z_3=\emptyset$ and that there is no edge that
  has an endvertex in each $Z_i$. Once more we recycle Claim~1 of Lemma~\ref{lem:q4c1}.

  \begin{claim}[resume]
        For distinct $i,j,k\in[3]$ and $x\in V(G)$, if $x\in Z_i\cap Z_j$
    then $x\in Y_k$.
    \uend
  \end{claim}

  \begin{cs}
    \case1
    There is an $i\in[3]$ such that $S_i\subseteq Y_0\cup S_0$. \\
    Without loss of generality, we may assume that $i=1$ and
    $(Y_1,S_1,Z_1)=(Y_0,S_0,Z_0)$. We may further assume that
    $S_i\not\subseteq Y_0\cup S_0$ for $i=2,3$. Then $|Z_i\cap
    S_0|>|Y_i\cap S_0|$. 

    By Claim~1 we have
    $Z_2\cap Z_3\cap S_0 =Z_2\cap Z_3\cap S_1=\emptyset$. Thus for
    some $i\in\{2,3\}$ $|Z_i\cap S_0|<2$. Without loss of generality
    we assume $|Z_2\cap S_0|<2$. Then $|Y_2\cap S_0|=\emptyset$ and
    thus $|S_2\cap S_0|=2$. Since $S_2\not\subseteq Y_0\cup S_0$, we
    have $|S_2\cap Z_0|=1$. As the vertex in $S_2\cap Z_0$ is not a
    split vertex, there is a connected component $C$ of
    $G\setminus (S_0\cup S_2)$ such that $N(C)=S_0\cup S_2$. Then
    $V(C)\subseteq Z_0\cap Z_2=Z_1\cap Z_2$. Now let $v\in Z_3\cap S_0$, and let
    $w\in V(C)$ be adjacent to $v$. Then the edge $vw$ has an
    endvertex in each $Z_i$.

    \case2
    $|S_i\cap Z_0|\neq\emptyset$ For all $i\in[3]$.\\
    Then $|Z_i\cap S_0|>|Y_i\cap S_0|$. If $|Z_i\cap Z_j\cap
    S_0|=\emptyset$ for all $i\neq j$, then $|Z_i\cap S_0|=1$ and thus
    $|Y_i\cap S_0|=0$ for all $i$. Thus $|S_i\cap S_0|=2$ and
    $|S_i\cap Y_0|=\emptyset$, because $S_i\not\subseteq S_0\cup
    Y_0$. But this implies $Y_0\subseteq Z_1\cap Z_2\cap Z_3$, which
    is a contradiction.

    Hence without loss of generality we may assume that $Z_1\cap
    Z_2\cap S_0\neq\emptyset$. Let $s\in Z_1\cap
    Z_2\cap S_0$. Then by Claim~1, $s\in Y_3$. Then $|Y_3\cap
    S_0|\ge 1$, and this implies $|Z_3\cap S_0|\ge 2$. Let $s',s''\in
    Z_3\cap S_0$. Then $S_0=\{s,s',s''\}$. 

    If $|S_3\cap Z_0|\le 1$, there is a connected component $C$ of 
    $G\setminus (S_0\cup S_3)$ such that $N(C)=S_0\cup S_3$. But then
    there is a path from $s\in Y_3$ to $s'\in Z_3$ in $G\setminus
    S_3$, which is impossible. Hence $|S_3\cap Z_0|\ge 2$.

    Thus $|S_3\cap Y_0|\le 1$. Since the graph $G[Y_0\cup S_0]$ is
    connected, we have $|Y_0\cap S_3|= 1$, and the unique vertex
    $y\in Y_0\cap S_3$ separates $s$ from $\{s',s''\}$. Then
    $ss',ss''\not\in E(G)$. Furthermore, $sy\in E(G)$ and $y$ is the
    only neighbour of $s$ in $Y_0\cup S_0$, because otherwise $\{y,s\}$ would
    be separator of $G$. By Claim~1, $y\not\in Z_1\cap Z_2$. Say,
    $y\not\in Z_2$. then $y\in S_2$, because $y$ is adjacent to
    $s\in Z_2$. As $S_2\not\subseteq Y_0\cup S_0$, it now follows that
    $s'$ and $s''$ are not both in $S_2$. As
    $|Z_2\cap S_0|>|Y_2\cap S_2|$, one of these vertices, say, $s'$ is
    in $Z_2$.

    By Claim~1, $s'\in Z_2\cap Z_3$ implies $s'\in
    Y_1$. Arguing as above with $(Y_1,S_1,Z_1)$ instead of
    $(Y_3,S_3,Z_3)$, we see that $Z_1\cap S_0=\{s,s''\}$ and $|S_1\cap
    Z_0|=2$ and $|S_1\cap Y_0|=1$, and the unique vertex $y'\in
    S_1\cap Y_0$ separates $s'$ from $s,s''$ in $G$. Furthermore,
    $s's,s's''\not\in E(G)$, and $s'y'\in E(G)$ and $y'$ is the
    only neighbour of $s'$ in $Y_0\cup S_0$.

    Now we have $s''\in Z_1\cap Z_3$, and again by the same argument
    we see that  $s''\in Y_2$ and $Z_2\cap S_0=\{s,s'\}$ and $|S_2\cap
    Z_0|=2$ and $|S_2\cap Y_0|=1$ and the unique vertex $y''\in
    S_1\cap Y_0$ separates $s'$ from $s,s''$ in $G$. Furthermore,
    $s''s,s''s'\not\in E(G)$, and $s''y''\in E(G)$ and $y''$ is the
    only neighbour of $s''$ in $Y_0\cup S_0$.
    
    Let us rename the vertices $s,s',s''$ to $s_{12},s_{23},s_{13}$
    and the vertices $y,y',y''$ to $y_{12},y_{23},y_{13}$. Then for
    distinct $i,j,k$ we have $s_{ij}\in S_0\cap Z_i\cap Z_j\cap Y_k$ and $S_k\cap Y_0=\{y_{ij}\}$ and
    $N(s_{ij})\cap (Y_0\cup S_0)=\{y_{ij}\}$. Note that this implies
    that $S_0=\{s_{12},s_{13},s_{23}\}$ is an independent set.
   It follows that
    \begin{equation}
      \label{eq:34}
      Y_0\setminus\{y_{ij}\}\subseteq Z_k,
    \end{equation}
    because all $y\in Y_0\setminus\{y_{ij}\}$ are
    reachable in $G\setminus\{y_{ij}\}$ by a path from
    $\{s_{ik},s_{jk}\}\subseteq Z_k$.

    As the separation $(Y_0,S_0,Z_0)$ is non-degenerate and $S_0$ is
    an independent set, we have $|Y_0|>1$. Since
    $N(S_0)=\{y_{12},y_{23},y_{13}\}$ and $N(y_{ij})\cap
    S_0=\{s_{ij}\}$ and $G$ is 3-connected, it is
    easy to see that this implies that the vertices $y_{ij}$ are
    mutually distinct. Now let $e=vw$ be an arbitrary edge of
    $G[Y_0]$. Such an edge exists, and it follows from \eqref{eq:34}
    that the edge has an endvertex in each $Z_k$. Again, this is a
    contradiction.
    \qedhere
  \end{cs}
\end{proof}

Let $W,X\subseteq V(G)$. Then a \emph{$(W,X)$-separation} is a vertex
separation $(Y,S,Z)$ such that $W\subseteq Y\cup S$ and $X\subseteq
Z\cup S$. A $(W,X)$-separation $(Y,S,Z)$ is \emph{minimum} if its
order is minimal, that is, there is no $(W,X)$-separation $(Y',S',Z')$
such that $|S'|<|S|$. It is \emph{leftmost minimum} if it is minimum
and, subject to this condition, $Y$ is inclusionwise minimal.

The following lemma follows from Lemma~2.4 of \cite{marosuraz13}.

\begin{lem}[\cite{marosuraz13}]\label{lem:minsep}
  There is an algorithm that, given a graph $G$ and sets $W,X\subseteq
  V(G)$, computes a leftmost minimum $(W,X)$-separation in time
  $O(k(n+m))$, where $n:=|G|$, $m:=|E(G)|$, and $k$ is the order of a
  minimum $(X,Y)$-separation.
\end{lem}

A $(W,X)$-separation $(Y,S,Z)$ is \emph{proper} if $W\cap
Y\neq\emptyset$ and $X\cap Z\neq\emptyset$. Note that there is a
proper $(W,X)$-separation if and only if there is a $w\in W\setminus X$ and an
$x\in X\setminus W$ such that $wx\not\in E(G)$. 
A \emph{(leftmost) minimum
  proper $(W,X)$-separation} is a proper $(W,X)$-separation that is
(leftmost) minimum among all proper $(W,X)$-separations. 
While there always is a unique leftmost minimum
$(W,X)$-separation, as can be a proved by a straightforward
submodularity argument, there is not necessarily a unique leftmost
minimum proper $(W,X)$-separation. However, the proof of the following
lemma shows that there are at most $k^2$ leftmost
minimum proper $(W,X)$-separations, where $k$ is the order of a
leftmost minimum $(W,X)$-separation.

\begin{lem}\label{lem:q4dec2}
  Let $k\ge 1$. Then there is a linear time algorithm that, given a
  graph $G$ and sets $W,X\subseteq
  V(G)$, decides if there is a proper $(W,X)$-separation of order at
  most $k$, and if there is computes the set of all leftmost minimum
  proper $(W,X)$-separations. 
\end{lem}

Note that we treat $k$ as a constant here. In fact, we will only apply
the lemma for $k=3$.

\begin{proof}
  Let $G$ be a graph and $W,X\subseteq V(G)$. Let us first assume that
  $|W|,|X|\le k$. For a vertex $v\in V(G)$, we let $G_{v}$ be
  the graph obtained from $G$ by replacing $v$ by fresh vertices
  $v_1,\ldots,v_{k+1}$ and adding edges from $v_i$ to $v_j$ for all $i\neq
  j$ and from $v_i$ to $w$ for all $i$ and all $w\in N^G(v)$. 

  Now let $w\in W\setminus X$ and $x\in X\setminus W$ such that $w\neq x$, and consider the
  graph $G_{w,x}:=(G_w)_x$. Let
  $W_w:=(W\setminus\{w\})\cup\{w_1,\ldots,w_{k+1}\}$ and
  $X_x:=(X\setminus\{x\})\cup\{x_1,\ldots,x_{k+1}\}$. Observe that if
  $(Y,S,Z)$ is a
  minimum $(W_w,X_x)$-separation in $G_{w,x}$ of order $|S|\le k$,
  then $\{w_1,\ldots,w_{k+1}\}\subseteq Y$ and
  $\{x_1,\ldots,x_{k+1}\}\subseteq Z$. Thus $(Y,S,Z)$ ``projects'' to
  a proper $(W,X)$-separation 
  \[
  P(Y,S,Z):=\big((Y\setminus\{w_1,\ldots,w_{k+1}\})\cup\{w\},S,(Z\setminus\{x_1,\ldots,x_{k+1}\})\cup\{x\})
  \]
  of $G$. Moreover, if $(Y,S,Z)$ is leftmost minimum, then $P(Y,S,Z)$
  is leftmost minimum among all $(W,X)$-separations $(Y',S',Z')$ with
  $w\in Y'$ and $x\in Z'$.

  Now we let $\CP$ be the set of all $P(Y,S,Z)$, where $(Y,S,Z)$ is a
  leftmost minimum $(W_w,X_x)$-separation in $G_{w,x}$ for some $w\in
  W,x\in X$ with $w\neq x$. All separations in $\CP$ are proper
  $(W,X)$-separations, and provided there is a proper
  $(W,X)$-separation of order at most $k$, all leftmost minimum proper
  $(W,X)$-separations are in the set $\CP$. In fact, the leftmost minimum proper
  $(W,X)$-separations are precisely the $(Y,S,Z)\in\CP$ with minimum $|S|$
  and, subject to this, inclusionwise minimal $Y$.

  By Lemma~\ref{lem:minsep} and the
  assumption $|W|,|X|\le k$, the set $\CP$ can be computed in linear
  time, and then we can filter out those separations that are actually
  leftmost minimum.

  It remains to deal with the case that $|W|>k$ or $|X|>k$. If both
  $|W|>k$ and $|X|>k$, every $(W,X)$-separation of order at most $k$
  is proper. Thus the assertion of the lemma follows directly from
  Lemma~\ref{lem:minsep}. If $|W|\le k$ and $|X|>k$, we consider
  $(W_w,X)$-separations in the graph $G_w$, for all $w\in W$, and if
  $|W|> k$ and $|X|\le k$, we consider $(W,X_x)$-separations in the
  graph $G_x$, for all $x\in X$.
\end{proof}

Let us say that a separation $(Y_0,S_0,Z_0)\in\Sep_{=3}(G)$
\emph{defines} a tangle if $(Y_0,S_0,Z_0)$ is non-degenerate and $Z_0$
is connected in $G$ and $(Y_0,S_0,Z_0)$ has no split vertex. Then the
tangle \emph{defined by} $(Y_0,S_0,Z_0)$ is $\CT(Y_0,S_0,Z_0)$ (of
Lemma~\ref{lem:q4dec1}).

\begin{lem}\label{lem:q4dec3}
  There is an algorithm that, given a 3-connected graph $G$ and
  a separation $(Y_0,S_0,Z_0)$ of $G$ of order $3$ defining the
  tangle $\CT=\CT(Y_0,S_0,Z_0)$, computes the set $\CNDT$ and the set
  of all non-degenerate crossedges of $\CT$ in time
  $O(n(n+m))$.
\end{lem}

\begin{proof}
  We show how to compute the set $\CMT$; then we can easily filter
  out the non-degenerate separations in $\CMT$ to obtain $\CNDT$.

  Let $x\in Z_0$. Observe that if $(Z,S,Y)$ is a proper
  $(S_0,\{x\})$-separation of order at most $3$, then
  $(Y,S,Z)\in\CT$. This follows immediately from the definition of
  $\CT$. It implies the following equivalence for every separation
  $(Y,S,Z)$ of $G$ of order $3$.
  \begin{eroman}
  \item $(Y,S,Z)\in\CMT$ and $(Y,S,Z)$ does not cross
    $(Y_0,S_0,Z_0)$.
  \item There is an $x\in Z_0$ such that $(Z,S,Y)$ is a leftmost
    minimum proper $(S_0,\{x\})$-separation.
  \end{eroman}
  We can use this equivalence to compute the set of all
  $(Y,S,Z)\in\CMT$ such that $(Y,S,Z)$ does not cross $(Y_0,S_0,Z_0)$
  (repeatedly applying the algorithm of Lemma~\ref{lem:q4dec2} to all
  $x\in Z_0$). Note that the equivalence also gives us a linear bound on the number of
  such $(Y,S,Z)$.

  It remains to deal with the $(Y,S,Z)\in\CMT$ crossing
  $(Y_0,S_0,Z_0)$. For each $s\in S_0$ that has a unique neighbour
  $y\in Y_0\cup S_0$, the edge $sy$ may be a crossedge. This gives us at most
  three potential crossedges, and we deal with them separately. So let
  $s\in S$ and $y\in Y_0$ such that $N(s)\cap (Y_0\cap
  S_0)=\{y\}$. Then for every separation $(Y,S,Z)\in\Sep_{=3}(G)$
  the following are equivalent.
  \begin{eroman}[resume]
  \item $y\in S$ and $\big(Z\cap(S_0\cup Z_0),S\cap(S_0\cup Z_0),Y\cap (S_0\cup
    Z_0)\big)$ is a leftmost minimum proper
    $(S\setminus\{s\},\{s\})$-separation in the graph $G[S_0\cup
    Z_0]$.
  \item $(Y,S,Z)\in\CMT$ and $(Y,S,Z)$ crosses $(Y_0,S_0,Z_0)$ with
    crossedge $ys$.
  \end{eroman}
  To see this, note that (iii) implies that $|S\cap Z_0|=2$, because
  $(Y_0,S_0,Z_0)$ has no split vertex. The equivalence between (iii)
  and (iv) allows us to compute the remaining separations in $\CMT$.

  As we have an overall linear bound on the number of separations in
  $\CMT$ and hence in $\CNDT$, we can easily compute the set of
  non-degenerate crossedges.
\end{proof}
  
Let us a call a 3-separator $S$ of $G$ \emph{degenerate} if there is a
connected component $C$ of $G\setminus S$ such that the separation
$(G\setminus(S\cup V(C)),S,V(C))$ is degenerate. It is easy to see
that this is the case if and only if $S$ is an independent set and
$G\setminus S$ has exactly two connected components, one of which has
order $1$.

\begin{lem}\label{lem:compute-nd}
  There is an a algorithm that, given a 3-connected graph $G$,
  decides if $G$ has a non-degenerate 3-separator and computes one if
  there is in time $O(n^2(n+m))$.
\end{lem}
  
\begin{proof}
  We first test if there is an $S\subseteq V(G)$ such that $|S|=3$ and
  all connected components of $G\setminus S$ have order $1$. In this
  case, $S$ is a non-degenerate 3-separator if $|G|\ge 6$ or if $|G|\ge 5$ and $S$ is
  not an independent set.

  In the following, we assume that for every $S\subseteq V(G)$ such
  that $|S|=3$ there is at least one connected components $C$ of
  $G\setminus S$ such that $|C|\ge 2$. Now suppose that $S$ is a
  non-degenerate 3-separator of $G$. Let $Y$ be the vertex set of a
  connected component of $G$ of size $|Y|\ge 2$, and let
  $Z:=V(G)\setminus(S\cup Y)$. Let $y\in Y$ and $z\in Z$.

  Then there is a leftmost minimum proper
  $(\{y\},\{z\})$-separation $(Y',S',Z')$ with $Y'\cup S'\subseteq
  Y\cup S$, because $(Y,S,Z)$ is a minimum proper
  $(\{y\},\{z\})$-separation. The separator $S'$ is non-degenerate
  unless $S'$ is an independent set and
  $S'=N(y)$. However, in this case there is a leftmost minimum proper
  $(S',\{z\})$ separation $(Y'',S'',Z'')$ such that $S''$ is
  non-degenerate. To see this, let $y'\in N(y)\cap Y$. Then there is a
  leftmost minimum proper
  $(S',\{z\})$ separation $(Y'',S'',Z'')$ with $y,y'\in Y''$ and
  $Y''\cup S''\subseteq Y\cup S$, because $(Y,S,Z)$ is a minimum proper
  $(S',\{z\})$ separation with $y,y'\in Y$. The set $S''$ is a
  non-degenerate 3-separator.
  
  Thus we can find a non-degenerate 3-separator as follows. For all
  pairs $y,z$ of distinct vertices, we compute all leftmost minimum proper
  $(\{y\},\{z\})$-separations $(Y',S',Z')$ and check if there is one
  such that $S'$ is a non-degenerate 3-separator. If $y$ has degree
  $3$ and $S':=N(y)$ is an independent set, we also compute all leftmost minimum proper
  $(S',\{z\})$ separations $(Y'',S'',Z'')$ and check if $S''$ is a non-degenerate 3-separator.
\end{proof}
  
\begin{proof}[Proof of Theorem~\ref{theo:q4dec}]
  If $G$ has no non-degenerate 3-separator, then $G$ is quasi-4-con\-nected,
  and we return the trivial tree decomposition with a  one-node
  tree. In the following, we assume that $G$ has at least one
  non-degenerate 3-separator.

  We view the tree $T$ in the tree decomposition as directed with all
  edges pointing away from the root, and we denote the descendant order
  in the tree by $\unlhd^T$, or just $\unlhd$ if $T$ is clear from the
  context. With each (directed) edge $e=(s,t)$ of the tree we associate
  a separation $\sep(s,t)=(Y,S,Z)$ of order $3$ such that $Z$
  is connected in $G$ and
  \begin{align}
    \label{eq:35}
    S&=\beta(t)\cap\beta(s),\\
    \label{eq:36}
    S\cup Z&=\bigcup_{u\unrhd t}\beta(u).
  \end{align}
  We build the tree decomposition iteratively starting from the root
  $r$ of the tree. We pick an arbitrary non-degenerate 3-separator
  $S_r$ of $G$ and let $\beta(r):=S_r$. For every connected component $C$ of $G\setminus S_r$ we
  create a child $t$ of $r$, and we let
  $\sep(r,t):=(V(G)\setminus(S_r\cup V(C)),S_r,V(C))$. 

  At every step of the construction, we pick a leaf $t$ of the current
  tree such that $\beta(t)$ is not yet defined. Let $s$ be the parent
  of $t$ and $\sep(s,t)=(Y_0,S_0,Z_0)$.
  \begin{cs}
    \case1
    $|Z_0|\le 1$.\\
    Then $|S_0\cup Z_0|\le 4$, and we let $\beta(t):=S_0\cup Z_0$. The
    node $t$ will remain a leaf of the final tree.
    \case2
    $|Z_0|>1$ and $(Y_0,S_0,Z_0)$ has a split vertex $z_0\in
    Z_0$.\\
    Then we let $\beta(t):=S_0\cup\{z_0\}$. For every connected
    component $C$ of $G\setminus (S_0\cup\{z_0\})$ with $V(C)\subseteq
    Z_0$ we create a child $u$ of $t$ and let
    $\sep(t,u):=(V(G)\setminus(N(C)\cup V(C)),N(C),V(C))$.
    \case3
    $|Z_0|>1$ and $(Y_0,S_0,Z_0)$ has no
    split vertex.\\
    Let $\CT=\CT(Y_0,S_0,Z_0)$ be the $G$-tangle of Lemma~\ref{lem:q4dec1}. Note that
    $(Y_0,S_0,Z_0)\in\CNDT$. We associate a quasi-4-connected region
    $R_{\CT}$ with $\CT$ as described in
    Section~\ref{sec:tangleregion}, where we make sure that we
    contract all crossedges of crossings that involve $S_0$ to their
    endvertex in $S_0$. Then, in the terminology of
    Section~\ref{sec:tangleregion}, $S_0=S_0^{\slashes m}$.

    We let $\beta(t):=R_{\CT}$. 

    Recall Corollary~\ref{cor:4t25b}. For every
    $S\in\CNDS\setminus\{S_0\}$ and every connected component $C$ of
    $G\setminus R_{\CT}$ with $N(C)=S^{\slashes m}$ we create a  child $u$ of $t$ and let
    $\sep(t,u):=(V(G)\setminus(S^{\slashes m}\cup V(C)),S^{\slashes m},V(C))$.
  \end{cs}
  The completes the description of our construction. We need to
  describe a time $O(n^2(n+m))$-algorithm implementing it. By
  Lemma~\ref{lem:compute-nd}, we can compute a non-degenerate
  3-separator $S_r$ (for the root $r$) within this time if there is
  one.

  Now we show that we can handle every step of the construction in
  time $O(n(n+m))$. So let $t$ be a leaf of the current tree, $s$ its
  parent, and $(Y_0,S_0,Z_0):=\sep(s,t)$. Case~1 is easy. For Case~2,
  we need to compute all connected components of
  $G\setminus(S_0\cup\{z\})$ for all $z\in Z_0$, which we can do in
  time $O(n(n+m))$. For Case~3, we need to compute $\CMT$ and
  $R_{\CT}$ for the tangle $\CT=\CT(Y_0,S_0,Z_0)$, and
  Lemma~\ref{lem:q4dec3} allows us to do this.
\end{proof}

\begin{rem}
  There is one minor issue that we ignored so far in order to not
  complicate things unnecessarily. It may happen that some
  quasi-4-connected components $\torso G{R_{\CT}}$ of $G$ do not appear as
  torsos $\torso G{\beta(t)}$ in the decomposition, because if $\torso
  G{R_{\CT}}$ is a subgraph of $\THF$ (see Figure~\ref{fig:th4}), it
  will be treated in Case~2 of the construction: the vertex
  $v_4$ is a split vertex with respect to the separation
  $\{v_1,v_2,v_3\}$, and the vertices $w_i$ will be split off. But we
  can easily detect this during the construction and avoid to split off
  those vertices if we want to.

  If we carry out the construction exactly as in the proof of the
  theorem, then the $G$-tangles of order $4$ are associated with all
  nodes $t$ such that 
  \begin{itemize}
  \item either $|\beta(t)|\ge 5$
  \item or $|\beta(t)|=4$ and for each subset $S\subseteq \beta(t)$ of
    size $|S|=3$ there is a neighbour $u$ of $t$ such that
    $\beta(u)\cap \beta(t)=S$.
  \end{itemize}
  In the second case, the neighbours of $t$ allow us to find a
  non-exceptional extension of the quasi-4-connected region
  $\beta(t)$.
  \uend
\end{rem}

\section{Conclusions}
Relaxing 4-connectedness, we introduce the notion of
quasi-4-connectedness of graphs and prove that
every graph has a decomposition into quasi-4-connected
components. We show that these quasi-4-connected component correspond
to the tangles of order $4$, putting our result in the context of
recent work on tangles an decompositions
\cite{cardiehar+13a,cardiehar+13b,cardiehun+14,groschwe15a,hun11,gm10}. Furthermore,
we prove that our decomposition can be computed in cubic time. I think that our decomposition generalises the
decomposition of a graph into its 3-connected components in a natural
way and as such is a fundamental and interesting result
in structural graph theory. Although we did not explore this in
the present paper, I also believe that the result may turn out to be a
useful algorithmic tool, just like the decomposition into 3-connected
components (though maybe not quite as broadly applicable). 

The most obvious question is whether our result has a generalisation
to ``quasi-$k$-connected components'', whatever they may be, for
$k\ge 5$. I am skeptical, because we exploit many special
properties of separators of order $3$ here, most importantly the
limited way in which they can cross. However, our decomposition is
not a straightforward generalisation of the decomposition into
3-connected components either, and arguably, our main contributions are the
conceptual ideas related to quasi-4-connectedness. It may well be that
new conceptual ideas lead to perfectly nice decompositions of higher
order.

Finally, in particularly when thinking of applications, it would be
desirable to have a decomposition algorithm working in quadratic or
even in linear time. I see no fundamental obstructions to the
existence of such
an algorithm.

\newpage
\appendix
\part*{Appendix}

In this appendix, we discuss how our notion of tangles based on
separations of the vertex set relates to Robertson and Seymour's
original definition of tangles \cite{gm10} based on partitions of the
edge set. Let us say that an \emph{RS-separation} of a graph $G$ is a pair $(A,B)$ of subgraphs of $G$ such
that $A\cup B=G$ and $E(A)\cap E(B)=\emptyset$. The \emph{order} of
the RS-separation $(A,B)$ is $\ord(A,B):=|V(A)\cap V(B)|$.

Robertson and Seymour define a $G$-tangle of order $k$ to be a family $\CT$ of RS-separations of
$G$ of order less than $k$ satisfying the following conditions.
\begin{nlist}{T'}
\item\label{li:rt1}
  For all RS-separations $(A,B)$ of $G$ of order less than $k$, either
  $(A,B)\in\CT$ or $(B,A)\in\CT$.
\item\label{li:rt2}
  If $(A_1,B_1), (A_2,B_2),(A_3,B_3)\in\CT$ then $A_1\cup A_2\cup
  A_3\neq G$.
\item\label{li:rt3}
  $V(A)\neq V(G)$ for all $(A,B)\in\CT$.
\end{nlist}
Let us call such tangles \emph{RS-tangles} in the following.

With every RS-separation $(A,B)$ we associate the separation
$\angles{A,B}=(Y,S,Z)$ with $Y:=V(A)\setminus V(B)$, $S:=V(A)\cap
V(B)$, and $Z:=V(B)\setminus V(A)$. %

\begin{prop}
  Let $G$ be a graph.
  \begin{enumerate}
  \item Let $\CT$ be a $G$-tangle of order $k$.
    Then 
    \[
    \CT':=\big\{(A,B)\bigmid (A,B)\text{ RS-separation of $G$ of order $<k$
      with }\angles{A,B}\in\CT\big\}
    \]
    is an RS-tangle of $G$ of order $k$.
  \item Let $\CT'$ be a an RS-tangle of $G$ of order $k$. Then
     \[
    \CT:=\big\{\angles{A,B}\bigmid (A,B)\in\CT'\big\}
    \]
   is a $G$-tangle of order $k$.
  \end{enumerate}
\end{prop}

We omit the straightforward proof. the proposition shows that our
version of tangle and Robertson and Seymour's original version are
essentially the same, and it also shows how to translate our results
to the original framework

\end{document}